\documentclass{article}[12]
\usepackage{amsmath}
\usepackage{amsfonts}
\usepackage{amssymb}
\usepackage{amsthm}
\usepackage{ifpdf}
\usepackage{wrapfig}
\usepackage{fullpage}
\usepackage{cite}
\usepackage{subfigure}

\ifpdf

  \usepackage[pdftex]{epsfig}
  \usepackage[pdftex]{hyperref}

\else

    \usepackage[dvips]{epsfig}
    \newcommand{\href}[2]{#2}

\fi

\theoremstyle{definition}
\newtheorem{theorem}{Theorem}[section]
\newtheorem{lemma}[theorem]{Lemma}

\begin{document}

\begin{titlepage}
\title{Fast Arithmetic in Algorithmic Self-Assembly}
\author{
Alexandra Keenan\thanks{Department of Computer Science, University of Texas - Pan American,
      \protect\url{{abkeenan, rtschweller, mjsherman, zhongxingsi}@utpa.edu} This author's research was supported in part by National Science Foundation Grant CCF-1117672.}
\and
Robert Schweller\footnotemark[1]
%\thanks{Department of Computer Science, University of Texas - Pan American,
%      \protect\url{rtschweller@utpa.edu} This author's research was supported in part by National Science Foundation Grant CCF-1117672.}
\and
Michael Sherman\footnotemark[1]
%\thanks{Department of Computer Science, University of Texas - Pan American,
%      \protect\url{mjsherman@utpa.edu} This author's research was supported in part by National Science Foundation Grant CCF-1117672.}
\and
Xingsi Zhong\footnotemark[1]
%\thanks{Department of Computer Science, University of Texas - Pan American,
%  \protect\url{zhongxingsi@gmail.com} This author's research was supported in part by National Science Foundation Grant CCF-1117672.}
}
\date{}

\maketitle
\thispagestyle{empty}

\begin{abstract}
In this paper we consider the time complexity of adding two $n$-bit numbers together within the tile self-assembly model.  The (abstract) tile assembly model is a mathematical model of self-assembly in which system components are square tiles with different glue types assigned to tile edges.  Assembly is driven by the attachment of singleton tiles to a growing seed assembly when the net force of glue attraction for a tile exceeds some fixed threshold.  Within this frame work, we examine the time complexity of computing the sum or product of 2 $n$-bit numbers, where the input numbers are encoded in an initial seed assembly, and the output is encoded in the final, terminal assembly of the system.  We show that the problems of addition and multiplication have worst case lower bounds of $\Omega( \sqrt{n} )$ in 2D assembly, and $\Omega(\sqrt[3]{n})$ in 3D assembly.  In the case of addition, we design algorithms for both 2D and 3D that meet this bound with worst case run times of $O(\sqrt{n})$ and $O(\sqrt[3]{n})$ respectively, which beats the previous best known upper bound of $O(n)$. Further, we consider average case complexity of addition over uniformly distributed $n$-bit strings and show how to achieve $O(\log n)$ average case time with a simultaneous $O(\sqrt{n})$ worst case run time in 2D.  For multiplication, we present an $O(n^{5/6})$ time multiplication algorithm which works in 3D, which beats the previous best known upper bound of $O(n)$.  As additional evidence for the speed of our algorithms, we implement our addition algorithms, along with the simpler $O(n)$ time addition algorithm, into a probabilistic run-time simulator and compare the timing results.

%Previous work has shown how to add 2 $n$-bit numbers in order $O(n)$ time in the tile self-assembly model.  This same work also shows that this run time meets a corresponding lower bound for the addition problem.  However, this lower bound argument only holds if the input is formatted in a particular way.  In this paper we propose a general framework for computation input and output formats in the tile self-assembly model and show that with a more exotic input format, the addition of 2 $n$-bit integers can be performed in $O(\sqrt{n})$ time.  We further show that this meets a lower bound of $\Omega(\sqrt{n})$ that holds regardless of how the input is formatted.  Also, we show that there exists a construction that performs the addition of 2 $n$-bit integers in $O(\log{n})$ average case time with $O(n)$ worst case time.  Finally, we combine the two methods into one that has the best of both, a $O(\log{n})$ average case time and a $O(\sqrt{n})$ worst-case time.
\end{abstract}
\end{titlepage}

% ----------------------------------------------------------------
\pagestyle{plain}
%\newpage
%\setcounter{page}{1}

% ----------------------------------------------------------------
% ----- Comment the notes section out for final draft.
%\input{notes}

\section{Introduction.}

\begin{table}[t] \label{tab:results}\caption{Summary of results.}
	\centering
	\begin{tabular}{|l||c c|c|}
		\hline
		~						& \multicolumn{2}{|c|}{Worst Case}				                         & Average Case				\\
& UB & LB & \\\hline\hline
%&&&\\
		Addition(2D)			& \multicolumn{2}{|c|}{$\Theta(\sqrt{n})$} 	 & $O(\log{n})$ 				\\
&(Thm.\ref{thm:combinedAddition}) & (Thm.\ref{thm:addition} ) & (Thm.\ref{thm:combinedAddition}) \\ \hline
%&&&\\
		Addition(3D)			& \multicolumn{2}{|c|}{$\Theta(\sqrt[3]{n})$} 	 & $O(\log{n})$ 				\\
& (Thm.\ref{thm:combined3d} ) & (Thm. \ref{thm:addition} ) & (Thm.\ref{thm:combined3d})\\ \hline
%&&&\\
        Multiplication (3D) & $O(n^{5/6})$ & $\Omega(n^{1/3})$ & - \\
        & (Thm.~\ref{thm:56multiplication}) & (Thm.~\ref{thm:multiplication}) & \\ \hline
		Previous Best Addition(2D)		& $O(n)$ (See \cite{BRUN2007})	& \textbf{-} &	\textbf{-}			\\
		Previous Best Multiplication(2D)	& $O(n)$ (See \cite{BRUN2007}) 	& \textbf{-}	&	\textbf{-}		\\
		\hline
	\end{tabular}
\end{table}

%\begin{table}[t] \label{tab:results}\caption{Summary of results.}
%	\centering
%	\begin{tabular}{|l|c|c|}
%		\hline
%		~						& Time Complexity				                         & Average Case				\\ \hline
%		Addition(2D)			& $\Theta(\sqrt{n})$ (Thm. \ref{thm:addition}, \ref{thm:combinedAddition})		 & $O(\log{n})$ (Thm.\ref{thm:combinedAddition})				\\
%		Addition(3D)			& $\Theta(\sqrt[3]{n})$ (Thm. \ref{thm:addition}, \ref{thm:combined3d} )	 & $O(\log{n})$ (Thm.\ref{thm:combined3d})				\\
%		Previous Best Addition(2D)		& $O(n)$~\cite{BRUN2007}	& \textbf{-}				\\
%		Multiplication(d-D)	& $\Omega(\sqrt[d]{n})$ (Thm. \ref{thm:multiplication})	& \textbf{-}				\\
%		\hline
%	\end{tabular}
%\end{table}

%Self-assembly is the process by which systems of simple objects organize themselves through local interactions into larger, more complex objects.  Self-assembly is a driving force in nature that is responsible for the assembly of the most interesting molecular machines yet observed.  Harnessing this power by developing technology for the efficient programming of self-assembling molecular systems is a line of research fundamental to the future of nanotechnology.  A number of research directions to this end have emerged in recent years.  One key direction is the development of molecular algorithms to solve fundamental computational problems.

Self-assembly is the process by which systems of simple objects autonomously organize themselves through local interactions into larger, more complex objects.  Self-assembly processes are abundant in nature and serve as the basis for biological growth and replication.  Understanding how to design and efficiently program molecular self-assembly systems promises to be fundamental for the future of nanotechnology.  One particular direction of interest is the design of molecular computing systems for the efficient solution of fundamental computational problems.  In this paper we study the complexity of computing arithmetic primitives within a well studied model of algorithmic self-assembly, the abstract tile assembly model.

The abstract tile assembly model (aTAM) models system monomers with four sided Wang tiles with glue types assigned to each edge.  Assembly proceeds by tiles attaching, one by one, to a growing initial seed assembly whenever the net glue strength of attachment exceeds some fixed temperature threshold.  The aTAM has been shown to be capable of universal computation~\cite{Winf98}, and research leveraging this computational power has lead to efficient assembly of complex geometric shapes and patterns with a number of recent results in FOCS, SODA, and ICALP%~\cite{RNaseSODA2010,ACGHKMR02,AGKS05g,BryChiDotKarSek10,ChenGoelLuhrs08,CookFuSch11,Dot10,USAreal,SFTSAFT,KS06,RotWin00,SolWin07,SS2013FEC}.
\cite{RNaseSODA2010, BryChiDotKarSek10,CookFuSch11, USAreal, SFTSAFT, SS2013FEC, fu2012SAGT, ChaGopRei09, Dot10, KaoSchS08,DPR2013ICALP}. This universality also allows the model to serve directly as a model for computation in which an input bit string is encoded into an initial assembly. The process of self-assembly and the final produced terminal assembly represent the computation of a function on the given input.  Given this framework, it is natural to ask how fast a given function can be computed in this model.  Tile assembly systems can be designed to take advantage of massive parallelism when multiple tiles attach at distinct positions in parallel, opening the possibility for faster algorithms than what can be achieved in more traditional computational models.  On the other hand, tile assembly algorithms must use up geometric space to perform computation, and must pay substantial time costs when communicating information between two physically distant bits.  This creates a host of challenges unique to this physically motivated computational model that warrant careful study.

In this paper we consider the time complexity of adding or multiplying two $n$-bit numbers within the abstract tile assembly model.  We show that addition and multiplication have worst-case lower bounds of $\Omega(\sqrt{n})$ time in 2D and $\Omega(\sqrt[3]{n})$ time in 3D.  These lower bounds are derived by a reduction from a simple problem we term the \emph{communication} problem in which two distant bits must compute the AND function between themselves.  This general reduction technique can likely be applied to a number of problems and yields key insights into how one might design a sub-linear time solution to such problems.  We in turn show that for the problem of addition these lower bounds are matched by corresponding worst case $O(\sqrt{n})$ and $O(\sqrt[3]{n})$ run time algorithms, respectively, which improves upon the previous best known result of $O(n)$~\cite{BRUN2007}.  We then consider the average case complexity of addition given two uniformly generated random $n$-bit numbers and construct a $O(\log n)$ average case time algorithm that achieves simultaneous worst case run time $O(\sqrt{n})$ in 2D.  To the best of our knowledge this is the first tile assembly algorithm proposed for efficient average case adding.  Finally, we present a 3D algorithm that achieves $O(n^{5/6})$ time for multiplication which beats the previous fastest algorithm of time $O(n)$~\cite{BRUN2007} (which works in 2D).  We analyze our algorithms under two established timing models described in Section~\ref{sec:runTimeModels}.

Our results are summarized in Table~\ref{tab:results}. In addition to our analytical results, tile self-assembly software simulations were conducted to visualize the diverse approaches to fast arithmetic presented in this paper, as well as to compare them to previous work. The adder tile constructions described in Sections \ref{sec:avgcase}, \ref{sec:worstcase} and \ref{sec:combinedcase}, and the previous best known algorithm from~\cite{BRUN2007} were simulated using the two timing models described in Section \ref{sec:runTimeModels}.  These results can be seen in the graphs in Section \ref{sec:sim}.

\section{Definitions}

\subsection{Basic Notation.}
Let $\mathbb{N}_n$ denote the set $\{ 1, \ldots , n\}$ and let $\mathbb{Z}_n$ denote the set $\{0, \ldots ,n-1\}$.  Consider two points $p,q \in \mathbb{Z}^d$, $p = (p_1, \ldots p_d)$, $q = (q_1,\ldots , q_d)$. Define $\Delta_{p,q} \triangleq \max_{1\leq i \leq d} \{  |p_i - q_i| \}$.

\subsection{Abstract Tile Assembly Model.}
\begin{figure}
	\centering
	\subfigure[Incorrect binding.]{
		~~
		\includegraphics[scale=1.0]{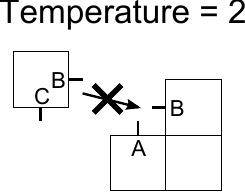}
		~~
	}
	\hspace{0.5 in}
	\subfigure[Correct binding.]{
		\includegraphics[scale=1.0]{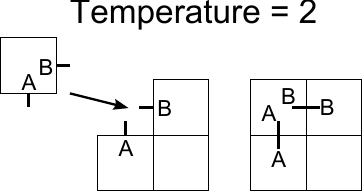}
	}
	\caption{Cooperative tile binding in the aTAM.}
	\label{fig:atam}
\end{figure}
%[RTS:  A simple figure to better elucidate the model might be good:  a picture of a tile with some glues, and a bit label.  Then, a picture of an assembly, maybe showing matched and mismatched glues.]

%The abstract Tile Assembly Model, aTAM, one starts from a predefined seed structure that can grow to produce finite or infinite structures representing the output of computations or predefined shapes. Blah...

\paragraph{Tiles.}
Consider some alphabet of symbols $\Pi$ called the \emph{glue types}.  A tile is a finite edge polygon (polyhedron in the case of a 3D generalization) with some finite subset of border points each assigned some glue type from $\Pi$.  Further, each glue type $g \in \Pi$ has some non-negative integer strength $str(g)$.  For each tile $t$ we also associate a finite string \emph{label} (typically ``0", or ``1", or the empty label in this paper), denoted by label($t$), which allows the classification of tiles by their labels.  In this paper we consider a special class of tiles that are unit squares (or unit cubes in 3D) of the same orientation with at most one glue type per face, with each glue being placed exactly in the center of the tile's face.  We denote the \emph{location} of a tile to be the point at the center of the square or cube tile.  In this paper we focus on tiles at integer locations.
%
%or unit cubes in $3D$ of the same orientation with at most one glue type per face, with each glue being placed exactly in the center of the tile's face.

\paragraph{Assemblies.}
An assembly is a finite set of tiles whose interiors do not overlap.  Further, to simplify formalization in this paper, we further require the center of each tile in an assembly to be an integer coordinate (or integer triplet in 3D).  If each tile in $A$ is a translation of some tile in a set of tiles $T$, we say that $A$ is an assembly over tile set $T$.   For a given assembly $\Upsilon$, define the \emph{bond graph} $G_\Upsilon$ to be the weighted graph in which each element of $\Upsilon$ is a vertex, and the weight of an edge between two tiles is the strength of the overlapping matching glue points between the two tiles.  Note that only overlapping glues that are the same type contribute a non-zero weight, whereas overlapping, non-equal glues always contribute zero weight to the bond graph.  The property that only equal glue types interact with each other is referred to as the \emph{diagonal glue function} property and is perhaps more feasible than more general glue functions for experimental implementation.  An assembly $\Upsilon$ is said to be \emph{$\tau$-stable} for an integer $\tau$ if the min-cut of $G_\Upsilon$ is at least $\tau$.

\paragraph{Tile Attachment.}
Given a tile $t$, an integer $\tau$, and a $\tau$-stable assembly $A$, we say that $t$ may attach to $A$ at temperature $\tau$ to form $A'$ if there exists a translation $t'$ of $t$ such that $A' = A \bigcup \{t'\}$, and $A'$ is $\tau$-stable. For a tile set $T$ we use notation $A \rightarrow_{T,\tau} A'$ to denote that there exists some $t\in T$ that may attach to $A$ to form $A'$ at temperature $\tau$.  When $T$ and $\tau$ are implied, we simply say $A \rightarrow A'$.  Further, we say that $A \rightarrow^*A'$ if there exists a finite sequence of assemblies $\langle A_1 \ldots A_k\rangle$ such that $A \rightarrow A_1 \rightarrow \ldots \rightarrow A_k \rightarrow A'$.

%\paragraph{aTAM System.}
%Given a set of tiles $T$ called a \emph{tile set}, a temperature $\tau$, and a starting assembly $S$ called the seed, an aTAM system assumes there is an infinite amount of each tile $t \in T$ and focuses only on the growth of the seed through tile attachments assuming no other reactions can occur in the system.

\paragraph{Tile Systems.}
A tile system $\Gamma = (T,S, \tau)$ is an ordered triplet consisting of a set of tiles $T$ referred to as the system's \emph{tile set}, a $\tau$-stable assembly $S$ referred to as the system's \emph{seed} assembly, and a positive integer $\tau$ referred to as the system's \emph{temperature}. A tile system $\Gamma = (T,S,\tau)$ has an associated set of \emph{producible} assemblies, $\texttt{PROD}_\Gamma$, which define what assemblies can grow from the initial seed $S$ by any sequence of temperature $\tau$ tile attachments from $T$.  Formally, $S \in \texttt{PROD}_\Gamma$ as a base case producible assembly.  Further, for any $A\in \texttt{PROD}_\Gamma$, if $A \rightarrow_{T,\tau} A'$, then $A' \in \texttt{PROD}_\Gamma$.  That is, assembly $S$ is producible, and for any producible assembly $A$, if $A$ can grow into $A'$, then $A'$ is also producible.  %In other words, the set of producible assemblies contains any assembly that can be obtained by any valid sequence of tile attachments starting from assembly $S$.
We further define the set of \emph{terminal} assemblies $\texttt{TERM}_\Gamma$ to be the subset of $\texttt{PROD}_\Gamma$ containing all producible assemblies that have no attachable tile from $T$ at temperature $\tau$.  Conceptually, $\texttt{TERM}_\Gamma$ represents the final collection of output assemblies that are built from $\Gamma$ given enough time for all assemblies to reach a final, terminal state.  General tile systems may have terminal assembly sets containing 0, finite, or infinitely many distinct assemblies.  Systems with exactly 1 terminal assembly are said to be \emph{deterministic}.  For a deterministic tile system $\Gamma$, we say $\Gamma$ \emph{uniquely assembles} assembly $A$ if $\texttt{TERM}_\Gamma = \{ A \}$.  In this paper, we focus exclusively on deterministic systems.  For recent consideration of non-determinism in tile self-assembly see~\cite{ChaGopRei09,BryChiDotKarSek10,CookFuSch11,KaoSchS08,Dot10}.

\subsection{Problem Description.}

We now formalize what we mean for a tile self-assembly system to compute a function.  To do this we present the concept of a \emph{tile assembly computer} (TAC) which consists of a tile set and temperature parameter, along with input and output \emph{templates}.  The input template serves as a seed structure with a sequence of \emph{wildcard positions} for which tiles of label ``0" and ``1" may be placed to construct an initial seed assembly.  An output template is a sequence of points denoting locations for which the TAC, when grown from a filled in template, will place tiles with ``0" and ``1" labels that denote the output bit string.  A TAC then is said to compute a function $f$ if for any seed assembly derived by plugging in a bitstring $b$, the terminal assembly of the system with tile set $T$ and temperature $\tau$ will be such that the value of $f(b)$ is encoded in the sequence of tiles placed according to the locations of the output template.  We now develop the formal definition of the TAC concept.  We note that the formality in the input template is of substantial importance.  Simpler definitions which map seeds to input bit strings, and terminal assemblies to output bitstrings, are problematic in that they allow for the possibility of encoding the computation of function $f$ in the seed structure.  Even something as innocuous sounding as allowing more than a single type of ``0'' or ``1" tile as an input bit has the subtle issue of allowing pre-computing of $f$\footnote{This subtle issue seems to exist with some previous formulations of tile assembly computation.}.

\paragraph{Input Template.}
Consider a tile set $T$ containing exactly one tile $t_0$ with label ``0", and one tile $t_1$ with label ``1".  An $n$-bit input template over tile set $T$ is an ordered pair $U = (R, B(i))$, where $R$ is an assembly over $T-\{t_0, t_1\}$, $B: \mathbb{N}_n \rightarrow \mathbb{Z}^2$, and $B(i)$ is not the position of any tile in $R$ for any $i$ from 1 to $n$.  The sequence of $n$ coordinates denoted by $B$ conceptually denotes ``wildcard" tile positions for which copies of $t_0$ and $t_1$ will be filled in for any instance of the template.  %For notation we will use $U_i$ to denote the point $B(i)$, and for bit string $b = b_1 ,\ldots b_n$,
For notation we define assembly $U_{b}$ over $T$, for bit string $b = b_1 ,\ldots b_n$,  to be the assembly consisting of assembly $R$ unioned with a set of $n$ tiles $t^i$ for $i$ from 1 to $n$, where $t^i$ is equal a translation of tile $t_{b(i)}$ to position $B(i)$.  That is, $U_{b}$ is the assembly $R$ with each position $B(i)$ tiled with either $t_0$ or $t_1$ according to the value of $b_i$.

%\paragraph{Input Template.}
%Consider a tile set $T$ containing exactly one tile $t_0$ with label ``0", and one tile $t_1$ with label ``1".  An $(n,m)$-input template over tile set $T$ is an ordered triplet $U = (R, A(i), B(i))$ where $R$ is an assembly over $T-\{t_0, t_1\}$, $A: \mathbb{N}_n \rightarrow \mathbb{Z}^2$, and $B: \mathbb{N}_m \rightarrow \mathbb{Z}^2$, and $A(i) \neq B(j)$ for any $i,j$.  The sequence of $n$ and $m$ coordinates denoted by $A$ and $B$ conceptually denote ``wildcard" tile positions for which copies of $t_0$ and $t_1$ will be filled in for any instance of the template.  For bit strings $a = a_1 ,\ldots a_n$ and $b = b_1 ,\ldots , b_m$ of length $n$ and $m$ respectively, define assembly $U_{a,b}$ over $T$ to be the assembly consisting of assembly $R$ with tile $t_j$, $j=\{0,1\}$, added at any position $A(i)$ or $B(i)$ such that $a_i = j$ or $b_i = j$ respectively.

\paragraph{Output Template.}
A $k$-bit output template is simply a sequence of $k$ coordinates denoted by function $C: N_k \rightarrow \mathbb{Z}^2$.  For an output template $V$, an assembly $A$ over $T$ is said to represent binary string $c=c_1 , \ldots , c_k$ over template $V$ if the tile at position $C(i)$ in $A$ has label $c_i$ for all $i$ from 1 to $k$.  Note that output template solutions are much looser than input templates in that there may be multiple tiles with labels ``1" and ``0", and there are no restrictions on the assembly outside of the $k$ specified wildcard positions.  The strictness for the input template stems from the fact that the input must ``look the same" in all ways except for the explicit input bit patterns.  If this were not the case, it would likely be possible to encode the solution to the computational problem into the input template, resulting is a trivial solution.

\paragraph{Function Computing Problem.}
 A \emph{tile assembly computer} (TAC) is an ordered quadruple $\Im=(T, U, V, \tau)$ where $T$ is a tile set, $U$ is an $n$-bit input template, and $V$ is a $k$-bit output template.  A TAC is said to compute function $f: \mathbb{Z}_{2}^n \rightarrow \mathbb{Z}_{2}^k$ if for any $b \in \mathbb{Z}_{2}^n$ and $c \in \mathbb{Z}_{2}^k$ such that $f(b)=c$, then the tile system $\Gamma_{\Im,b} = (T, U_b , \tau)$ uniquely assembles an assembly $A$ which represents $c$ over template $V$.  For a TAC $\Im$ that computes the function $f: Z_{2}^{2n} \rightarrow Z_{2}^{n+1}$ where $f(r_1 \ldots r_{2n}) = r_1\ldots r_n + r_{n+1}\ldots r_{2n}$, we say that $\Im$ is an $n$-bit \emph{adder} TAC with inputs $a=r_1 \ldots r_n$ and $b=r_{n+1}\ldots r_{2n}$.  An $n$-bit \emph{multiplier} TAC is defined similarly.
 %
%
%
%Consider an $n$-bit input and $k$-bit output function $f: Z_{2}^n \rightarrow Z_{2}^k$.
%
%A tile set $T$ is said to \emph{compute} function $f$ at temperature $\tau$ with $n$-bit input template $R$ and $k$-bit output template $V$ if for any $b \in Z_{2}^n$ and $c \in Z_{2}^k$ such that $f(b)=c$, then the tile system $\Gamma = (T, R_b , \tau)$ uniquely assembles an assembly $A$ which represents $c$ over template $V$.

%\paragraph{Addition and Multiplication.}  In this paper we focus on the particular problems of addition and multiplication.  In the case of addition, we are given two $n$-bit strings $X=x_1 \ldots x_n$ and $Y=y_1 \ldots y_n$ and wish to compute the string $C$ that is the sum of $X$ and $Y$.  Such a problem can be described in terms of computing a $2n$-bit input, $n+1$-bit output function in which the concatenation of strings $X$ and $Y$ serve as the input to the function. Multiplication can be defined similarly with a $2n$-bit input, $2n$-bit output function.
%%A tileset $T$ is said to be an \emph{addition} tile set at temperature $\tau$ if:  For any positive integers $n$ and $m$ there exists an $(n,m)$-input template $U$ and an $(max(n,m)+1)$-output template $V$ such that for any $n$-bit and $m$-bit strings $a$ and $b$, $\Gamma = (U_{a,b} , T, \tau)$ uniquely assembles $P$ such that $P$ is a $V$ representation of $a+b$.

\subsection{Run Time Models}\label{sec:runTimeModels}
We analyze the complexity of self-assembly arithmetic under two established run time models for tile self-asembly: the \emph{parallel} time model~\cite{BeckerRR06,BRUN2007} and the \emph{continuous} time model~\cite{AdChGoHu01,ACGHKMR02,CGM04,BeckerRR06}.  Informally, the parallel time model simply adds, in parallel, all singleton tiles that are attachable to a given assembly within a single time step.  The continuous time model, in contrast, models the time taken for a single tile to attach as an exponentially distributed random variable.  The parallelism of the continuous time models stems from the fact that if an assembly has a large number of attachable positions, then the \emph{first} tile to attach will be an exponentially distributed random variable with rate proportional to the number of attachment sites, implying that a larger number of open positions will speed up the expected time for the next attachment.  Technical definitions for each run time model are provided in Section~\ref{app:runTimeModels}.  For a deterministic tile system $\Gamma$, we use notation $\rho_\Gamma$ and $\varsigma_\Gamma$ to denote the parallel and continuous run time of $\Gamma$ respectively.  When not otherwise specified, we use the term \emph{run time} to refer to parallel run time by default. %After describing both timing models, we introduce some preliminary general analysis showing that the two models can differ only slightly.

\section{Lower Bound for Long Distance Communication}
%In this section we formulate a class of problems we term the \emph{communication} problems in which the goal is to compute a simple AND function on a 2-bit input given that the input template separates the 2 input bits some specified distance $\Delta$.  We formulate this problem for the purposes of providing lower bounds on the worst-case time complexity for this problem.  We then reduce this problem to addition and multiplication problems in 2D and 3D to provide worst case lower bounds for addition and multiplication.
%
%\subsection{High-Level Sketch of Lower Bound Proofs}  
To prove lower bounds for addition and multiplication in 2D and 3D, we do the following.  First, we consider two identical tile systems with the exception of their respective seed assemblies which differ in exactly one tile location.  We show in Lemma~\ref{lemma:diffProp} that after $\Delta$ time steps, all positions more than $\Delta$ distance from the point of initial difference of the assemblies must be identical among the two systems.  We then consider the \emph{communication} problem, formally defined in Section~\ref{sec:communicationProblem}, in which we compute the AND function of two input bits under the assumption that the input template for the problem separates the two bits by distance $\Delta$.  For such a problem, we know that the output position of the solution bit must be at least distance $\frac{1}{2} \Delta$ from one of the two input bits. As the correct output for the AND function must be a function of both bits, Lemma~\ref{lemma:diffProp} implies that at least $\frac{1}{2} \Delta$ steps are required to guarantee a correct solution as argued in Theorem~\ref{thm:communicationLower}.

With the lower bound of $\frac{1}{2}\Delta$ established for the communication problem, we move on to the problems of addition and multiplication of $n$-bit numbers.  We show how the communication problem can be reduced to these problems, thereby yielding corresponding lower bounds.  In particular, consider the addition problem in 2D.  As the input template must contain positions for $2n$ bits, in 2D it must be the case that some pair of bits are separated by at least $\Omega(\sqrt{n})$ distance according to Lemma~\ref{lemma:pointCrowding}.  Focusing on this pair of bit positions in the addition template, we can create a corresponding communication problem template with the same two positions as input.  To guarantee the correct output, we hard code the remaining bit positions of the addition template such that the addition algorithm is guaranteed to place the AND of the desired bit pair in a specific position in the output template, thereby constituting a solution to the $\Delta = \Omega(\sqrt{n})$ communication problem, which implies the addition solution cannot finish faster than $\Omega(\sqrt{n})$ in the worst case.  A similar reduction can be applied to multiplication.  The precise reductions are detailed in Theorems~\ref{thm:addition} and~\ref{thm:multiplication}.  The result statements are as follows.

\begin{theorem} \label{cor:addition}
Any $d$-dimension $n$-bit adder TAC has worst case run-time $\Omega(\sqrt[d]{n})$.
\end{theorem}

\begin{theorem}\label{cor:multiplication}
Any $d$-dimension $n$-bit multiplier TAC has worst case run-time $\Omega(\sqrt[d]{n})$.
\end{theorem}

\section{Addition In Average Case Logarithmic Time} \label{sec:avgcase}
\begin{figure}[htp]
	\centering
	\includegraphics[scale=1.5]{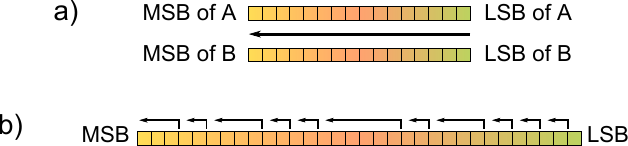}
	\caption{Arrows represent carry origination and propagation direction. a) This schematic represents the previously described $O(n)$ worst case addition for addends $A$ and $B$~\cite{BRUN2007}. The least significant and most significant bits of $A$ and $B$ are denoted by LSB and MSB, respectively. b) The average case $O(\log{n})$ construction described in this paper is shown here. Addends $A$ and $B$ populate the linear assembly with bit $A_i$ immediatly adjacent to $B_i$. Carry propagation is done in parallel along the length of the assembly.}
	\label{fig:coolfig_avgcase}
\end{figure}
We construct an adder TAC that resembles an electronic carry-skip adder in that the carry-out bit for addend pairs where each addend in the pair has the same bit value is generated in a constant number of steps and immediately propagated. When each addend in a pair of addends does not have the same bit value, a carry-out cannot be deduced until the value of the carry-in to the pair of addends is known. When such addends combinations occur in a contiguous sequence, the carry must ripple through the sequence from right-to-left, one step at a time as each position is evaluated. Within these worst-case sequences, our construction resembles an electronic ripple-carry adder. We show that using this approach it is possible to construct an $n$-bit adder TAC that can perform addition with an average runtime of $O(\log{n})$ and a worst-case runtime of $O(n)$.

\begin{lemma} \label{lem:longest_substring}
Consider a non-negative integer $N$ generated uniformly at random from the set $\{0,1,...,2^{n}-1\}$.  The expected length of the longest substring of contiguous ones in the binary expansion of $N$ is $O(\log{n})$.
\end{lemma}

For a detailed discussion of Lemma \ref{lem:longest_substring} please see Schilling\cite{SCHILL1990}.

\begin{theorem} \label{thm:avg_case_adder}
For any positive integer $n$, there exists an $n$-bit adder TAC (tile assembly computer) that has worst case run time $O(n)$ and an average case run time of $O(\log{n})$.
\end{theorem}
The proof of Theorem \ref{thm:avg_case_adder} follows from the construction of the adder in Appendices \ref{sec:avgcase_construction}, \ref{sec:avgcase_time}, and \ref{sec:avgcase_correct}.

\section{Optimal $O(\sqrt{n})$ Addition}\label{sec:worstcase}
We show how to construct an adder TAC that achieves a run time of $O(\sqrt{n})$, which matches the lower bound proved in Theorem \ref{thm:addition}.  This adder TAC closely resembles an electronic carry-select adder in that the addends are divided into sections of size $\sqrt{n}$ and the sum of the addends comprising each is computed for both possible carry-in values. The correct result for the subsection is then selected after a carry-out has been propagated from the previous subsection.  Within each subsection, the addition scheme resembles a ripple-carry adder. This construction works well with massive parallelism and allows us to construct an optimal $O(\sqrt{n})$ adder TAC in two dimensions.

\begin{theorem}\label{thm:worstcase2d}
There exists a 2D $n$-bit adder TAC with a worst case run-time of $O(\sqrt{n})$.
\end{theorem}

The proof of Theorem \ref{thm:worstcase2d} follows from the construction of the tile assembly adder in Section \ref{sec:worst_construction_overview} and Appendices \ref{sec:worst_construction}, \ref{sec:worst_time}, and \ref{sec:worst_correct}.

\subsection{Construction Overview} \label{sec:worst_construction_overview}
Due to space limitations, an overview of the construction is included in this section. A more detailed description of the adder construction (along with the full tile set) can be found in Appendix \ref{sec:worst_construction}. Figures \ref{fig:addition_input_template_o}-\ref{fig:addition_output_template_o} are examples of I/O templates for a $9$-bit adder TAC. The inputs to the addition problem in this instance are two 9-bit binary numbers $A$ and $B$ with the least significant bit of $A$ and $B$ represented by $A_{0}$ and $B_{0}$, respectively. Upon inclusion of the seed assembly (Figure \ref{fig:worstcase_addition_o}a) to the tile set (Figure \ref{fig:worsttileset}), each pair of bits from $A$ and $B$ are summed (for example, $A_{0}+B_{0}$) (Figure \ref{fig:worstcase_addition_o}b).  Each row computes this addition step independently and outputs a carry or no-carry west face glue on the westernmost tile of each row (Figure \ref{fig:worstcase_addition_o}c). As the addition tiles bind to the seed, tiles from the incrementation tile set (Figure \ref{fig:worsttileset}c) may also begin to attach. The purpose of the incrementation tiles is to determine the sum for each $A$ and $B$ bit pair in the event of a no-carry from the row below and in the event of a carry from the row below (Figure \ref{fig:worstcase_increment_o}). The final step of the addition mechanism presented here propagates carry or no-carry information northwards from the southernmost row of the assembly. When the carry propagation column reaches the top of the assembly, the most significant bit of the sum may be determined and the calculation is complete (Figure \ref{fig:worstcase_carryprop_o}a-f).
\begin{figure}[htp]
	\centering
	\subfigure[Addition input template.]{
		\includegraphics[scale=0.80]{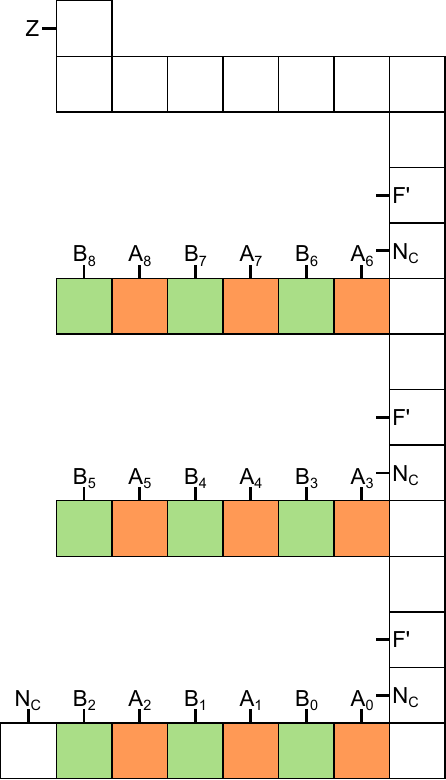}
		\label{fig:addition_input_template_o}
	}
	\subfigure[Addition output template.]{
		\includegraphics[scale=0.80]{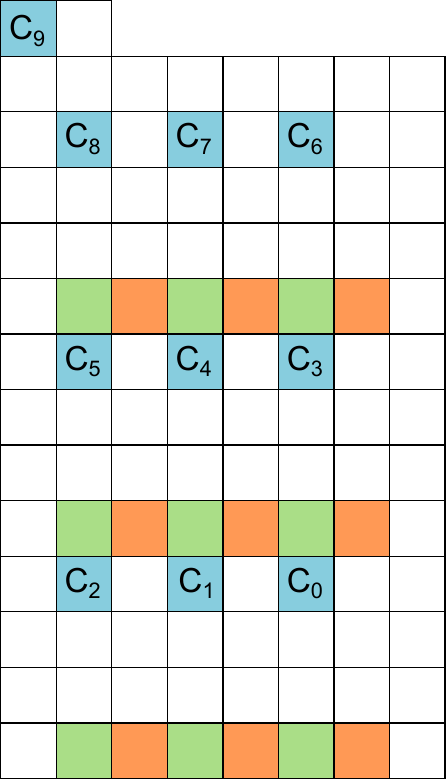}
		\label{fig:addition_output_template_o}
	}
\caption{These are example I/O templates for the worst case $O(\sqrt{n})$ time addition introduced in Section \ref{sec:worstcase}.}
\end{figure}

\begin{figure}[htp]
	\centering
	\includegraphics[scale=0.80]{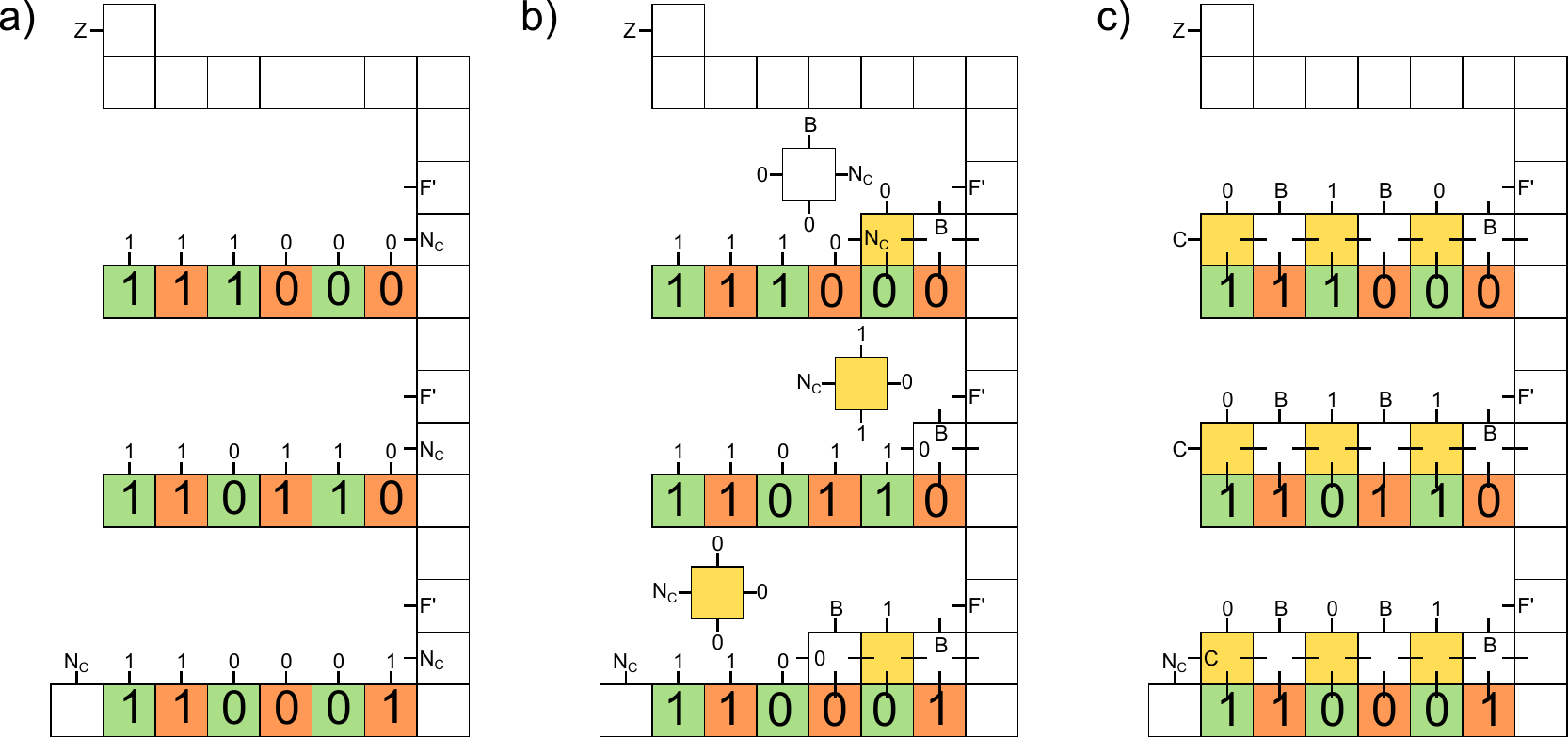}
	\caption{Step 1: Addition.}
	\label{fig:worstcase_addition_o}
\end{figure}

\begin{figure}[htp]
	\centering
	\includegraphics[scale=0.80]{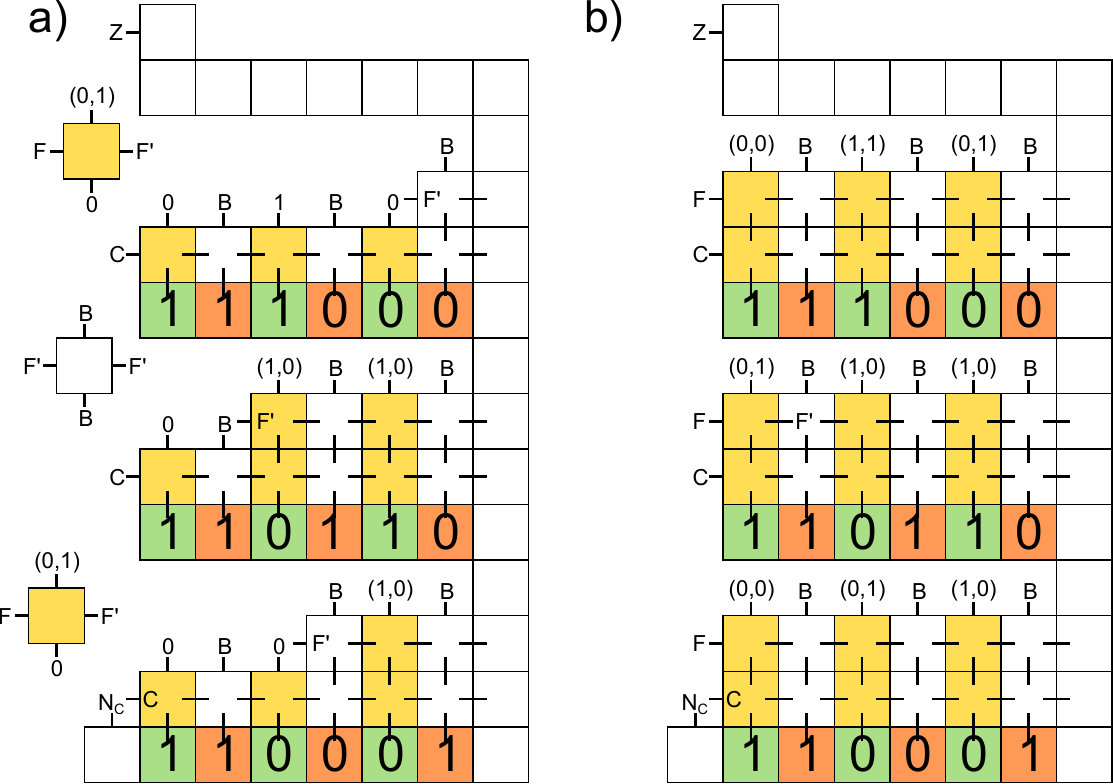}
	\caption{Step 2: Increment.}
	\label{fig:worstcase_increment_o}
\end{figure}

\begin{figure}[htp]
	\centering
	\includegraphics[scale=0.80]{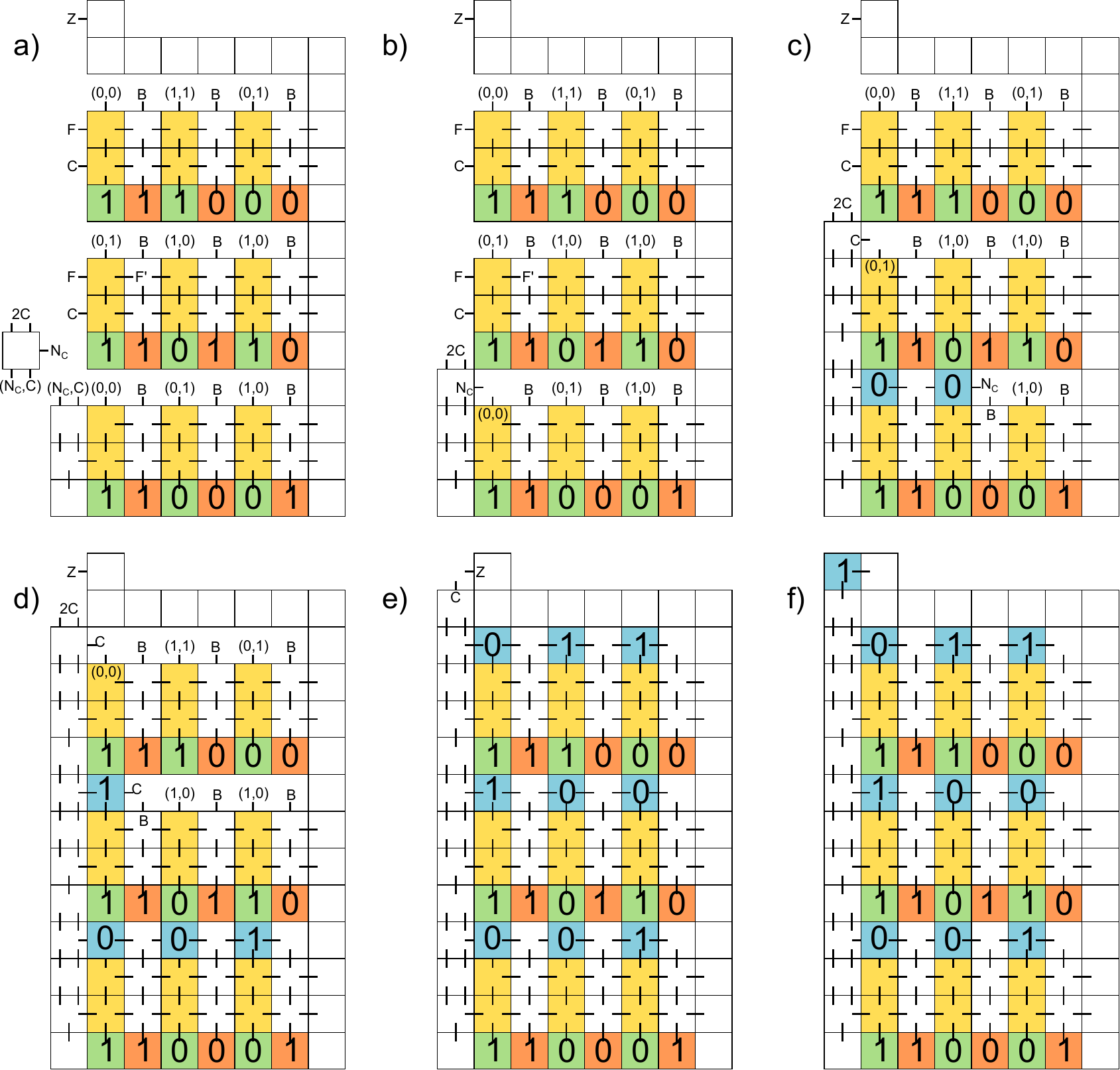}
	\caption{Step 3: Carry Propagation and Output.}
	\label{fig:worstcase_carryprop_o}
\end{figure}
%\section{$O(\log{n})$ Average Case, $\sqrt{n}$ Worst Case Addition.} \label{sec:combinedcase}
\section{Towards Faster Addition} \label{sec:combinedcase}
%\begin{figure}[htp]
%	\centering
%	\includegraphics[scale=1]{images/coolfig_combinedcase}
%	\label{fig:coolfig_combinedcase}
%\end{figure}
\begin{figure}[htp]
	\centering
	\subfigure[]{
		~~~~~~
		\includegraphics[scale=1.3]{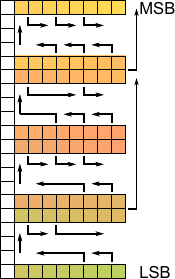}
		\label{fig:coolfig_combinedcase}
		~~~~~~
	}
	~~~~~~~~
	\subfigure[]{
		~~~~~~
		\includegraphics[scale=1.3]{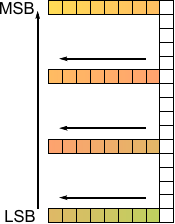}
		\label{fig:coolfig_worstcase}
		~~~~~~
	}
	\caption{Arrows represent carry origination and direction of propagation for a) the $O(\log{n})$ average case, the $O(\sqrt{n})$ worst case combined construction, and b) the $O(\sqrt{n})$ worst case construction.}
\end{figure}
In this section we combine the approaches described in Sections \ref{sec:avgcase} and \ref{sec:worstcase} in order to achieve both $O(\log{n})$ average case addition and $O(\sqrt{n})$ worst case addition. This construction resembles the construction described in Section \ref{sec:worstcase} in that the numbers to be added are divided into sections and values are computed for both possible carry-in bit values. Additionally, the construction described here lowers the average case run time by utilizing the carry-skip mechanism described in Section \ref{sec:avgcase} within each section and between sections.

\begin{theorem} \label{thm:combinedAddition}
There exists a 2-dimensional $n$-bit adder TAC with an average run-time of $O(\log{n})$ and a worst case run-time of $O(\sqrt{n})$.
\end{theorem}

The proof follows from the construction of the adder in Appendices \ref{sec:combined_construction}, \ref{sec:combined_time}, and \ref{sec:combined_correct}.

\section{Addition in Three Dimensions}\label{sec:3Dcombinedcase}
In this section we extend our adder construction into the third dimension to achieve a $O(\sqrt[3]{n})$ upper bound, which meets the $\Omega(\sqrt[3]{n})$ lower bound from Theorem \ref{thm:addition}.  Due to space and time constraints a general overview is given.  We begin by creating $\sqrt[3]{n}$ total $\sqrt[3]{n} \times \sqrt[3]{n}$ constructions exactly as per Section \ref{sec:combinedcase} stacked one atop the other in an alternating fashion such that every odd plate, beginning with the first, has its MSB in the $NW$ corner while the even plate has its LSB in that same corner. We continue by applying the same addition algorithm that was presented in Section \ref{sec:combinedcase} to all plates where every lower plate passes its carry-out to the upper plate.  This is very similar to how carry out bits are passed between sections in the construction described in Section \ref{sec:combinedcase}.  Finally, every lower plate will pass the appropriate carry to the next lower plate. Please see Figure \ref{fig:3dcombined} for a visual overview of this process.

\begin{figure}[htp]
	\centering
	\includegraphics[scale=1.0]{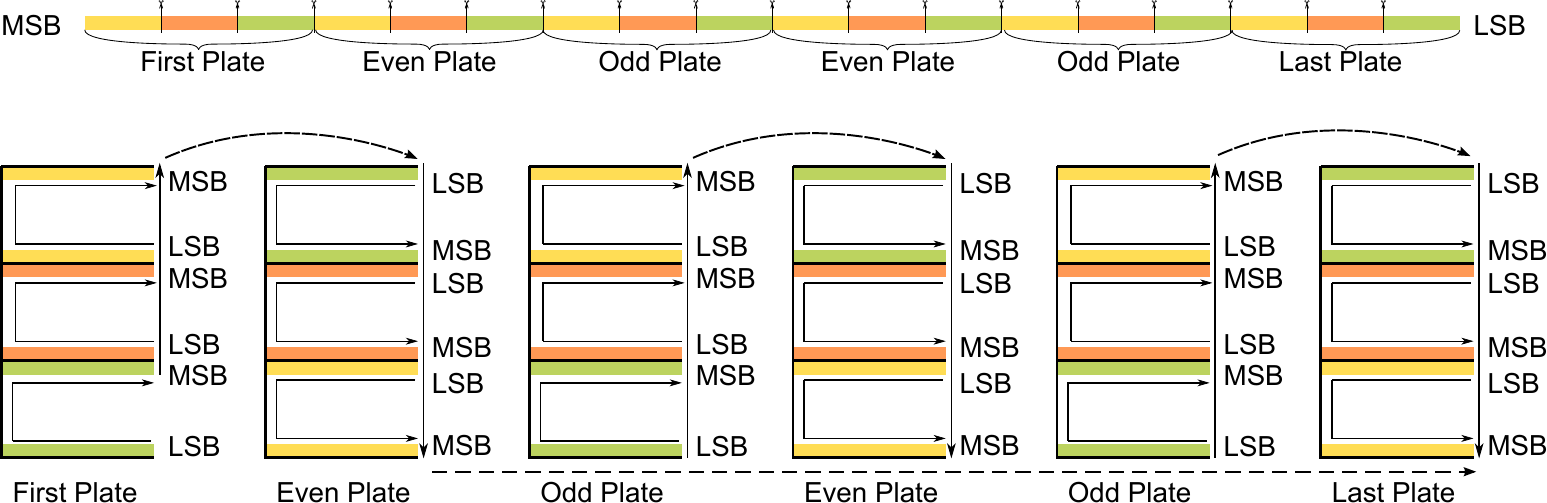}
	\caption{An overview of the process to add two number in optimal time in $3D$.}
	\label{fig:3dcombined}
\end{figure} 

\begin{theorem}\label{thm:combined3d}
There exists a 3-dimensional $n$-bit adder TAC with an average run-time of $O(\log{n})$ and a worst case run-time of $O(\sqrt[3]{n})$.
\end{theorem}

We present a high-level sketch of the proof of Theorem \ref{thm:combined3d} in Appendix \ref{proof:combined3d}

\subsection{Proof Sketch of Theorem \ref{thm:combined3d}} \label{proof:combined3d}

We present a high-level sketch of the proof of Theorem \ref{thm:combined3d} here. Similar with the $O(\log{n})$ average case $O(\sqrt{n})$ worst case combined addition presented in Section \ref{sec:combinedcase}, the binary numbers to be added are separated into sections each with length $O(\sqrt[3]{n})$. We arrange these numbers on $\sqrt[3]{n}$ scafolds of size $\sqrt[3]{n} \times \sqrt[3]{n}$ $2D$. The dashed lines in Figure \ref{fig:3dcombined} are carries transmitted between each plate in the third dimension. The lines on the north are carries transmitted from each odd plate to its next even plate. The dashed lines on the southernmost part of Figure \ref{fig:3dcombined} are the carries transmitted from the first grouped plates to the last plate.  Since the numbers are all fit compactly with only a constant amount space between each plate and a constant amount of space between each section, any one side of the cube is at most $O(\sqrt[3]{n})$.  Therefore, this size constraint along with the algorithms previously presented allow us to have an optimal $O(\sqrt[3]{n})$ time complexity with an average case complexity of $O(\log{n})$.

\section{Multiplication} \label{sec:multiplication}
In this section we sketch our $O(n^{\frac{5}{6}})$ time multiplication algorithm.  Let $a=a_{n-1}a_{n-2}\ldots a_0$ and $b=b_{n-1}b_{n-2}\ldots b_0$ be two $n$-bit numbers to be multiplied.  The key idea of our algorithm is to compute the product $ab$ by computing the sum $ab = \sum_{i=0}^{n-1} a\cdot b_i\cdot 2^i$ by repeated application of our $O(\sqrt{n})$ time addition algorithm in a parallelized fashion.  The statement of our result is as follows:

\begin{theorem}\label{thm:56multiplication}
There exists a 3D $n$-bit multiplier TAC with a worst case run-time of $O(n^{\frac{5}{6}})$.
\end{theorem}

The seed assembly for our construction consists of a 2D $O(\sqrt{n})\times O(\sqrt{n})$ pad encoding inputs $a$ and $b$ in a ``snaked" pattern as is done with the seed in Appendix~\ref{sec:combined_construction}.  Our algorithm expands this seed into a $O(n^{5/6})\times O(n^{5/6}) \times O(n^{5/6})$ \emph{mega-cube} consisting of $n$ \emph{number blocks} each of dimension $O(n^{1/2}) \times O(n^{1/2}) \times O(n^{1/2})$.  Each number cube conceptually is associated with a distinct term from the numbers $a\cdot b_i \cdot 2^i$, and the mega-cube is organized into a $O(n^{1/3}) \times O(n^{1/3}) \times O(n^{1/3})$ array of number cubes (yielding the total dimension of $O(n^{5/6})$).  

Assembly of the mega-cube proceeds by initializing the growth of $O(n^{1/3})$ rows of number cubes as shown in Figure~\ref{fig:Mult_HighLevel} (b).  The growth that initializes each row extends from the seed assembly at a $O(1)$ rate, thereby initializing all rows in time $O(n^{5/6})$.  Each row proceeds by computing the appropriate value $a\cdot b_i\cdot 2^i$ for the given number cube and adding it via our 2D addition algorithm to a running total in $O(\sqrt{n})$ time per number cube.  Note that computing a given value $a\cdot b_i \cdot 2^i$ is simply a bit shift rather than a full fledged multiplication and can therefore be derived quickly as shown in Appendix~\ref{app:multiplcation}.  Therefore, each row finishes in $O(n^{1/2} \cdot n^{1/3}) = O(n^{5/6})$ time once initialized, yielding a $O(n^{5/6})$ finish time for the entire first level of the mega-cube shown in Figure~\ref{fig:Mult_HighLevel} (b).  Further, in addition to initializing the growth of the first plane of the mega-cube, a column is initialized to grow up into the third dimension to seed each of the $O(n^{1/3})$ levels of the mega-cube, implying that all levels finish in time $O(n^{5/6})$.  Finally, each sum from the rows of a given level are computed into subtotals for each level as depicted in Figure~\ref{fig:Mult_HighLevel} (c), followed by the summation of these subtotals into a final solution as shown in Figure~\ref{fig:Mult_HighLevel} (d).

%Based on the Optimal $O(\sqrt{n})$ Addition in previous section, we designed a two $n$ bit binary number multiplication TAC that achieves a run time of $O(n^{\frac{5}{6}})$. The multiplication of two $n$ bit numbers is basically apply the summation of $n$ numbers and each number has at most $2*n$ bits. The TAC we designed will deploy those numbers to a cube and do the summation of all the numbers in parallel; firstly add the numbers in each columns \ref{fig:Mult_HighLevel} b), then adding the results of all columns in each layer\ref{fig:Mult_HighLevel} c), finally add results of each layer together and got the final result of multiplication\ref{fig:Mult_HighLevel} d).

%In the running time analysis section, we will prove that the running time of any two number addition in this cube will take running time of $O(\sqrt{n})$, then adding all $O(\sqrt[3]{n})$ numbers in one column will take $O(n^{\frac{5}{6}})$, and also the summation of all $O(\sqrt[3]{n})$ columns in each layer and all $O(\sqrt[3]{n})$ layers in the entire cube will take the runtime of $O(n^{\frac{5}{6}})$, so the total runtime would be $O(n^{\frac{5}{6}})$.
\begin{figure}[htp]
	\centering
	\includegraphics[width=4in]{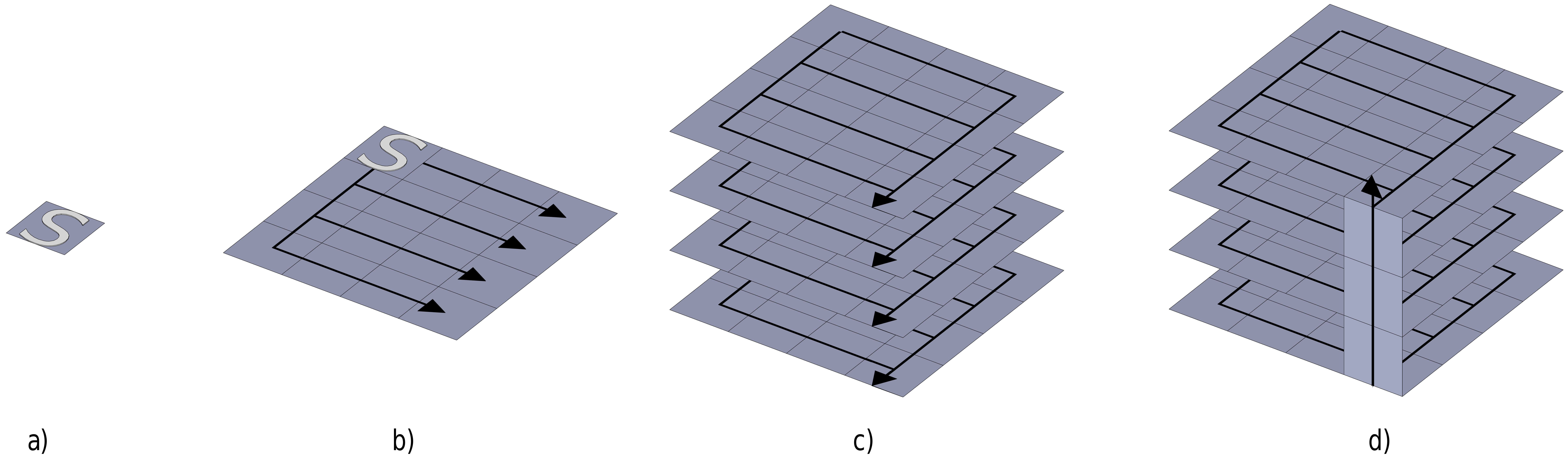}
	\caption{High level sketch of multiplication}
	\label{fig:Mult_HighLevel}
\end{figure}

The details for our multiplication system are provided in Appendix  \ref{sec:multiplication_construction}%, \ref{sec:multiplication_time}, and \ref{sec:multiplication_correct}.  
The full details of our construction are complex and our description in this abstract may be difficult to follow for a general audience.  As ongoing work we are developing an expanded write-up and analysis that will be more approachable and we will present it in the final published version of this paper.  Further, our run time bound of $O(n^{5/6})$ applies to the parallel time model.  We therefore get a $O^*(n^{5/6})$ continuous time bound from Lemma~\ref{lemma:runTimeConnection}.  We believe a tighter analysis should yield an equal asymptotic bound for both run time models.

\section{Simulation} \label{sec:sim}
Tile self-assembly software simulations were conducted to visualize the diverse approaches to fast arithmetic presented in this paper, as well as to compare them to previous work. The adder tile constructions described in Sections  \ref{sec:avgcase}-\ref{sec:combinedcase}, and the previous best\cite{BRUN2007} were simulated using the two timing models described in Sections \ref{sec:runtime}.

Figure \ref{fig:graphs} demonstrates the power of the approaches described in this paper compared to previous work. 

\begin{figure}[htp]
	\centering
	\includegraphics[width=5in]{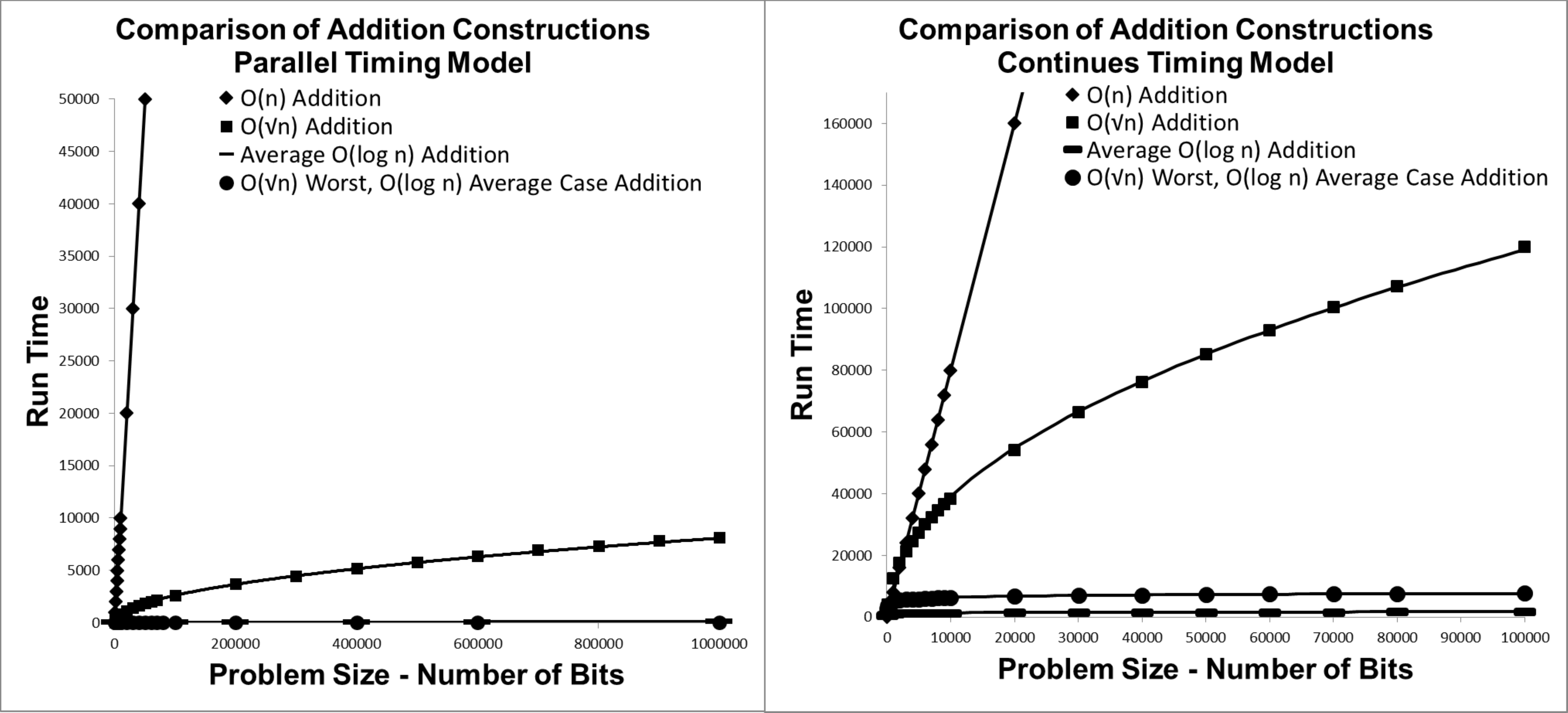}
	\caption{A comparison of adder constructions.}
	\label{fig:graphs}
\end{figure}
\section{Future Work}

%%%Beef this section up%%%

The results of this paper are just a jumping off point and provide numerous directions for future work. One promising direction is further exploration of the complexity of multiplication.  Can $O(n^{1/3})$ or $O(n^{1/2})$ be achieved in 3D?  Is sublinear multiplication possible in 2D, or is there a $\Omega
(n)$ lower bound?  What is the average case complexity of multiplication?

Another direction for future work is the consideration of metrics other than run time.  One potentially important metric is the geometric space taken up by the computation.  Our intent is that fast function computing systems such as those presented in this paper will be used as building blocks for larger and more complex self-assembly algorithms.  For such applications, the area and volume taken up by the computation is clearly an important constraint.  Exploring tradeoffs between run time and imposed space limitations for computation may be a promising direction with connections to resource bounded computation theory.  Along these lines, another important direction is the general development of design methodologies for creating black box self-assembly algorithms that can be plugged into larger systems with little or no ``tweaking".

%the probabilistic run time model is logarithmically slower than the parallel run time model, as evidenced by the graphs in Section \ref{sec:sim} of this paper. What effect might each run time model have on non-constant sized tile sets? A rigorous study of the differences between these two run time models might give a better understanding of how the parallel run time model may suffer due to its simplicity.

A final direction focusses on the consideration of non-deterministic tile assembly systems to improve expected run times even for maniacally designed worst case input strings.  Is it possible to achieve $O(\log n)$ expected run time for the addition problem regardless of the input bits?  If not, are there other problems for which there is a provable gap in achievable assembly time between deterministic and non-deterministic systems?

%[TODO : discuss implementing multiplication algorithm into the simulator to empirically verify $n^{5/6}$ run time.]

%the use of non-deterministic tile assembly systems could show that it may be possible to solve certain problems faster when non-determinism is permitted. Using a non-deterministic run time model, does there exist a $O(\log{n})$ average time adder TAC for \emph{all} input numbers?

%\paragraph{Sorting and Summation.} Comming soon...
%\paragraph{Output Formats.} Output bits in different places depending on computed answer.
%\paragraph{Better Simulators.} Implement timing model and configuration tools as a plugin to a large scale simulator such as the ISTAS~\cite{}.

\section*{Acknowledgements}
We would like to thank Ho-Lin Chen and Damien Woods for helpful discussions regarding Lemma~\ref{lemma:maxExpSums} and Florent Becker for discussions regarding timing models in self-assembly.  We would also like to thank Matt Patitz for helpful discussions of tile assembly simulation.

%----------------------------------------------------------------
\newpage
\bibliographystyle{amsplain}
\bibliography{tam, tamUSA}

\appendix

\newpage

\section{Run Time Model Definitions}\label{app:runTimeModels}
\paragraph{Parallel Run-time.} \label{sec:runtime}
%The first run time model we consider in this paper was first proposed in~\cite{BeckerRR06}.
For a deterministic tile system $\Gamma = (T,S,\tau)$ and assembly $A \in \texttt{PROD}_\Gamma$, the 1-\emph{step transition} set of assemblies for $A$ is defined to be $ STEP_{\Gamma, A} = \{ B \in \texttt{PROD}_\Gamma | A \rightarrow_{T,\tau} B \} $.  For a given $A \in \texttt{PROD}_\Gamma$, let $\texttt{PARALLEL}_{\Gamma, A} = \bigcup_{B\in STEP_{\Gamma, A}}  B$, ie, $\texttt{PARALLEL}_{\Gamma, A}$ is the result of attaching all singleton tiles that can attach directly to $A$.  Note that since $\Gamma$ is deterministic, $\texttt{PARALLEL}_{\Gamma, A}$ is guaranteed to not contain overlapping tiles and is therefore an assembly.  For an assembly $A$, we say $A \rightrightarrows_\Gamma A'$ if $A' = \texttt{PARALLEL}_{\Gamma, A}$.  We define the \emph{parallel run-time} of a deterministic tile system $\Gamma = (T,S,\tau)$ to be the non-negative integer $k$ such that $A_1 \rightrightarrows_\Gamma A_2 \rightrightarrows_\Gamma \ldots \rightrightarrows_\Gamma A_k $ where $A_1 = S$ and $\{A_k \} = \texttt{TERM}_\Gamma$.  As notation we denote this value for tile system $\Gamma$ as $\rho_\Gamma$.  For any assemblies $A$ and $B$ in $\texttt{PROD}_\Gamma$  such that $A_1  \rightrightarrows_\Gamma A_2 \rightrightarrows_\Gamma \ldots \rightrightarrows_\Gamma A_{k}$ with $A=A_1$ and $B=A_k$, we say that $A \rightrightarrows^k_\Gamma B$. Alternately, we denote $B$ with notation $A \rightrightarrows^k_\Gamma$.  For a TAC $\Im=(T,U,V,\tau)$ that computes function $f$, the run time of $\Im$ on input $b$ is defined to be the parallel run-time of tile system $\Gamma_{\Im,b} = (T,U_b, \tau)$.  Worst case and average case run time are then defined in terms of the largest run time inducing $b$ and the average run time for a uniformly generated random $b$.

\paragraph{Continuous Run-time.}\label{sec:continuousRuntime}
%The second run time model we consider is the \emph{continuous} time model~\cite{} in which the
In the \emph{continuous} run time model the assembly process is modeled as a continuous time Markov chain and the time spent in each assembly is a random variable with an exponential distribution.  The states of the Markov chain are the elements of $PROD_\Gamma$ and the transition rate from $A$ to $A'$ is $\frac{1}{|T|}$ if $A \rightarrow A'$ and 0 otherwise.  For a given tile assembly system $\Gamma$, we use notation $\Gamma_{MC}$ to denote the corresponding continuous time Markov chain.  Given a deterministic tile system $\Gamma$ with unique assembly $U$, the run time of $\Gamma$ is defined to be the expected time of $\Gamma_{MC}$ to transition from the seed assembly state to the unique sink state $U$.  As notation we denote this value for tile system $\Gamma$ as $\varsigma_\Gamma$ .  One immediate useful fact that follows from this modeling is that for a producible assembly $A$ of a given deterministic tile system, if $A \rightarrow A'$, then the time taken for $\Gamma_{MC}$ to transition from $A$ to the first $A''$ that is a superset of $A'$ (through 1 or more tile attachments) is an exponential random variable with rate $1/|T|$.  For a more general modeling in which different tile types are assigned different rate parameters see~\cite{AdChGoHu01,ACGHKMR02,CGM04}.  Our use of rate $1/|T|$ for all tile types corresponds to the assumption that all tile types occur at equal concentration and the expected time for 1 tile (of any type) to hit a given attachable position is normalized to 1 time unit.  Note that as our tile sets are constant in size in this paper, the distinction between equal or non-equal tile type concentrations does not affect asymptotic run time.  For a TAC $\Im=(T,U,V,\tau)$ that computes function $f$, the run time of $\Im$ on input $b$ is defined to be the continuous run-time of tile system $\Gamma_{\Im,b} = (T,U_b, \tau)$.  Worst case and average case run time are then defined in terms of the largest run time inducing $b$ and the average run time for a uniformly generated random $b$.

\paragraph{Relating Parallel Time and Continuous Time.}
The following Lemma states an upper and lower bound on continuous time with respect to parallel time.  Both bounds are straightforward to derive with the lower bound appearing in~\cite{BeckerRR06} and the upper bound being derivable in a fashion similar to the proof of Theorem 5.2 in~\cite{AdChGoHu01}.  The most dramatic distinction between parallel and continuous time occurs when the number of tile types, $|T|$, is large, as this slows down assembly in the continuous model but does not affect parallel time. When $|T| = O(1)$, which we adhere to in this paper, the lemma implies that the timing models are very close.

\begin{lemma}\label{lemma:runTimeConnection}
Consider a deterministic aTAM system $\Gamma = (T,S,\tau)$ that uniquely assembles (finite) assembly $A$.  Let $\rho_\Gamma$ denote the parallel time for $\Gamma$ to assemble $A$ and let $\varsigma_\Gamma$ denote the continuous time for $\Gamma$ to assemble $A$.  Then,
\begin{itemize}
    \item $\varsigma_\Gamma \geq |T|\rho_\Gamma$
    \item $\varsigma_\Gamma = O(|T|\rho_\Gamma\log(\rho_\Gamma + |S|))$
\end{itemize}
\end{lemma}

It is not clear if the $\log(\rho_\Gamma)$ factor in the upper bound for continuous time is necessary, and in fact for the algorithms in this paper the parallel and continuous run times are equal up to constant factors. However, the $\log(|S|)$ factor for general systems is important.  Consider a $1\times |S|$ line of tiles for a seed assembly to which $|S|$ copies of a single tile type may attach to form a terminal $2\times |S|$ assembly.  This system has a parallel run time of just 1, but $\Theta(\log |S|)$ continuous run time, implying that the $\log |S|$ factor in our bound is tight in the general case.  In~\cite{BeckerRR06} a comparison between these two run time models is considered in the extension in which tile types may have different concentrations.

\section{Lower Bounds}
\subsection{Communication Problem.}\label{sec:communicationProblem}  The $\Delta$-communication problem is the problem of computing the function $f( b_1, b_2 )$ $ = b_1 \wedge b_2$ for bits $b_1$ and $b_2$ in the 3D aTAM under the additional constraint that the input template for the solution $U = (R, B(i))$ be such that $\Delta = \max ( | B(1)_x - B(2)_x |, |B(1)_y - B(2)_y | , | B(1)_z - B(2)_z | )$.
 %That is, the $\Delta$-communication problem is to compute the AND of two bits given that the input locations of the two bits are distance $\Delta$ from each other in at least one of the possible dimensions.

We first establish a lemma which intuitively states that for any 2 seed assemblies that differ in only a single tile position, all points of distance greater than $r$ from the point of difference will be identically tiled (or empty) after $r$ time steps of parallelized tile attachments:

\begin{lemma}\label{lemma:diffProp}
Let $S_{p, t}$ and $S_{p, t'}$ denote two assemblies that are identical except for a single tile $t$ versus $t'$ at position $p=(p_x,p_y,p_z)$ in each assembly.  Further, let $\Gamma = (T, S_{p , t}, \tau)$ and $\Gamma' = (T, S_{p, t'}, \tau)$ be two deterministic tile assembly systems such that $S_{p,t} \rightrightarrows^r_{\Gamma'} R$ and $S_{p,t'} \rightrightarrows^r_{\Gamma'} R'$ for non-negative integer $r$.  Then for any point $q=(q_x, q_y, q_z)$ such that $r < \max( |p_x - q_x|, |p_y-q_y|, |p_z-q_z|)$,  it must be that $R_{q} = R'_{q}$, ie, $R$ and $R'$ contain the same tile at point $q$.
\end{lemma}

The proof for Lemma \ref{lemma:diffProp} can be found in Appendix \ref{proof:diffProp}.

\begin{theorem}\label{thm:communicationLower}
Any solution to the $\Delta$-communication problem has run time at least $\frac{1}{2} \Delta$.
\end{theorem}

The proof for Theorem \ref{thm:communicationLower} may be found in Appendix \ref{proof:communicationLower}.

\subsection{$\Omega( \sqrt[d]{n})$ Lower Bounds for Addition and Multiplication.}
We now show how to reduce instances of the communication problem to the arithmetic problems of addition and multiplication in 2D and 3D to obtain lower bounds of $\Omega( \sqrt{n} )$ and $\Omega( \sqrt[3]{n} )$ respectively.

We first show the following Lemma which lower bounds the distance of the farthest pair of points in a set of $n$ points.  We believe this Lemma is likely well known, or is at least the corollary of an established result.  We include it's proof for completeness.

\begin{lemma} \label{lemma:pointCrowding}
For positive integers $n$ and $d$, consider $A \subset \mathbb{Z}^d$ and $B \subset \mathbb{Z}^d$ such that $A \bigcap B = \emptyset$ and $|A|=|B|=n$.  There must exist points $p \in A$ and $q \in B$ such that $\Delta_{p , q} \geq \lceil\frac{1}{2}\lceil\sqrt[d]{2n}\rceil \rceil - 1$.

%Given a set $A$ containing $n$ distinct dimension $d$ points, and a second set $B$ containing $n$ distinct dimension $d$ points such that $A \bigcap B = \emptyset$, there must exist points $p_1 \in A$ and $p_2 \in B$ such that $\Delta_{p_1 , p_2} \geq \lceil\frac{1}{2}\lceil\sqrt[d]{2n}\rceil \rceil - 1$.
\end{lemma}

The proof of Lemma \ref{lemma:pointCrowding} can be found in Appendix \ref{proof:pointCrowding}.

\begin{theorem} \label{thm:addition}
Any $n$-bit adder TAC that has a dimension $d$ input template for $d=1$, $d=2$, or $d=3$, has a worst case run time of $\Omega( \sqrt[d]{n} )$.
\end{theorem}

The proof of Theorem \ref{thm:addition} can be found in Appendix \ref{proof:addition}.
%Consider a TAC $\Im = (T, R, V, \tau)$ that is an $n$-bit adder.
%
%
%
% For clarity, let $A(1) \ldots A{n}$ denotes the first $n$-bit positions of the input template of $\Im$, and let $B(1) \ldots B(n)$ represent the remaining $n$-bit positions of the input template, and let $f(a,b) = c = a+b$ denote the function that $\Im$ computes where $a$ and $b$ are $n$-bit numbers, and $c$ is a $(n+1)$-bit number.  In dimension $d$, there must be points $A(i)$ and $B(j)$ such that $\Delta_{i,j} = \max(  | A(1)_{x_1} - B(1)_{x_1} | , \ldots , |A(1)_{x_d} - B(1)_{x_d}| )$ is at least $\sqrt[d]{n}$.  Without loss of generality, assume $i<j$.  From this input template, create a new 2-bit input template by hard coding all positions of the template other than positions
%
%
%
%Now consider bit string $ a_n \ldots a_i \ldots a_1$ and $b_n \ldots b_j \ldots b_1$ such that $b_k = 0$ for all $k>$

%%ok, restart

As the bound of dimension $d$ on the input template of a TAC alone lower bounds the run time of the TAC, we get the following corollary.

We now provide a lower bound for multiplication.

\begin{theorem} \label{thm:multiplication}
Any $n$-bit multiplier TAC that has a dimension $d$ input template for $d=1$, $d=2$, or $d=3$, has a worst case run time of $\Omega(\sqrt[d]{n})$.
\end{theorem}

The proof of Theorem \ref{thm:multiplication} can be found in Appendix \ref{proof:multiplication}.

As with addition, the lower bound implied by the limited dimension of the input template alone yields the general lower bound for $d$ dimensional multiplication TACS.

\subsection{Proof of Lemma \ref{lemma:diffProp} } \label{proof:diffProp}
We show this by induction on $r$.  As a base case of $r=0$, we have that $R=S_{p,t}$ and $R'=S_{q,t'}$, and therefore $R$ and $R'$ are identical at any point outside of point $p$ by the definition of $S_{p,t}$ and $S_{p,t'}$.

Inductively, assume that for some integer $k$ we have that for all points $w$ such that $k < \Delta_{p,w} \triangleq  \max(|p_x - w_x|,$ $|p_y-w_y|,$ $|p_z-w_z|)$, we have that $R^k_w = R^{'k}_w$, where $S_{p,t} \rightrightarrows^k_\Gamma R^k$, and $S_{p,t'} \rightrightarrows^k_{\Gamma'} R^{'k}$.  Now consider some point $q$ such that $k+1 < \Delta_{p,q} \triangleq \max{|p_x - q_x|,|p_y-q_y|,|p_z-q_z|}$, along with assemblies $R^{k+1}$ and $R^{' k+1}$ where $S_{p,t} \rightrightarrows^{k+1} R^{k+1}$, and $S_{p,t'} \rightrightarrows^{k+1} R^{'k+1}$.  Consider the direct neighbors (6 of them in 3D) of point $q$.  For each neighbor point $c$, we know that $\Delta_{p,c} > k$.  Therefore, by inductive hypothesis, $R^k_c = R^{'k}_c$ where $S_{p,t} \rightrightarrows^{k}_\Gamma R^{k}$, and $S_{p,t'} \rightrightarrows^{k}_{\Gamma'} R^{'k}$.  Therefore, as attachment of a tile at a position is only dependent on the tiles in neighboring positions of the point, we know that tile $R^{k+1}_q$ may attach to both $R^{k}$ and $R^{'k}$ at position $q$, implying that $R^{k+1}_q = R^{' k+1}_q$ as $\Gamma$ and $\Gamma'$ are deterministic.

\subsection{Proof of Theorem \ref{thm:communicationLower}} \label{proof:communicationLower}
Consider a TAC $\Im = (T,U=(R,B(i)),V(i),\tau)$ that computes the $\Delta$-communication problem.  First, note that $B$ has domain of 1 and 2, and $V$ has domain of just 1 (the input is 2 bits, the output is 1 bit).  We now consider the value $\Delta_V$ defined to be the largest distance between the output bit position of $V$ from either of the two input bit positions in $B$:  Let $\Delta_V \triangleq  \max( \Delta_{B(1), V(1)} , \Delta_{B(2), V(1)})$.  Without loss of generality, assume $\Delta_V = \Delta_{B(1),V(1)}$.  Note that $\Delta_V \geq \frac{1}{2} \Delta$.

Now consider the computation of $f(0,1) = 0$ versus the computation of $f(1,1) = 1$ via our TAC $\Im$.  Let $A^0$ denote the terminal assembly of system $\Gamma_0 = (T, U_{0,1}, \tau)$ and let $A^1$ denote the terminal assembly of system $\Gamma_1 = (T, U_{1,1}, \tau)$. As $\Im$ computes $f$, we know that $A^0_{V(1)} \neq A^1_{V(1)}$.  Further, from Lemma~\ref{lemma:diffProp}, we know that for any $r < \Delta_V$, we have that $W^0_{V(1)} = W^1_{V(1)}$ for any $W^0$ and $W^1$ such that $U_{0,1} \rightrightarrows^{r}_{\Gamma_0} W^0$ and $U_{1,1} \rightrightarrows^{r}_{\Gamma_1} W^1$.  Let $d_\Im$ denote the run time of $\Im$.  Then we know that $U_{0,1} \rightrightarrows^{d_\Im}_{\Gamma_0} A^0$, and $U_{1,1} \rightrightarrows^{d_\Im}_{\Gamma_1} A^1$ by the definition of run time.  If $d_\Im < \Delta_V$, then Lemma~\ref{lemma:diffProp} implies that that $A^0_{V(1)} = A^1_{V(1)}$, which contradicts the fact that $\Im$ compute $f$.  Therefore, the run time $d_\Im$ is at least $\Delta_V \geq \frac{1}{2} \Delta$.

\subsection{Proof of Lemma \ref{lemma:pointCrowding}} \label{proof:pointCrowding}
To see this, consider a bounding box of all $2n$ points.  If all $d$ dimensions of the bounding box were of length strictly less that $\sqrt[d]{2n}$, then the box could not contain all $2n$ points.  Therefore, at least one dimension is of length at least $\lceil\sqrt[d]{2n}\rceil$, implying that there are two points of distance at least $\lceil \sqrt[d]{2n} \rceil-1$ along that particular axis.  If these two points are in $A$ and $B$ respectively, then the claim follows.  If not, say both are from set $A$, then there must be a point in $B$ that is at least $\lceil \frac{1}{2} \lceil\sqrt[d]{2n}\rceil \rceil - 1$ from one of these two points in $A$, implying the claim.

\subsection{Proof of Theorem \ref{thm:addition}} \label{proof:addition}
To show the lower bound, we will reduce the $\Delta$-communication problem for some $\Delta = \Omega(\sqrt[d]{n})$ to the $n$-bit adder problem with a $d$-dimension template. Consider some $n$-bit adder TAC $\Im = (T,U=(F,W),V,\tau)$ such that $U$ is a $d$-dimension template.  The $2n$ sequence of wildcard positions $W$ of this TAC must be contained in $d$-dimensional space by the definition of a $d$-dimension template, and therefore by Lemma~\ref{lemma:pointCrowding} there must exist points $W(i)$ for $i \leq n$, and $W(n+j)$ for $j \leq n$, such that $\Delta_{W(i), W(n+j)} \geq \lceil\frac{1}{2}\lceil\sqrt[d]{2n}\rceil \rceil - 1 = \Omega(\sqrt[d]{n})$.  Now consider two $n$-bit inputs $a=a_n \ldots a_1$ and $b=b_n \ldots b_1$ to the adder TAC $\Im$ such that: $a_k = 0$ for any $k>i$ and any $k < j$, and $a_k = 1$ for any $k$ such that $j\leq k < i$.  Further, let $b_k = 0$ for all $k \neq j$.  The remaining bits $a_i$ and $a_j$ are unassigned variables of value either $0$ or $1$.  Note that the $i+1$ bit of $a + b$ is $1$ if and only if $a_i$ and $b_j$ are both value $1$.  This setup constitutes our reduction of the $\Delta$-communication problem to the addition problem as the adder TAC template with the specified bits hardcoded in constitutes a template for the $\Delta$-communication problem that produces the AND of the input bit pair.  We now specify explicitly how to generate a communication TAC from a given adder TAC.

For given $n$-bit adder TAC $\Im = (T,U=(F,W),V,\tau)$ with dimension $d$ input template, we derive a $\Delta$-communication TAC $\rho = (T, U^2 = (F^2, W^2), V^2,\tau)$ as follows.  First, let $W^2(1) = W(i)$, and $W^2(2) = W(n + j)$.  Note that as $\Delta_{W(i), W(n+j)} =\Omega(\sqrt[d]{n})$, $W^2$ satisfies the requirements for a $\Delta$-communication input template for some $\Delta=\Omega(\sqrt[d]{n})$.  Derive the frame of the template $F^2$ from $F$ by adding tiles to $F$ as follows:  For any positive integer $k > i$, or $k < j$, or $k>n$ but not $k=n+j$, add a translation of $t_0$ (with label ``0") translated to position $W(k)$.  Additionally, for any $k$ such that $j \leq k < i$, add a translation of $t_1$ (with label ``1") at translation $W(k)$.

Now consider the $\Delta$-communication TAC $\rho = (T, U^2=(F^2, W^2), V^2, \tau)$ for some $\Delta = \Omega(\sqrt[d]{n})$.  As assembly $U^2_{a_i, b_j} = U_{a_1\ldots a_n , b_1 \ldots b_n}$, we know that the worst case run time of $\rho$ is at most that of the worst case run time of $\Im$.  Therefore, by Theorem~\ref{thm:communicationLower}, we have that $\Im$ has a run time of at least $\Omega(\sqrt[d]{n})$.

\subsection{Proof of Theorem \ref{thm:multiplication}} \label{proof:multiplication}
Consider some $n$-bit multiplier TAC $\Im = (T, U=(F,W), V, \tau)$ with $d$-dimension input template.  By Lemma~\ref{lemma:pointCrowding}, some $W(i)$ and $W(n+j)$ must have distance at least $\Delta \geq \lceil\frac{1}{2}\lceil\sqrt[d]{2n}\rceil \rceil - 1$.  Now consider input strings $a=a_n \ldots a_1$ and $b=b_n \ldots b_1$ to $\Im$ such that $a_i$ and $b_j$ are of variable value, and all other $a_k$ and $b_k$ have value 0.  For such input strings, the $i+j$ bit of the product $ab$ has value 1 if and only if $a_i = b_j = 1$.  Thus, we can convert the $n$-bit multiplier system $\Im$ into a $\Delta$-communication TAC with the same worst case run time in the same fashion as for Theorem~\ref{thm:addition}, yielding a $\Omega(\sqrt[d]{n})$ lower bound for the worst case run time of $\Im$.

\section{Average Case $O(\log{n})$ Time Addition.} \label{app:logn}
\subsection{Construction} \label{sec:avgcase_construction}

\begin{figure}[htp]
	\centering
	\includegraphics[scale=1.2]{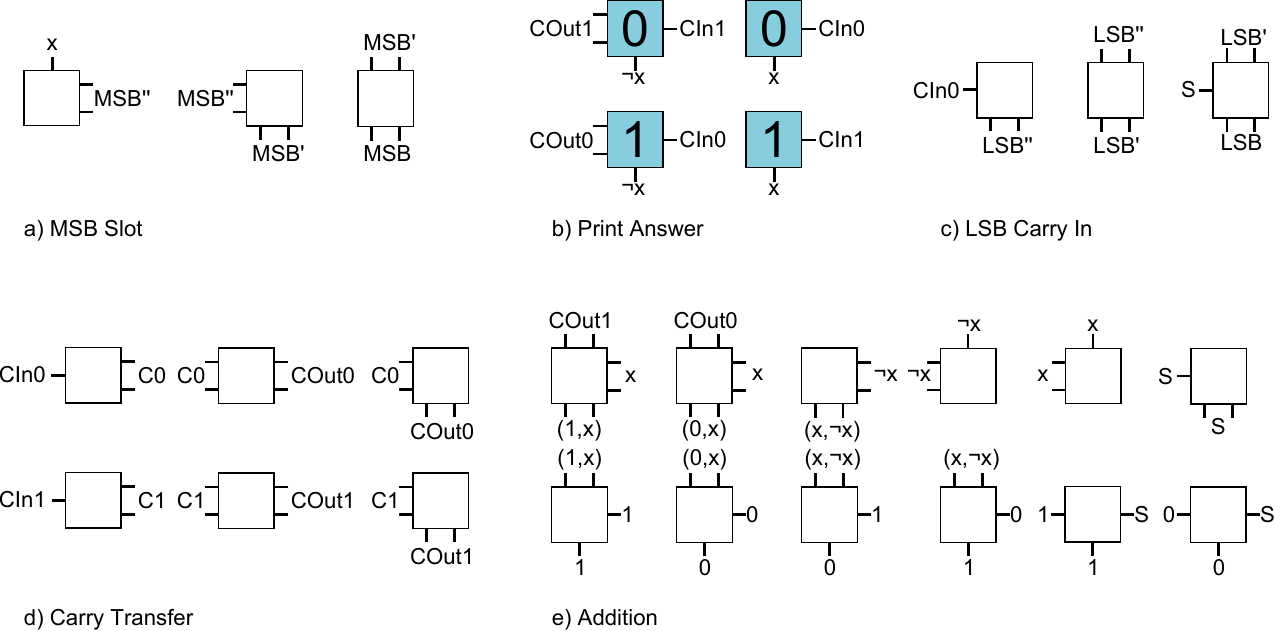}
	\caption{This is the complete set of tiles necessary to implement average case $O(\log{n})$ and worst case $O(n)$ addition.}
	\label{fig:avg_tile_set}
\end{figure}

We summarize the mechanism of addition presented here in a short example. The complete tile set may be found in Figure \ref{fig:avg_tile_set}.
\paragraph{Input Template.} The input template, or seed, for the construction of an adder with an $O(\log{n})$ average case is shown in Figure \ref{fig:LogLineSeed}.  This input template is composed of $n$ blocks, each containing three tiles.  Within a block, the easternmost tile is the $S$ labeled tile followed by two tiles representing $A_k$ and $B_k$, the $k$th bits of $A$ and $B$ respectively.  Of these $n$ blocks, the easternmost and westernmost blocks of the template assembly are unique.  Instead of an $S$ tile, the block furthest east has an $LSB$-labeled tile which accompanies the tiles representing the least significant bits of $A$ and $B$, $A_0$ and $B_0$. The westernmost block of the template assembly contains a block labeled $MSB$ instead of the $S$ block and accompanies the most significant bits of $A$ and $B$, $A_{n-1}$ and $B_{n-1}$.

\begin{figure}[htp]
	\centering
	\includegraphics[scale=1.4]{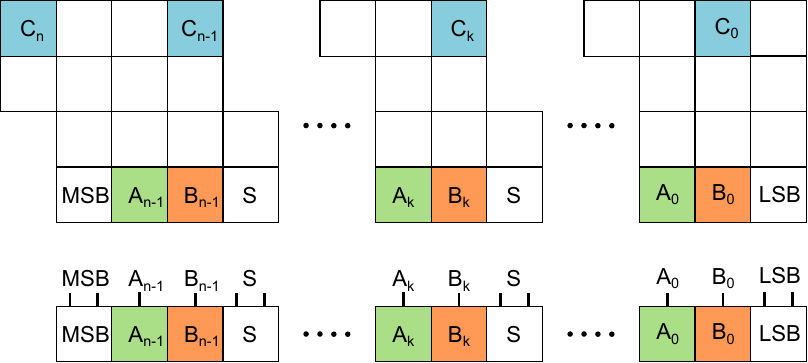}
	\caption{Top: Output template displaying addition result $C$ for $O(\log{n})$ average case addition construction. Bottom: Input template composed of $n$ blocks of three tiles each, representing $n$-bit addends $A$ and $B$. }
	\label{fig:LogLineSeed}
\end{figure}

\paragraph{Computing Carry Out Bits.}
For clarity, we demonstrate the mechanism of this adder TAC through an example by selecting two 4-bit binary numbers $A$ and $B$ such that the addends $A_i$ and $B_i$ encompass every possible addend combination. The input template for such an addition is shown in Figure \ref{fig:loglinesteps}a where orange tiles represent bits of $A$ and green tiles represent bits of $B$. Each block begins the computation in parallel at each $S$ tile.
\begin{figure}[htp]
	\centering
	\includegraphics[scale=0.90]{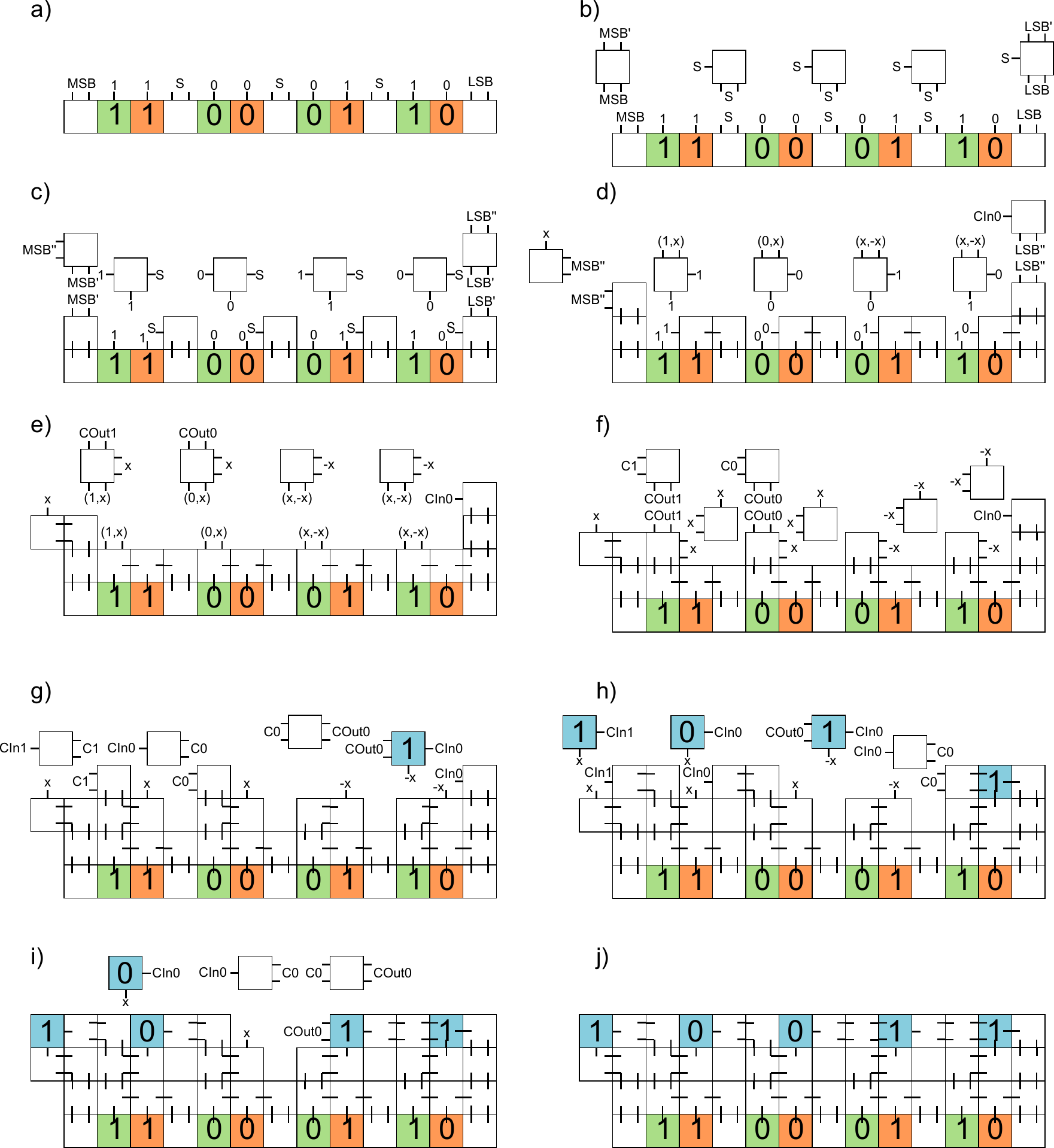}
	\caption{Example tile assembly system to compute the sum of 1001 and 1010 using the adder TAC presented in Section \ref{sec:avgcase}. a) The template with the two $4$-bit binary numbers to be summed. b) The first step is the parallel binding of tiles to the $S$, $LSB$, and $MSB$ tiles. c)Tiles bind cooperatively to west-face $S$ glues and glues representing bits of $B$. The purpose of this step is to propagate bit $B_i$ closer to $A_i$ so that in d) a tile may bind cooperatively, processing information from both $A_i$ and $B_i$. e) Note that addend-pairs consisting of either both $1$s or both $0$s have a tile with a north face glue consisting of either (1,x) or (0,x) bound to the $A_i$ tile. This glue represents a carry out of either $1$ or $0$ from the addend-pair. In (e-g) the carry outs are propagated westward via tile additions for those addend-pairs where the carry out can be determined. Otherwise, spacer tiles bind. h) Tiles representing bits of $C$ (the sum of $A$ and $B$) begin to bind where a carry in is known. i-j) As carry bits propagate through sequences of (0,1) or (1,0) addend pairs, the final sum is computed. }
	\label{fig:loglinesteps}
\end{figure}
 After six parallel steps (Figure \ref{fig:loglinesteps}b-g), all carry out bits, represented by glues $C0$ and $C1$, are determined for addend-pairs where both $A_i$ and $B_i$ are either both $0$ or both $1$. For addend pairs $A_i$ and $B_i$ where one addend is $0$ and one addend is $1$, the carry out bit cannot be deduced until a carry out bit has been produced by the previous addend pair, $A_{i-1}$ and $B_{i-1}$. By step seven, a carry bit has been presented to all addend pairs that are flanked on the east by an addend pair comprised of either both $0$s or both $1$s, or that are flanked on the east by the $LSB$ start tile, since the carry in to this site is always $0$ (Figure \ref{fig:loglinesteps}h). For those addend pairs flanked on the east by a contiguous sequence of size $j$ pairs consisting of one $1$ and one $0$, $2j$ parallel attachment steps must occur before a carry bit is presented to the pair.

\paragraph{Computing the Sum.}
Once a carry out bit has been computed and carried into an addend pair $A_i$ and $B_i$, two parallel tile addition steps are required to compute the sum of the addend pair (Figure \ref{fig:loglinesteps}g-j) .

\begin{figure}[htp]
	\centering
	\includegraphics[scale=1.0]{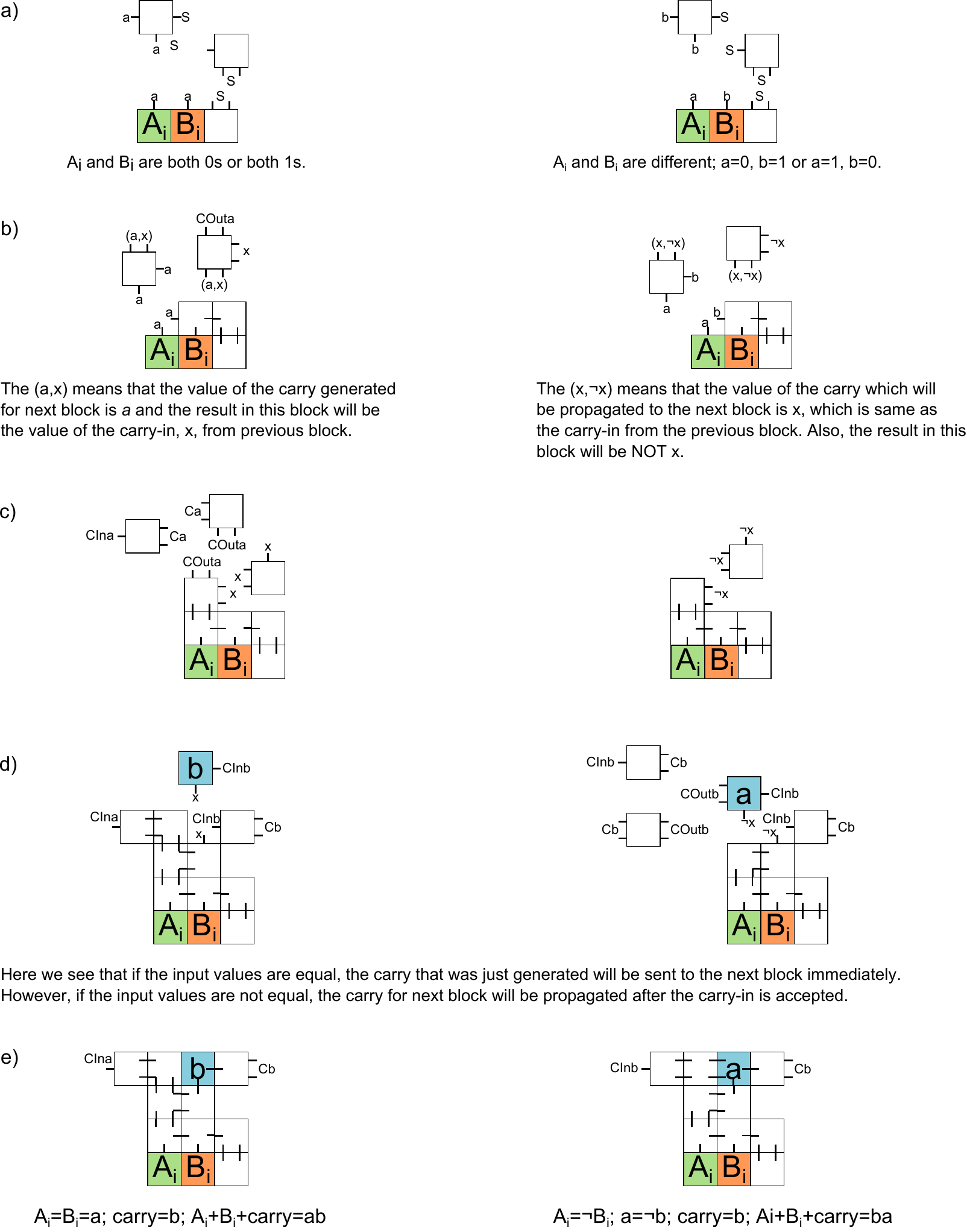}
	\label{fig:LogAddLin000}
    \caption{The set of figures on the left side show how a carry out may be generated before the \emph{addend-pair} has recieved a carry in. The set of figures on the right side show how a carry out can be dependent upon a carry in having been recieved by the \emph{addend-pair}.}
\end{figure}
\subsection{Time Complexity.} \label{sec:avgcase_time}
\subsubsection{Parallel Time}
\paragraph{$O(n)$ - worst case.}
We first show that this construction has an $O(n)$ worst case run-time under the timing model presented in Section \ref{sec:runtime} \emph{Run-time}. Consider a binary sequence $T$ of length $2n$ representing two $n$-bit binary numbers $A$ and $B$. $A_{0}$ and $B_{0}$ represent the least significant bits of $A$ and $B$, respectively, and $A_{n}$ and $B_{n}$ represent the most significant bits of $A$ and $B$, respectively.  The formatting of $T$ is such that $T_{i}=B_{i/2}$ if $i$ is even, and $T_{i}=A_{(i-1)/2}$ if $i$ is odd. Sequence $T$ is shown in Figure \ref{fig:bitpair}a.

\begin{figure}[htp]
	\centering
	\includegraphics[scale =0.8]{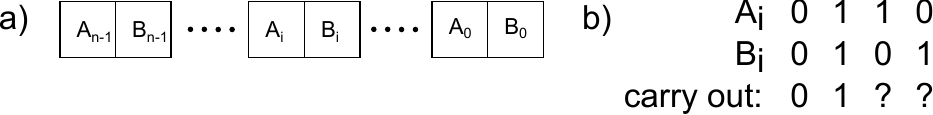}
	\caption{a) Binary sequence $T$ populated by n-bit binary numbers $A$ and $B$. b) Four possible $(A_{i}, B_{i})$ addend-pair combinations.}
	\label{fig:bitpair}
\end{figure}

$T$ contains $n$ addend-pairs, $(A_{i}, B_{i})$, which are ordered pairs consisting of the $i$th bit of $A$ and the $i$th bit of $B$. The four possible values for each addend-pair are shown in Figure \ref{fig:bitpair}b, along with the carry bits they produce upon addition.  In $T$, there exist sequences of various sizes up to $k$ addend-pairs such that every addend-pair in the interval from $(A_{i}, B_{i})$ to $(A_{i+j-1}, B_{i+j-1})$ of a size $j$ addend-pair sequence is $(1,0)$ or $(0,1)$. In the adder TAC outlined in Section \ref{sec:avgcase} and Appendix \ref{sec:avgcase_construction}, the value of the carry bit produced upon addition of $A_{i}+B_{i}$ for every $(0,0)$ and $(1,1)$ addend-pair is known and made available to the next addend-pair $(A_{i+1}, B_{i+1})$ after seven parallel tile addition steps, including the carry bit into $(A_{0}, B_{0})$. Therefore, after a constant number of parallel tile addition steps, the first addend-pair $(A_{i}, B_{i})$ in any $j$ addend-pair-length sequence of $(1,0)$ or $(0,1)$ addend-pairs will be presented with the carry bit from the previous addend-pair $(A_{i-1}, B_{i-1})$. After $2j$ subsequent parallel tile addition steps, the carry bit from the last addend-pair $(A_{i+j-1}, B_{i+j-1})$ of the $j$-size sequence of consecutive $(0,1)$ and $(1,0)$ addend-pairs is presented to $(A_{i+j}, B_{i+j})$. Once an addend-pair has recieved a carry-in bit, the final sum is computed in one tile addition step. If $k$ is the longest contiguous sequence of $(0,1)$ and $(1,0)$ addend-pairs, then $2k+8$ parallel tile addition steps are required to compute the sum $C$ of $A$ and $B$ (Figure \ref{fig:upperbound_example}). Therefore, the time complexity is $O(k)$.
\begin{figure}[htp]
	\centering
	\includegraphics[scale = 1.0]{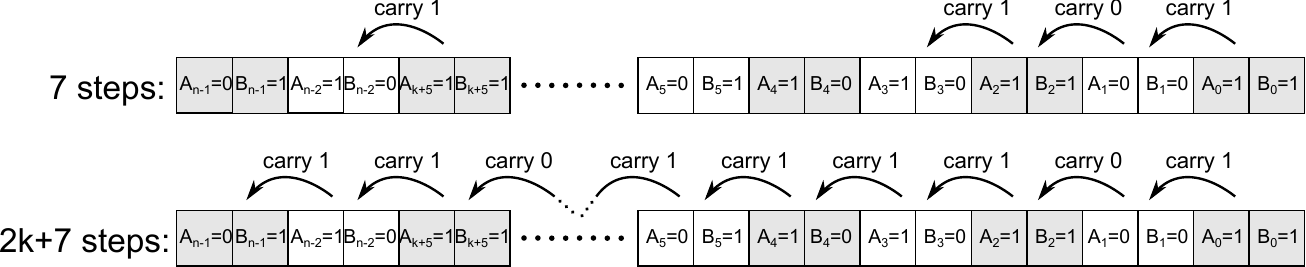}
	\caption{}
	\label{fig:upperbound_example}
\end{figure}
Since the growth is bounded upwards by the longest contiguous sequence $k$ of $(0,1)$ and $(1,0)$ addend-pairs, then the worst-case scenario occurs when $k=n$. Thus, the worst-case run time is $O(n)$.

\paragraph{$O(\log{n})$ - average case.}
We now show that the average case run-time is $O(\log{n})$ under the timing model presented in Section \ref{sec:runtime}. In two $n$-bit randomly generated binary numbers, $A$ and $B$, the probability of one of the $(1,0)$ $(0,1)$ addend-pair cases occurring at $(A_{i},B_{i})$ is $1/2$. The sequence of $(A_{i},B_{i})$ bit pairs can thus be thought of as a Bernoulli process in which the likelihood of occurrence of a $(0,1)$ or $(1,0)$ addend-pair is equal to the occurrence of a $(1,1)$ or $(0,0)$ addend-pair. As shown above, the runtime is bounded by $k$, the longest contiguous sequence of $(0,1)$ or $(1,0)$ addend-pairs, which might be thought of as the longest sequence of heads in $n$ independent fair coin tosses.  Using Lemma \ref{lem:longest_substring}, the expected longest run of heads in $n$ coin tosses is $O(\log{n})$~\cite{SCHILL1990}.  Therefore, the average case time complexity of the tile addition algorithm described above is $O(\log{n})$.

\subsubsection{Continuous Time}\label{sec:logadditioncontinuous}
To analyze the continuous run time of our TAC we first introduce a Lemma \ref{lemma:maxExpSums}.  This Lemma is essentially Lemma 5.2 of~\cite{woods2012EAS} but slightly generalized and modified to be a statement about exponentially distributed random variables rather than a statement specific to the self-assembly model considered in that paper.  To derive this Lemma, we make use of a well known Chernoff bound for the sum of exponentially distributed random variables stated in Lemma~\ref{lemma:chernoff}.

\begin{lemma}\label{lemma:chernoff}
Let $X = \sum_{i=1}^{n} X_i$ be a random variable denoting the sum of $n$ independent exponentially distributed random variables each with respective rate $\lambda_i$. Let $\lambda = \min\{\lambda_i | 1\leq i \leq n\}$.  Then $Pr[ X > (1+\delta)n/\lambda ] < (\frac{\delta+1}{e^\delta})^n$.
\end{lemma}

\begin{lemma}\label{lemma:maxExpSums}
Let $t_1 ,\ldots , t_m$ denote $m$ random variables where each $t_i = \sum_{j=1}^{k_i}t_{i,j}$ for some positive integer $k_i$, and the variables $t_{i,j}$ are independent exponentially distributed random variable each with respective rate $\lambda_{i,j}$.  Let $k= \max(k_1 , \ldots k_m)$, $\lambda = \min\{\lambda_{i,j} | 1\leq i\leq m, 1\leq j \leq k_i\}$, $\lambda' = \max\{\lambda_{i,j} | 1\leq i\leq m, 1\leq j \leq k_i\}$, and $T=\max(t_1 ,\ldots t_m)$.  Then $E[T]=O(\frac{k + \log m}{\lambda})$ and $E[T]=\Omega(\frac{k+\log m}{\lambda'})$.
\end{lemma}
\begin{proof}
First we show that $E[T]=O(\frac{k + \log m}{\lambda})$.  For each $t_i$ we know that $E[t_i] = \sum_{j=1}^{k_i}\lambda_{i,j} \leq k/\lambda$.  Applying Lemma~\ref{lemma:chernoff} we get that:

$$
Pr[t_i > k/\lambda (1+\delta)] < (\frac{\delta+1}{e^\delta})^k.
$$

Applying the union bound we get the following for all $\delta \geq 3$:
$$
Pr[T > \frac{(1+\delta)k}{\lambda}] \leq m(\frac{\delta+1}{e^\delta})^k \leq m e^{-k\delta/2}, \text{ for all }\delta > 3.
$$

Let $\delta_m = \frac{2 \ln m}{k}$.  By plugging in $\delta_m + x$ for $\delta$ in the above inequality we get that:

$$ Pr[T > \frac{(1 +\delta_m + x)k}{\lambda}] \leq e^{-kx/2}.$$

Therefore we know that:

$$ E[T] \leq \frac{k+2\ln m}{\lambda} + \int_{x=0}^{\infty} e^{-kx/2} = O(\frac{k +\log m}{\lambda}).$$

To show that $E[T]=\Omega(\frac{k+\log m}{\lambda'})$, first note that $T \geq \max_{1\leq i\leq m} \{t_{i,1}\}$ and that $E[\max_{1\leq i\leq m}\{t_{i,1}\}] = \Omega(\frac{\log m}{\lambda'})$, implying that $E[T] = \Omega(\frac{\log m}{\lambda'})$.  Next, observe that $T \geq t_i$ for each $i$.  Since $E[t_i]$ is at least $k/\lambda'$ for at least one choice of $i$, we therefore get that $E[T] = \Omega(k/\lambda')$.

\end{proof}

The next Lemma helps us upper bound continuous run times by stating that if the assembly process is broken up into separate phases, denoted by a sequence of subassembly \emph{waypoints} that must be reached, we may simply analyze the expected time to complete each phase and sum the total for an upper bound on the total run time.  The following Lemma follows easily from the fact that a given open tile position of an assembly is filled at a rate at least as high as any smaller subassembly.

\begin{lemma}\label{lemma:waypoints}
Let $\Gamma=(T,S,\tau)$ be a deterministic tile assembly system and $A_0\ldots A_r$ be elements of $PROD_\Gamma$ such that $A_0 = S$, $A_r$ is the unique terminal assembly of $\Gamma$, and $A_0 \subseteq A_2 \subseteq \ldots A_r$.  Let $t_i$ be a random variable denoting the time for $\Gamma_{MC}$ to transition from $A_{i-1}$ to the first $A'_{i}$ such that $A_{i} \subseteq A'_i$ for $i$ from 1 to $r$.  Then $\varsigma_\Gamma \leq \sum_{i=1}^r E[t_i]$.
\end{lemma}

We now bound the run time of our adder TAC by breaking the assembly process up according to waypoint assemblies $A_0=S,A_1,$ and $A_2$.  We then bound the expected transition time from $A_0$ to $A_1$ and $A_1$ to $A_2$ to get a bound from Lemma~\ref{lemma:waypoints}.

Let producible assembly $A_1$ be the seed line with the additional 4 or 5 tiles (dependant on the type of bit pair) placed atop each bit pair as shown in Figure~\ref{fig:loglinesteps} (g).  As each collection of 4 or 5 tiles can assemble independently, the time $t_1$ to transition from the seed to a superset of $A_1$ is bounded by the maximum of $n$ sums of at most 5 exponentially distributed random variables with rate $1/|T|$, where $T$ is the tileset for the adder TAC.  By Lemma~\ref{lemma:maxExpSums}, $E[t_1]= O(\log n)$ time for any input pair of $n$-bit numbers.

Now consider the time $t_2$ to transition from $A_1$ to the unique terminal assembly of the system, i.e, the completion of the third row of the assembly. For this phase we are interested in the time taken for the last finishing chain of tile placements for each maximal contiguous sequence of $(0,1) - (1,0)$ addend pairs.  If the largest of these sequences is length $k$, then $t_2$ is bounded by the maximum of at most $n$ sums of at most $3k$ exponentially distributed random variables with rate $1/|T|$.  Lemma~\ref{lemma:maxExpSums} therefore implies $E[t_2] = O(k)$.  As $k \leq n$ we get a worst case time of $E[t_2]=O(n)$.  Further, by Lemma~\ref{lem:longest_substring}, we have that the average value of $k$ is $O(\log n)$, yielding an average case time of $E[t_2]=O(\log n)$.  Applying Lemma~\ref{lemma:waypoints}, we get an upper bound for the total runtime of our TAC of $E[t_1]+E[t_2]$ which is $O(n)$ in the worst case and $O(\log n)$ in the average case.

Our $O(\log n)$ average case is the best that can be hoped for in the continuous time model as any adder TAC must place at least $n+1$ tiles to guarantee correctness.  From Lemma~\ref{lemma:maxExpSums}, this immediately gives an $\Omega(\log n)$ lower bound in all cases:

\begin{theorem}
The continuous run time for any $n$-bit adder TAC is $\Omega(\log n)$ in all cases.
\end{theorem}

\subsection{Correctness.} \label{sec:avgcase_correct}

\begin{figure}[htp]
	\centering
	\includegraphics[scale = 0.85]{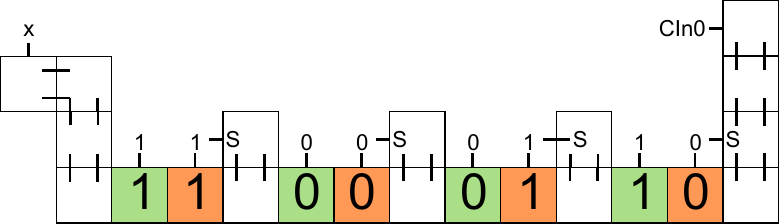}
	\caption{The scaffolding needed for the computation.}
	\label{fig:avg_case_scaffolding}
\end{figure}

The analysis for the average case $O(\log{n})$ worst case $O(n)$ adder TAC will begin from the point of the construction after which all scaffolding is in place (Figure \ref{fig:avg_case_scaffolding}).

The first step for any addend-pair in the seed is to perform its addition.  This addition step will output either $(1, x)$, $(0, x)$, or $(x, \lnot x)$ according to the following formulas:
$$1 + 1 = (1, x)$$
$$0 + 0 = (0, x)$$
$$1 + 0 = (x, \lnot x)$$
$$0 + 1 = (x, \lnot x)$$
%%%Notes about the previous sentance: What does this mean? Also, k may be a bad variable as it was used differently in previous section (4.2). %%%
The first value in the output represents the carry-out of the addend-pair and the second value represents the value of the addition.  In both cases, $x$ refers to the unknown carry-in which will come from the previous addend-pair's carry out.

Given any addend-pair, we can then determine whether or not we have enough information to generate a carry by looking at the first position of the generated output pair in the first step.  The two cases that can generate a carry, $(1, x)$ and $(0, x)$, do just that and immediately propagate a $1$ or $0$ respectively. The other case, $(x, \lnot x)$, simply waits for a carry to come in.  No addend-pair can calculate its value until it receives a carry from the previous addend-pair.  Once an addend-pair receives a carry-in, it can replace any $x$ or $\lnot x$ with the proper value and it can ``print" the correct value, i.e. the solution at that particular position.  Also, if $k_0$ was $x$ it can now propagate its carry to the next addend-pair.

\section{$O(\sqrt{n})$ Time Addition.} \label{app:sqrtn}
\subsection{Construction.} \label{sec:worst_construction}

\begin{figure}[htp]
	\centering
	\includegraphics[scale=0.8]{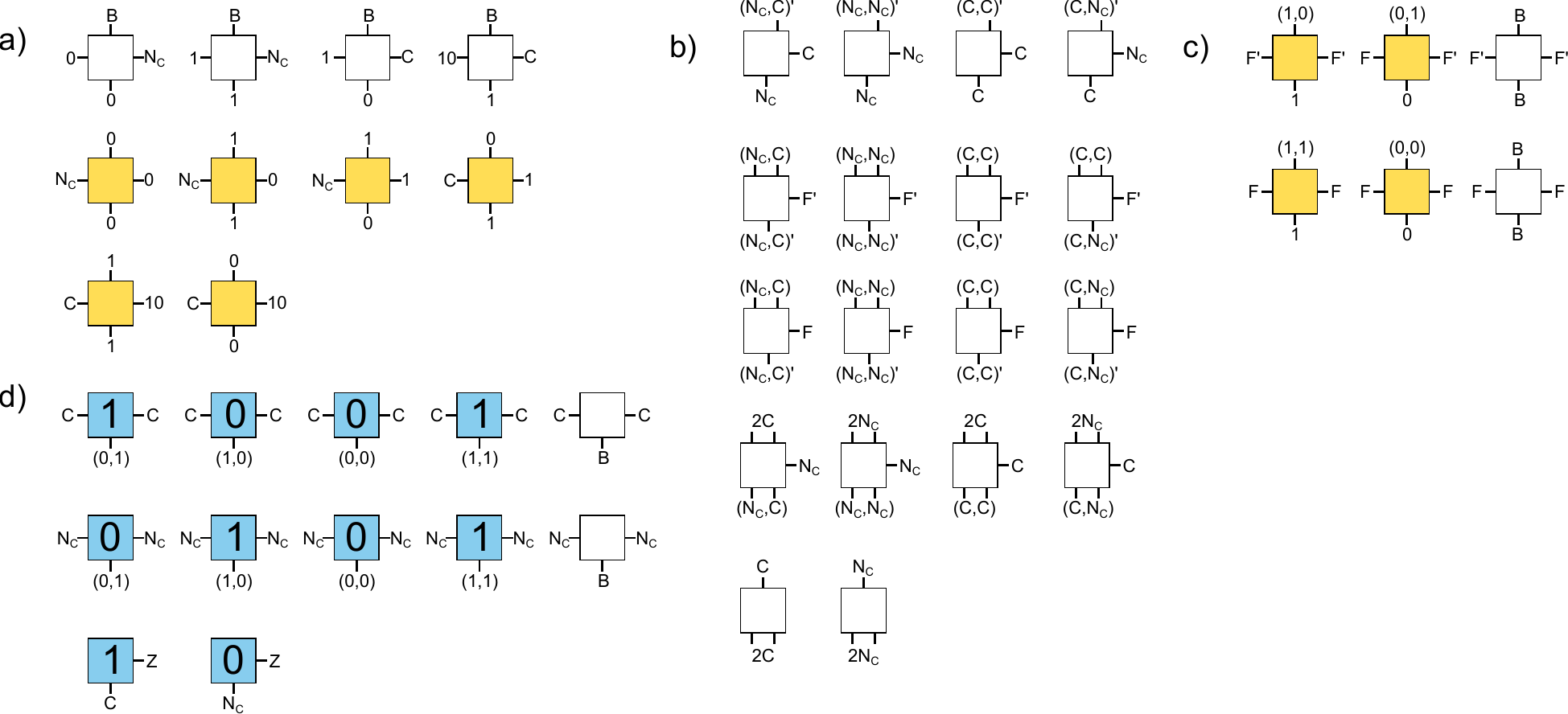}
	\caption{The complete tile set. a) Tiles involved in bit addition. b) Carry propagation tiles. c) Tiles involved in incrementation. d) Tiles that print the answer.}
	\label{fig:worsttileset}
\end{figure}

\paragraph{Input/Output Template.}
\begin{figure}[htp]
	\centering
	\subfigure[Addition input template.]{
		\includegraphics[scale=0.85]{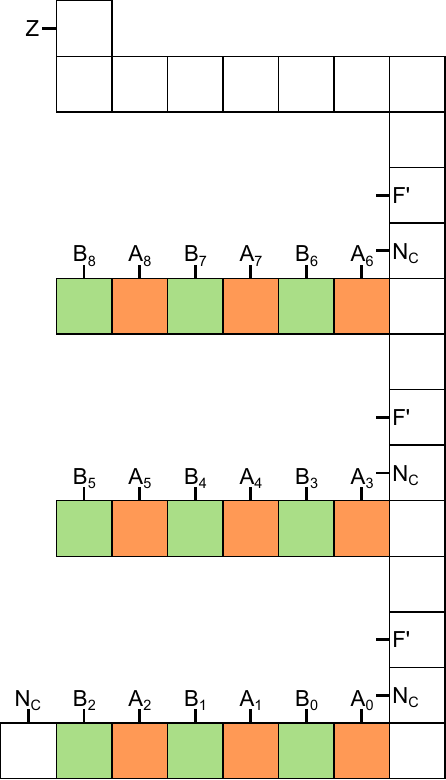}
		\label{fig:addition_input_template}
	}
	\subfigure[Addition output template.]{
		\includegraphics[scale=0.85]{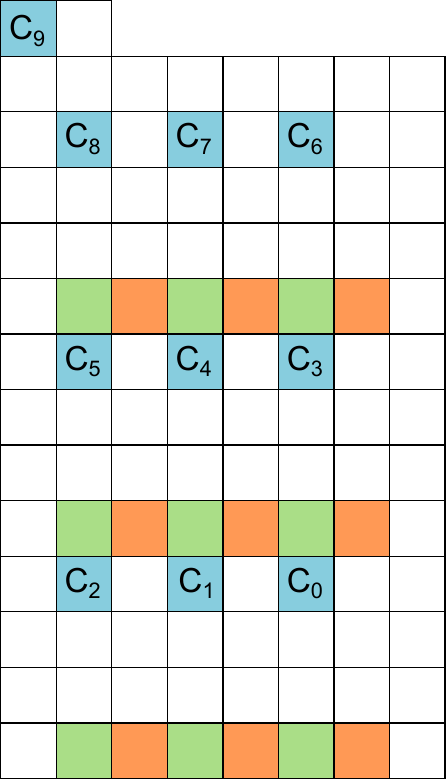}
		\label{fig:addition_output_template}
	}
\caption{These are example I/O templates for the worst case $O(\sqrt{n})$ time addition introduced in Section \ref{sec:worstcase}.}
\end{figure}

Figure \ref{fig:addition_input_template} and Figure \ref{fig:addition_output_template} are examples of I/O templates for a $9$-bit adder TAC. The inputs to the addition problem in this instance are two 9-bit binary numbers $A$ and $B$ with the least significant bit of $A$ and $B$ represented by $A_{0}$ and $B_{0}$, respectively.  The north facing glues in the form of $A_i$ or $B_i$ in the input template must either be a $1$ or a $0$ depending on the value of the bit in $A_i$ or $B_i$. The placement for these tiles is shown in Figure \ref{fig:addition_input_template} while a specific example of a possible input template is shown in Figure \ref{fig:worstcase_addition}a. The sum of $A+B$, $C$, is a ten bit binary number where $C_{0}$ represents the least significant bit. The placement for the tiles representing the result of the addition is shown in Figure \ref{fig:addition_output_template} while a specific example of an output is shown in Figure \ref{fig:worstcase_carryprop}j.

To construct an $n$-bit adder in the case that $n$ is a perfect square, split the two $n$-bit numbers into $\sqrt{n}$ sections each with $\sqrt{n}$ bits.  Place the bits for each of these two numbers as per the previous paragraph, except with $\sqrt{n}$ bits per row, making sure to alternate between $A$ and $B$ bits.  There will be the same amount of space between each row as seen in the example template \ref{fig:addition_input_template}. All $Z$, $N_C$, and $F'$, must be placed in the same relative locations. The solution, $C$, will be in the output template s.t. $C_i$ will be three tile positions above $B_i$ and a total of size $n + 1$.

Below, we use the adder tile set to add two nine-bit numbers: $A=100110101$ and $B=110101100$ to demonstrate the three stages in which the adder tile system performs addition.

\paragraph{Step One: Addition.}
With the inclusion of the seed assembly (Figure \ref{fig:worstcase_addition}a) to the tile set (Figure \ref{fig:worsttileset}), the first subset of tiles able to bind are the addition tiles shown in Figure \ref{fig:worsttileset}a. These tiles sum each pair of bits from $A$ and $B$ (for example, $A_{0}+B_{0}$) (Figure \ref{fig:worstcase_addition}b). Tiles shown in yellow are actively involved in adding and output the sum of each bit pair on the north face glue label. Yellow tiles also output a carry to the next more significant bit pair, if one is needed, as a west face glue label.  Spacer tiles (white) output a $B$ glue on the north face and serve only to propagate carry information from one set of $A$ and $B$ bits to the next. Each row computes this addition step independently and outputs a carry or no-carry west face glue on the westernmost tile of each row (Figure \ref{fig:worstcase_addition}c). In a later step, this carry or no-carry information will be propagated northwards from the southernmost row in order to determine the sum. Note that immediately after the first addition tile is added to a row of the seed assembly, a second layer may form by the attachment of tiles from the increment tile set (Figure \ref{fig:worsttileset})c. While these two layers may form nearly concurrently, we separate them in this example for clarity and instead address the formation of the second layer of tiles in \emph{Step Two: Increment} below.

\begin{figure}[htp]
	\centering
	\includegraphics[scale=0.85]{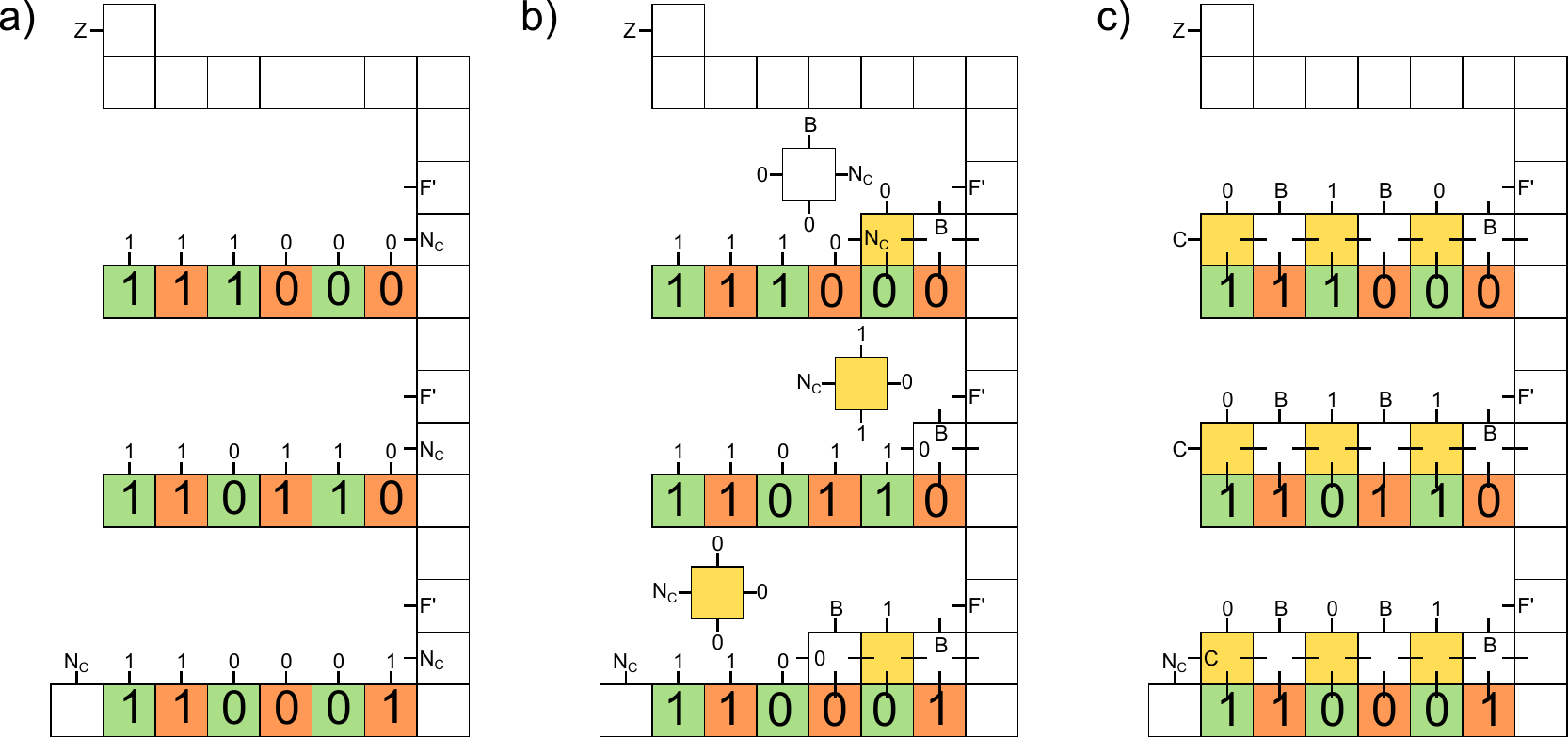}
	\caption{Step 1: Addition.}
	\label{fig:worstcase_addition}
\end{figure}

\paragraph{Step Two: Increment.}
As the addition tiles bind to the seed, tiles from the incrementation tile set (Figure \ref{fig:worsttileset}c) may also begin to cooperatively attach. For clarity, we show their attachment following the completion of the addition layer. The purpose of the incrementation tiles is to determine the sum for each $A$ and $B$ bit pair in the event of a no-carry from the row below and in the event of a carry from the row below (Figure \ref{fig:worstcase_increment}).The two possibilities for each bit pair are presented  as north facing glues on yellow increment tiles. These north face glues are of the form $(x,y)$ where $x$ represents the value of the sum in the event of no-carry from the row below while $y$ represents the value of the sum in the event of a carry from the row below. White incrementation tiles are used as spacers, with the sole purpose of passing along carry or no-carry information via their east/west face glues $F'$, which represents a no-carry, and $F$, which represents a carry.

\begin{figure}[htp]
	\centering
	\includegraphics[scale=0.85]{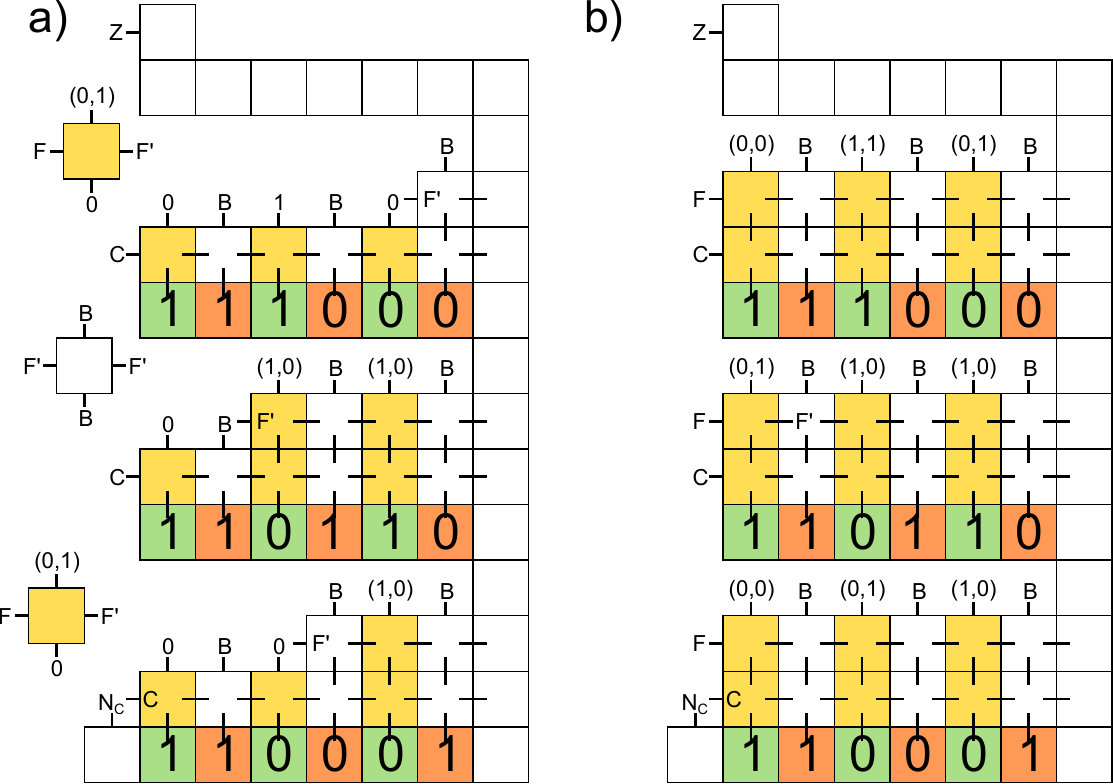}
	\caption{Step 2: Increment.}
	\label{fig:worstcase_increment}
\end{figure}

\paragraph{Step Three: Carry Propagation and Output.}
The final step of the addition mechanism presented here propagates carry or no-carry information northwards from the southernmost row of the assembly using tiles from the tile set in Figure \ref{fig:worsttileset}b and then outputs the answer using the tile set in Figure \ref{fig:worsttileset}d. Following completion of the incrementation layers, tiles may begin to grow up the west side of the assembly as shown in Figure \ref{fig:worstcase_carryprop}a. When the tiles grow to a height such that the empty space above the increment row is presented with a carry or no-carry as in Figure \ref{fig:worstcase_carryprop}b, the output tiles may begin to attach from west to east to print the answer (Figure \ref{fig:worstcase_carryprop}c). As the carry propagation column grows northwards and presents carry or no carry information to each empty space above each increment layer, the sum may be printed for each row Figures \ref{fig:worstcase_carryprop}d-e. When the carry propagation column reaches the top of the assembly, the most significant bit of the sum may be determined and the calculation is complete (Figure \ref{fig:worstcase_carryprop}f).

\begin{figure}[htp]
	\centering
	\includegraphics[scale=0.85]{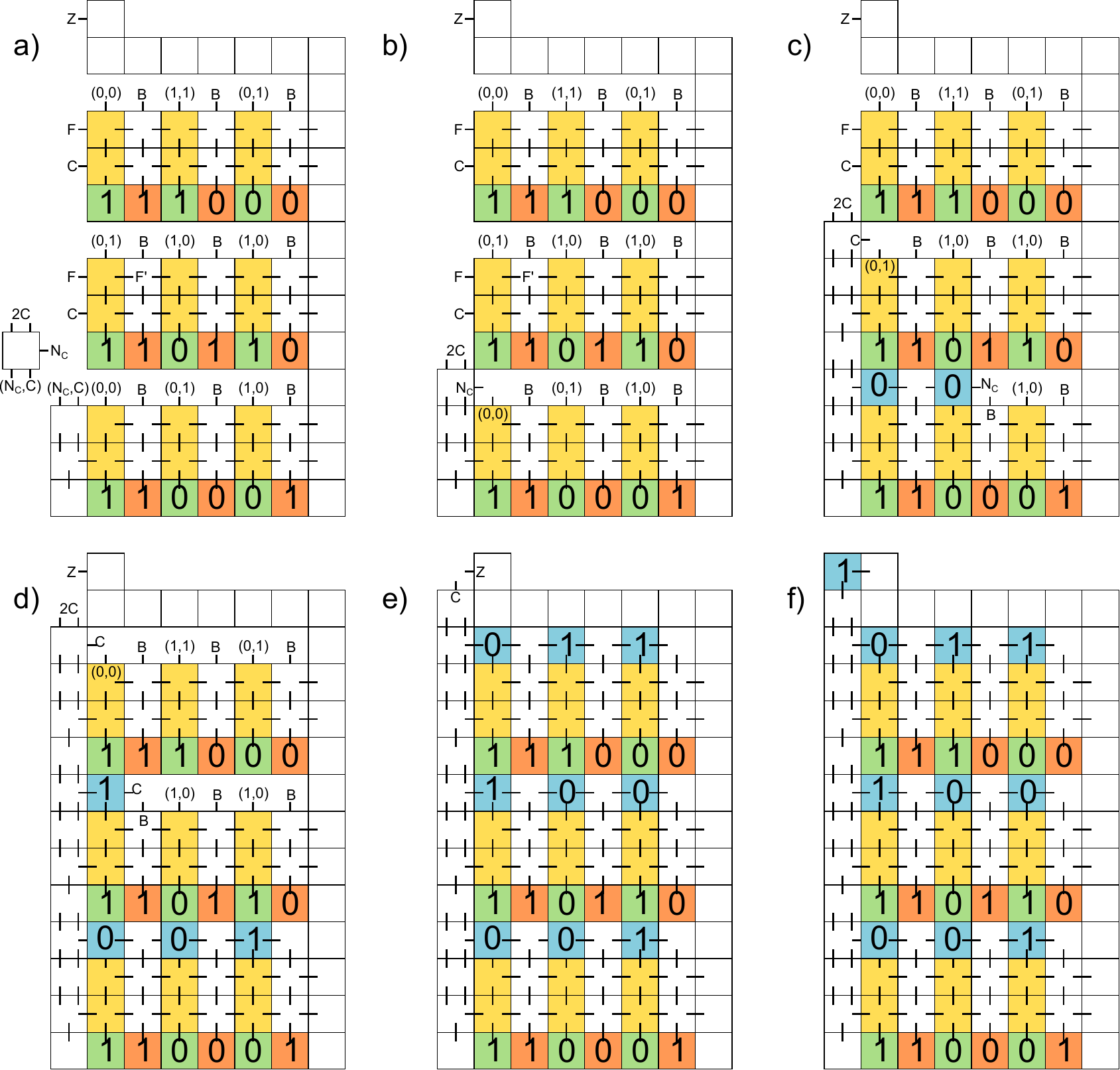}
	\caption{Step 3: Carry Propagation and Output.}
	\label{fig:worstcase_carryprop}
\end{figure}
\subsection{Time Complexity.} \label{sec:worst_time}

\subsubsection{Parallel Time}
Using the runtime model presented in Section \ref{sec:runtime} \emph{Run-time} we will show that the addition algorithm presented in this section has a worst case runtime of $O(\sqrt{n})$.  In order to ease the analysis we will assume that each logical step of the algorithm happens in a synchronized fashion even though parts of the algorithm are running concurrently.

The first step of the algorithm is the addition of two $O(\sqrt{n})$ numbers co-located on the same row.  This first step occurs by way of a linear growth starting from the leftmost bit all the way to the rightmost bit of the row.  The growth of a line one tile at a time has a runtime on the order of the length of the line.  In the case of our algorithm, the row is of size $O(\sqrt{n})$ and so the runtime for each row is $O(\sqrt{n})$. The addition of each of the $\sqrt{n}$ rows happens independently, in parallel, leading to a $O(\sqrt{n})$ runtime for all rows.
Next, we increment each solution in each row of the addition step, keeping both the new and old values.  As with the first step, each row can be completed independently in parallel by way of a linear growth across the row leading to a total runtime of $O(\sqrt{n})$ for this step.
After we increment our current working solutions we must both generate and propagate the carries for each row. In our algorithm, this simply involves growing a line across the leftmost wall of the rows.  The size of the wall is bounded by $O(\sqrt{n})$ and so this step takes $O(\sqrt{n})$ time.
Finally, in order to output the result bits into their proper places we simply grow a line atop the line created by the increment step.  This step has the same runtime properties as the addition and increment steps. Therefore, the output step has a runtime of $O(\sqrt{n})$ to output all rows.

There are four steps each taking $O(\sqrt{n})$ time leading to a total runtime of $O(\sqrt{n})$ for this algorithm.  This upper bound meets the lower bound presented in Corollary \ref{cor:addition} and the algorithm is therefore optimal.

\paragraph{Choice of $\sqrt{n}$ rows of $\sqrt{n}$ size.}  The choice for dividing the bits up into a $\sqrt{n} \times \sqrt{n}$ grid is straightforward.  Imagine that instead of using $O(\sqrt{n})$ bits per row, a much smaller growing function such as $O(\log{n})$ bits per row is used.  Then, each row would finish in $O(\log{n})$ time.  After each row finishes, we would have to propagate the carries.  The length of the west wall would no longer be bound by the slow growing function $O(\sqrt{n})$ but would now be bound by the much faster growing function $O(\frac{n}{\log{n}})$.  Therefore, there is a distinct trade off between the time necessary to add each row and the time necessary to propagate the carry with this scheme.  The runtime of this algorithm can be viewed as the $max(row_{size}, westwall_{size})$.  The best way to minimize this function is to divide the rows such that we have the same number of rows as columns, i.e. the smallest partition into the smallest sets.  The best way to partition the bits is therefore into $\sqrt{n}$ rows of $\sqrt{n}$ bits.

\subsubsection{Continuous Time}\label{sec:worstcasecontinuous}
To bound the continuous run time of our TAC for any pair of length $n$-bit input strings we break the assembly process up into 4 phases.  Let the phases be defined by 5 producible assemblies $S=A_0, A_1, A_2, A_3, A_4 = F$, where $F$ is the final terminal assembly of our TAC.  Let $t_i$ denote a random variable representing the time taken to grow from $A_{i-1}$ to the first assembly $A'$ such that $A'$ is a superset of $A_{i}$.  Note that $E[t_1] + E[t_2] + E[t_3] + E[t_4]$ is an upper bound for the continuous time of our TAC according to Lemma~\ref{lemma:waypoints}.  We now specify the four phase assemblies and bound the time to transition between each phase to achieve a bound on the continuous run time of our TAC for any input.

Let $A_1$ be the seed assembly plus the attachment of one layer of tiles above each of the input bits (see Figure~\ref{fig:worstcase_addition} (c) for an example assembly).  Let $A_2$ be assembly $A_1$ with the second layer of tiles above the input bits placed (see Figure~\ref{fig:worstcase_increment} (b)).  Let $A_3$ be assembly $A_2$ with the added  vertical chain of tiles at the western most edge of the completed assembly, all the way to the top ending with the final blue output tile.  Finally, let $A_4$ be the final terminal assembly of the system, which is the assembly $A_3$ with the added third layer of tiles attached above the input bits.

Time $t_1$ is the maximum of time taken to complete the attachment of the 1st row of tiles above input bits for each of the $\sqrt{n}$ rows.  As each row can begin independently, and each row grows from right to left, with each tile position waiting for the previous tile position, each row completes in time dictated by a random variable which is the sum of $\sqrt{n}$ independent exponentially distributed random variables with rate $1/|T|$, where $T$ is the tileset of our adder TAC.  Therefore, according to Lemma~\ref{lemma:maxExpSums}, we have that $t_1 = \Theta(\sqrt{n} + \log\sqrt{n}) = \Theta(\sqrt{n})$.  The same analysis applies to get $t_2 = \Theta(\sqrt{n})$.  The value $t_3$ is simply the value of the sum of $4\sqrt{n} + 2$ exponentially distributed random variables and thus has expected value $\Theta(\sqrt{n})$.  Finally, $t_4 = \Theta(\sqrt{n})$ by the same analysis given for $t_1$ and $t_2$.  Therefore, by Lemma~\ref{lemma:waypoints}, our TAC has worst case $O(\sqrt{n})$ continuous run time.

\subsection{Correctness.} \label{sec:worst_correct}
The first two steps of the algorithm are addition followed by incrementation.  It is important to note that this incrementation step not only outputs both the original (addition result) and incremented value but also whether or not the original value contained a zero.  This is important because it will allow us to later use this information to decide whether or not a row will contain a carry.  Now, these two steps of the algorithm rely on nothing but the data in the current row and are completed in parallel across all rows.  Therefore, with the given tile set in Figure \ref{fig:worsttileset}, these two steps are straightforward to verify.

% This basically shows us that we have enough information to answer the question... not sure it qualifies as a proof of correctness
The next step is for every row to select a solution from the two that were generated as well as propagate its carry information.  We will, for the moment, assume that some row has not received the information of the previous carry and will concentrate on this row.  At this step we know if the current row generated a carry $C \in \{T, F\}$, if the rows sum contains a zero $Z \in \{T, F\}$, and two possible row values $A$ and $B$.  $A$ represents the sum if the row receives a carry-in of $0$ and $B$ the opposite.  In order to continue from this step an answer must be selected and a carry or no-carry must be propagated.  If $Z = T$, then we immediately know that we will propagate whatever $C$ may be. If $Z = F$ and $C = T$, we also know that we must propagate a carry.  The only situation in which we do not know whether we will propagate a carry is when $Z = F$ and $C = F$.  When we encounter this situation we propagate whatever the carry was in the previous row.  Also, in order to decide whether we will select $A$ or $B$ we need only the previous carry.  Therefore, assuming we have the correct previous carry we can correctly select both the proper value for this row as well as propagate the correct carry to the next row.

Finally, because we know the initial carry is correct (it is part of the seed), we know that the first row can select the correct result as well as propagate the correct carry.  Then, because we know that the first row's carry is correct, we know that the second row can select the correct result as well as propagate the carry.  This chain continues until it reaches the last row leading to selecting all the correct values as well as propagating all the correct carries.

The last step is to select the most significant bits value which is solely based on the last row's carry.  If the last row propagates a carry it is a $1$, otherwise it is a $0$.   Since we know that the last carry is correct we know that the last value is selected properly.

% Assume you are given two binary numbers $A$ and $B$ of the same size, then any two bits in some position $i$ will do one of two things in terms of carries.  If $A_i + B_i$ are equal to $1$ then it will propagate the carry from the previous addition.  If $A_i + B_i$ is equal to either $0$ or $10$, then it can immediately propagate its carry because it is known that it will carry in the $10$ case or not carry in the $0$.

\section{$O(\log{n})$ Average Case, $O(\sqrt[2]{n})$ Worst Case Addition.} \label{app:combined}
\subsection{Construction} \label{sec:combined_construction}

\begin{figure}[htp]
	\centering
	\includegraphics[scale=1.0]{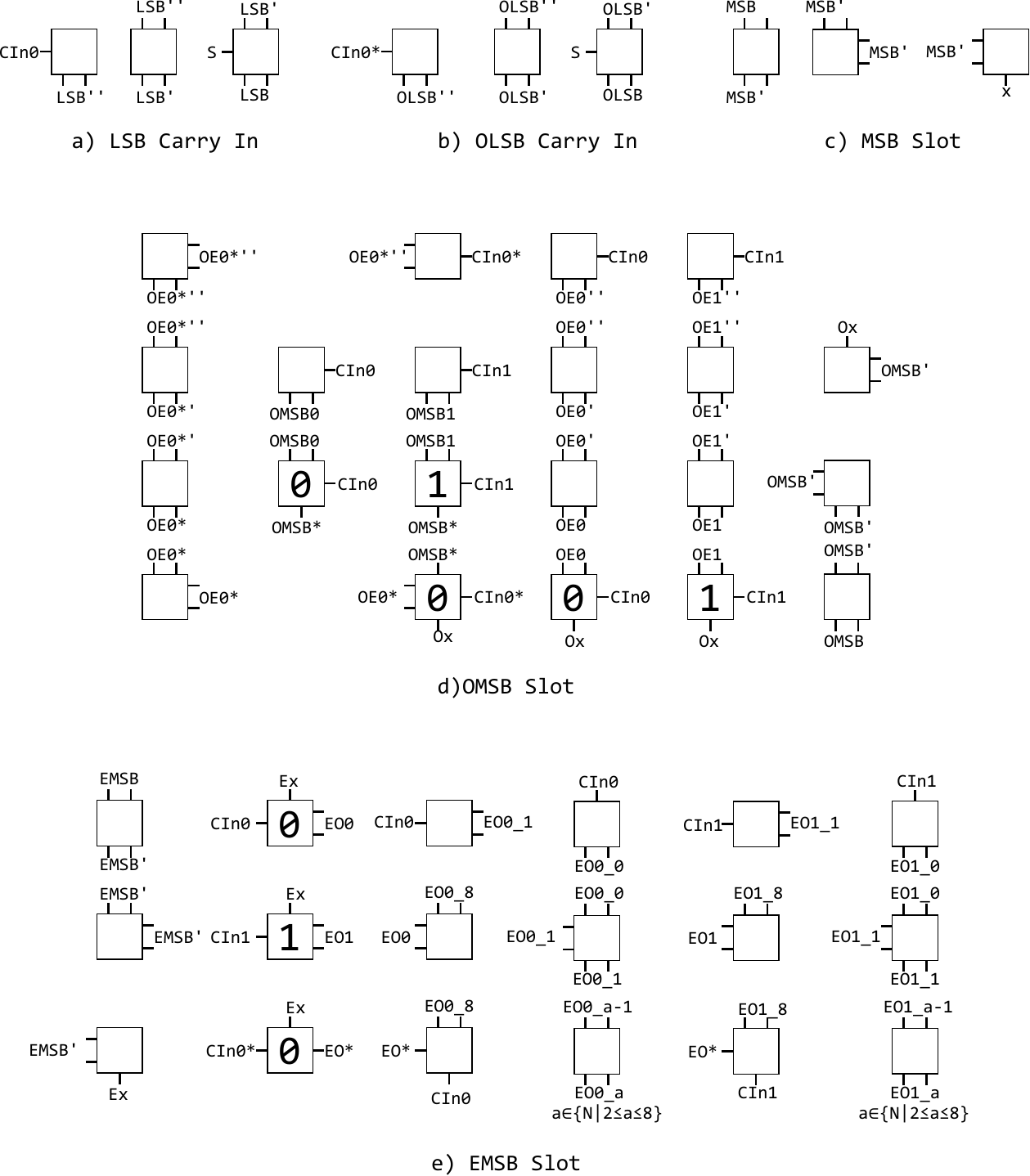}
	\label{fig:CombinedTileSet1}
	\caption{Partial set of tiles necessary to implement $O(\log{n})$ average case, $O(\sqrt{n})$ worst case combined addition (see also next figure).}
	
\end{figure}

\begin{figure}[htp]
	\centering
	\includegraphics[scale=1.0]{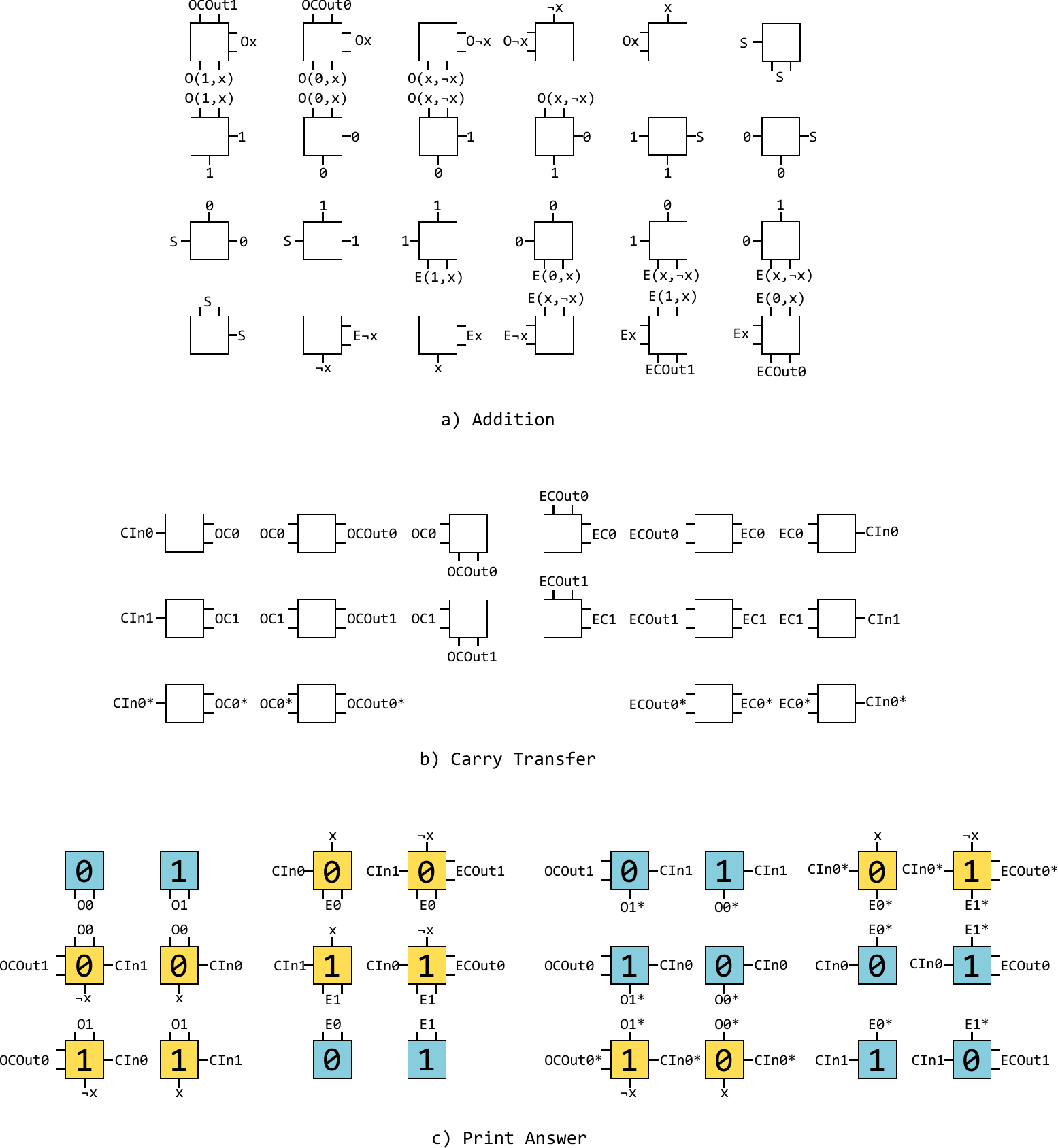}
	\label{fig:CombinedTileSet2}
	\caption{Tiles necessary to implement $O(\log{n})$ average case, $O(\sqrt{n})$ worst case combined addition (continued from previous figure).}
	
\end{figure}

	This construction combines the two-dimensional scaffold principle of the simple worst-case addition construction (Section \ref{sec:worstcase}) with the principle that certain addend-pairs can compute a carry out before they get a carry in, which was shown to reduce the average case run-time to $O(\log{n})$ in Section \ref{sec:avgcase_time}. Contrary to the addition construction in Section \ref{sec:worst_construction}, the directionality of adjacent rows is antiparallel in the construction described here. Every odd row beginning with row one, which is the southernmost row, has the least significant bit on the east and the most significant bit on the west. Every even row has bits in the opposite order, as shown in Figure \ref{fig:LogSquareInput}a. Each odd row, along with the even row above, should be considered as a single section in which the addition mechanism is nearly identical to the $O(\log{n})$ average case adder TAC.  Within each section, carry outs are propagated east to west on the odd row, up in constant time from the most significant bit on the odd row (OMSB) to the least significant bit on the even row (ELSB), and from west to east on the even row. Adjacent rows are anti-parallel so that the distance between the most significant bit (MSB) of each section is a constant distance from the least significant bit (LSB) of the section above.  This modification allows us to apply the $O(\log{n})$ average case addition between each pair of bits, and at the same time apply the carry propagation mechanism of the $O(\sqrt{n})$ worst case addition construction between the MSB of each even row.

\begin{figure}[htp]
	\centering
	\includegraphics[scale=.75]{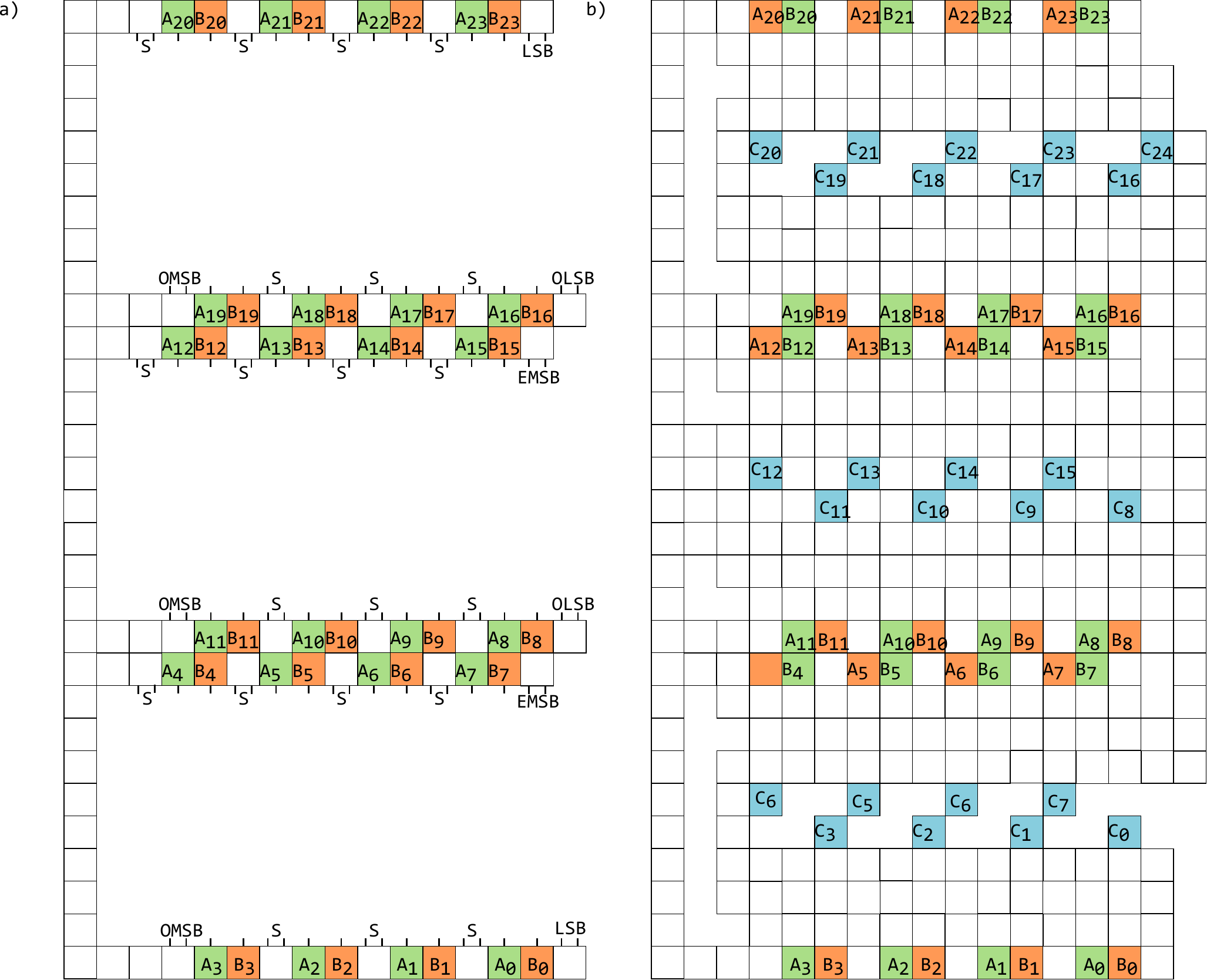}
	\caption{a) The input template for addition of two $n$-bit binary numbers $A$ and $B$. b)The output template for the addition of two $n$-bit binary numbers $A$ and $B$.}
	\label{fig:LogSquareInput}
\end{figure}

\paragraph{Carry Passing and Prediction of Results Within Sections.}
Within each section, or odd/even row pair, carry outs are propagated as stated in the above paragraph (Figure \ref{fig:LSExample1}a). In addition, each row will present ``predicted" results, that is, the result of the computation as if there were no carry from the MSB of the previous section. These results are shown as yellow tiles in Figure \ref{fig:LSExample1}b-c. Note the blue tiles in Figure \ref{fig:LSExample1}d. These blue tiles represent the final sum of the addend pairs. They may be computed without carry information propagated from a section below because the addend pairs are either 1)located in the southernmost section, where it is immediately known that the carry-in is $0$ or 2)are flanked by a less significant addend pair for which the carry out is immediately known. The north face glues on yellow tiles in Figure \ref{fig:LSExample1}d which contain a star * rely on a carry from a previous section. The mechanism of carry propagation within a section can be seen in Figure \ref{fig:LSExample2}a-b. Vertical columns grow up along the west side of each section to propagate the carry from the odd row to the even row.

\paragraph{Carry Propagation Between Sections.}
Figure \ref{fig:LSExample2}c shows a vertical column growing between the southernmost section and the section above. This column is propagating a carry from the most significant bit of the southernmost section $A$ to the least significant bit of the section $B$ to the north. This information continues to the most significant bit of section $B$ at which point a carry bit is computed to propagate to the section north of $B$ (Figure \ref{fig:LSExample2}d, Figure\ref{fig:LSExample3}a. The terminal tile assembly with the computed sum is presented in Figure \ref{fig:LSExample3}b.

\begin{figure}[htp]
	\centering
	\includegraphics[scale=0.6]{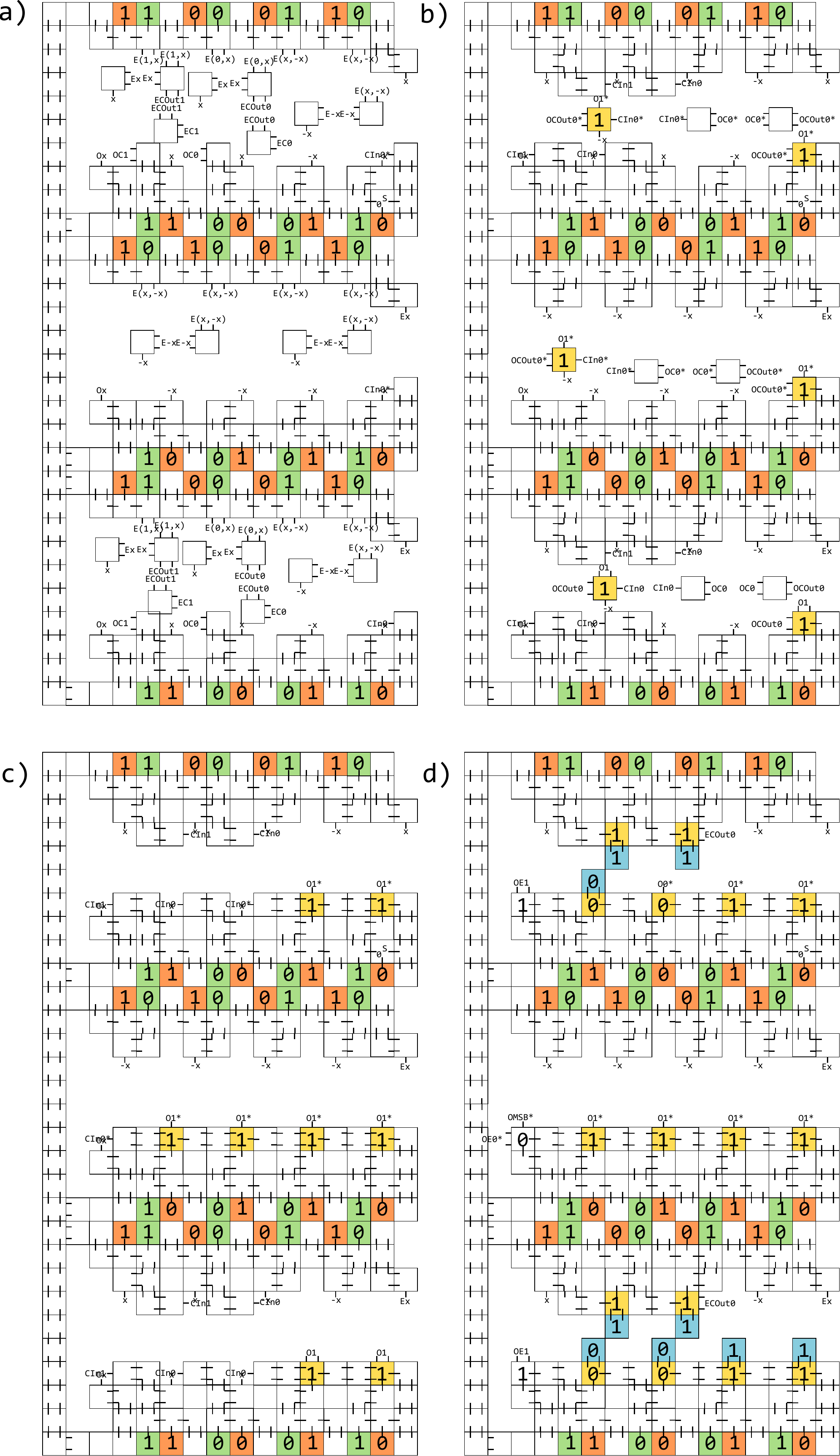}
	\caption{a) Carry out bits are computed for each addend pair where a carry out can be deduced immediately. b-c)``Predicted" results (yellow tiles) begin to attach. d) Blue tiles represent final result. }
	\label{fig:LSExample1}
\end{figure}

\begin{figure}[htp]
	\centering
	\includegraphics[scale=0.6]{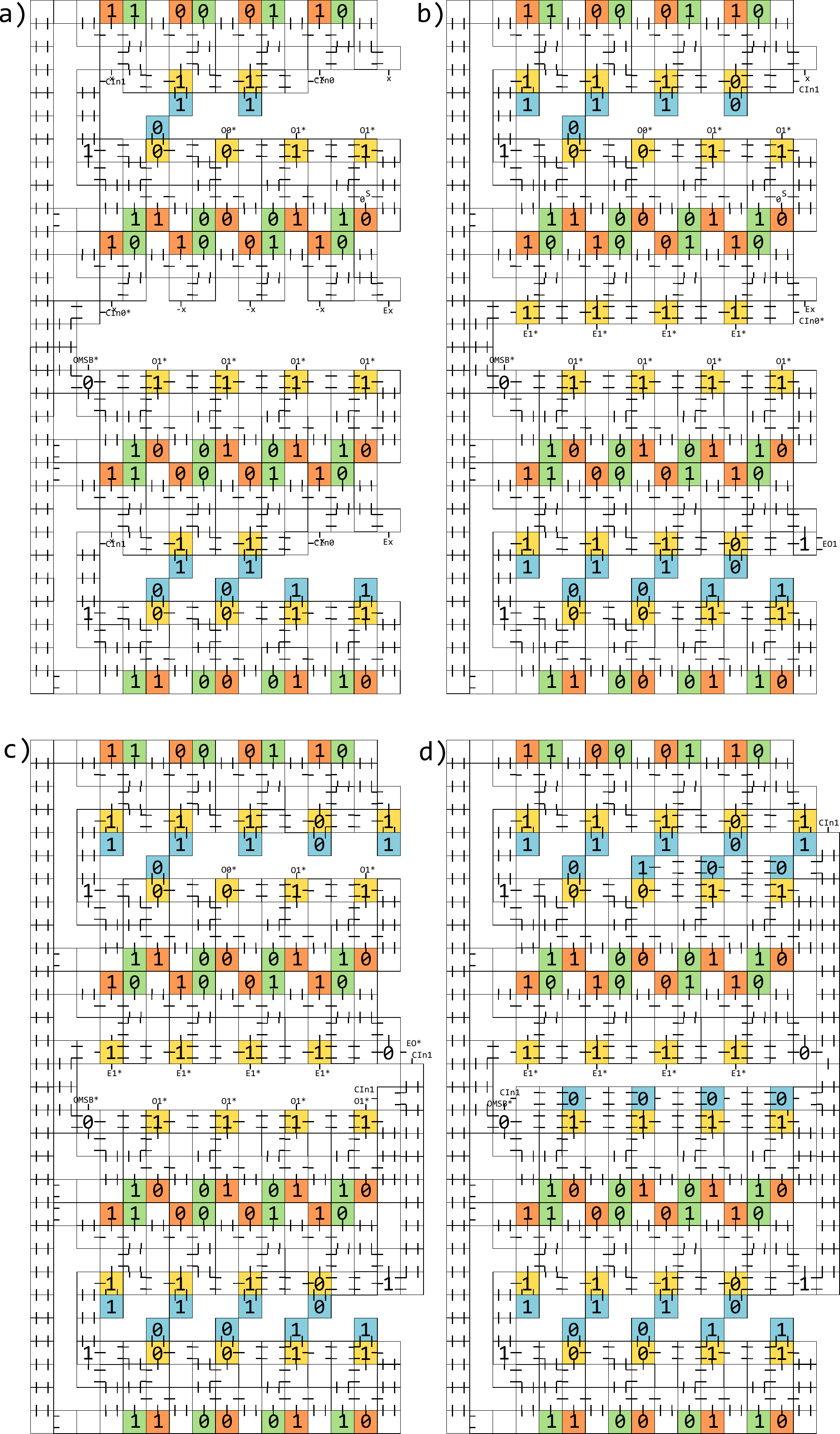}
	\caption{a-b) Carry bits are propagated within sections from the MSB of the odd row to the LSB of the even row using vertical columns. c-d)Carry bits are propagated between sections via vertical columns on the east side of the assembly. }
	\label{fig:LSExample2}
\end{figure}

\begin{figure}[htp]
	\centering
	\includegraphics[scale=0.6]{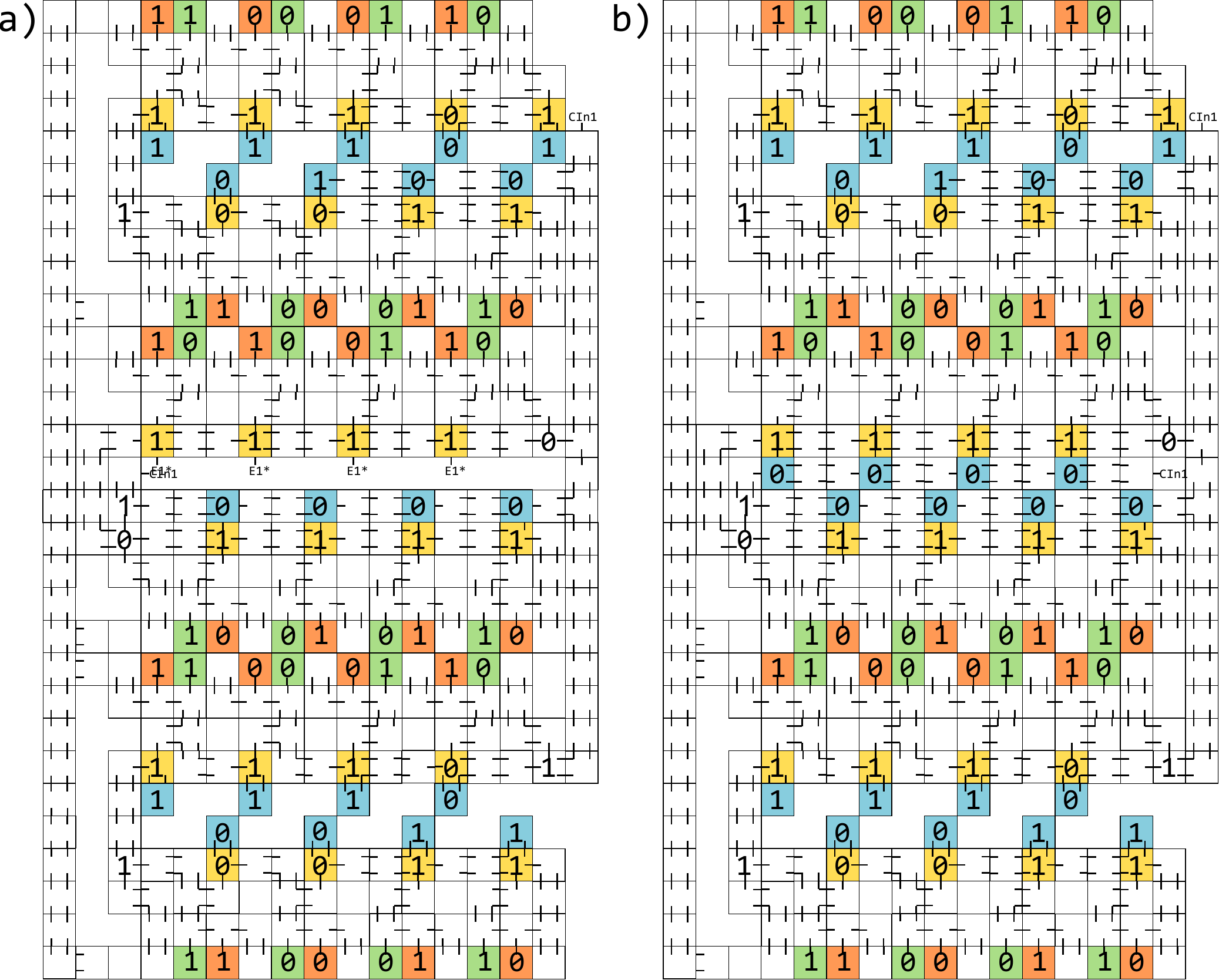}
	\caption{a) A carry bit from the first section propagates through the middle section in this assembly. b) The terminal assembly displaying the output $C$.}
	\label{fig:LSExample3}
\end{figure}
\subsection{Time Complexity.} \label{sec:combined_time}
\paragraph{$O(\sqrt{n})$ Worst Case and $O(log\ n)$ Average Case Run-Time}
The construction presented in this section represents a combination of the constructions presented in Sections \ref{sec:avgcase} and \ref{sec:worstcase}.  Thus, this runtime analysis combines elements from the time complexity proofs in \ref{sec:avgcase_time} and \ref{sec:worst_time}. In the construction described here, each section on the scaffold precomputes the sum as if there were a carry from the previous section and also as if there were no carry from the previous section. In some cases, whether a carry bit is passed in from the section below is irrelevant and parts of the final sum may be generated before this information is available. If the most significant addend-pair bits (MSB) of a given section can immediately generate a carry out for propagation to the next section, they do so. When a carry in arrives from the previous section, final values are selected from the precompution step, if necessary. In this analysis of time complexity, we first analyze the run time of the precomputataion step, and then consider the run time of propagating carry bits from section to section up the scaffold.

 Consider the binary sequence $T$ of length $2n$ composed of $n$-bit binary numbers $A$ and $B$ as defined in Section \ref{sec:avgcase_time}. For this construction, divide $T$ into $\frac{\sqrt{n}}{2}$ smaller sequences, or sections, $S_0, S_1, \cdots, S_{\frac{\sqrt{n}}{2}-1}$, with each section containing $2\sqrt{n}$ addend-pairs. These sections are arranged on a scaffold as depicted in Figure \ref{fig:upperbound_combined}. Note that the distance between the bottom half and the top half of a given section is constant and requires a constant number of steps to traverse. We first treat each section, $S_i$, as an independent addition problem without regard for a carry in from the previous section. Define $k$ as the longest contiguous sequence of $(0, 1)$ and $(0, 1)$ addend-pairs in $S_i$.  It follows from the proof in Section \ref{sec:avgcase_time} that the run-time for $S_i$ is bounded upwards by the length of $k$. The worst-case runtime for addition over the sequence $S_i$ would thus occur when $k$ = $2\sqrt{n}$. Therefore, the worst case run time for addition \emph{within} a section is $O(\sqrt{n})$. It also follows, as described in Section \ref{sec:avgcase_time}, that the expected length of $k$ is $O(\log{n})$. Therefore, the average run time for each addition within each section is $O(\log{n})$.
The precomputation within each of the $\frac{\sqrt{n}}{2}$ sections occurs independently in parallel, leading to a worst case $O(\sqrt{n})$, average case $O(\log{n})$ precomputation run-time for all sections.

\begin{figure}[htp]
	\centering
	\includegraphics[width=4in]{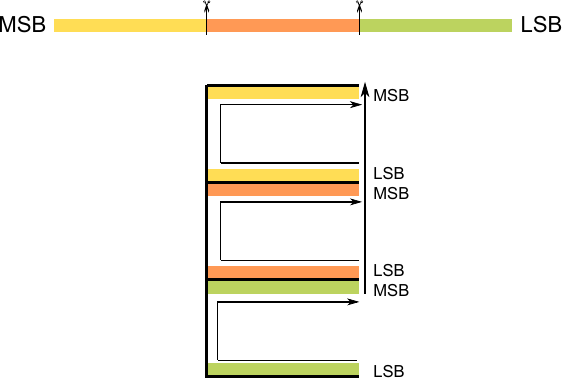}
	\caption{Top: $T$ is separated into sections. Bottom: Each section of $T$ is arranged on a scaffold. Arrows represent the general direction of carry bit propagation. MSB denotes the most significant bit of a section, and LSB denotes the least significant bit of a section.}
	\label{fig:upperbound_combined}
\end{figure}

After performing this precomputation within each section, we must propagate carries between each section. The distance between two neighboring sections is constant and may be traversed by a column of tiles in a constant number of steps. Therefore, a carry out bit, once generated, can be propagated from one section to the next in constant time. A carry out bit may be propagated to the next section, $S_{i+1}$, immediately if the most significant addend pair of section $S_i$ consists of $(0,0)$ or $(0,1)$. Otherwise, this most significant addend pair of $S_i$ must wait for a carry in before generating a carry out. Therefore, the composition of the most significant addend-pairs of the sections acts as a limiting factor to the speed with which carries may be propagated across all of the sections.  Consider the binary sequence $P$ which is comprised of the most significant addend pairs of each section and has a length of $\frac{\sqrt{n}}{2}$ addend-pairs. Let $k$ be the longest contiguous sequence of $(1, 0)$ and $(0, 1)$ addend-pairs in $P$. The propagation of carry bits through $P$ is bounded upwards by the length of $k$, with the worst-case being when $k = \frac{\sqrt{n}}{2}$ and an average case being when $k = O(\log{n})$. Thus, the propagation of carry bits between each section up to the most significant addend-pair of $T$ is bounded upwards by $O(\sqrt{n})$ and has an average run-time of $O(\log{n})$.
Therefore, this addition algorithm has an upper bound of $O(\sqrt{n})$ and an average run-time of $O(\log{n})$.

The parallel run time analysis presented above can be extended to the continuous run time model through application of Lemma~\ref{lemma:maxExpSums} and Lemma~\ref{lemma:waypoints} similar to what is done in Sections~\ref{sec:logadditioncontinuous} and \ref{sec:worstcasecontinuous}, yielding worst case $O(\sqrt{n})$ and average case $O(\log n)$ continuous run time bounds.
\subsection{Correctness.} \label{sec:combined_correct}
Every row with north facing glues performs its addition in the exact same way as described in Section \ref{sec:avgcase_construction} assuming an incoming carry of $0$.  As such we use the proof of correctness from Section \ref{sec:avgcase_correct} to show that this step is correct.  Also, every row with south facing glues performs this same addition except with an incoming carry from the addition of the lower north facing glue side.  Thus, the same proof also applies to this side.  Considering both of these additions as the first step, we can say that the first step correctly adds a section with a carry-in of $0$.

Since the first section has no sections before it, the addition of the first section has an input carry of $0$ allowing the copying of its value up to the final ``display" row.  Once an addend-pair pair knows its carry-in it can display its results.  This combined with the fact that we proved the addition of any section with an input carry of $0$ proves the correctness for the first section.

Once a section finishes the first step, there are three possible carry-outs from the MSB addend-pair in the section: $0$, $1$, and $0*$.  If the section's MSB addend-pair propagates a $0$ or $1$ then somewhere in the section some addend-pair generated its carry without depending on a carry from the previous section leading to a correct propagation.  If the section's MSB addend-pair propagates a $0*$, then no addend-pair in the section was able to generate a carry, meaning that no addend-pair in the section contained equal bits.  Therefore, the section will correctly propagate whatever carry it received to the next section. These two clauses cover all cases in terms of carry propagation from one section to next section.

Every section other than the first one performs its addition with an input carry of $0*$ which indicates an unknown carry with a value of $0$. In other words, it continues the calculation assuming a $0$ input carry but cannot copy any undetermined values, due to the unknown carry, into the ``display" row until it receives the real carry from the previous section.  This carry only gets propagated until it reaches an addend-pair with equal bits because at this point the addend-pair would have already generated its carry regardless of any incoming input carry and propagated it.  One can see that the beginning of any section, other than the first, is essentially calculated by the algorithm as if it was in the center of some series of addend-pairs with unequal bits.  When the correct carry propagates from the previous section to the current section, the correct values may then be copied into the ``display" row.  If a $0$ is carried in, then it is a simple copy up of the value into the ``display" row.  If a $1$ is carried in, then it is the inverse of what was previously calculated. Assuming some section receives the correct carry from the previous section, that section will ``display" the correct result. We have shown that the first section propagates a carry correctly, and therefore all subsequent sections propagate their correct carry out. We have also shown that if each section propagates the correct carry-out to the next section, then the addition of addend-pairs within each section is performed correctly, producing the correct sum of $A$ and $B$, $C$.

\section{Multiplication}\label{app:multiplcation}
\subsection{Construction.} \label{sec:multiplication_construction}
This TAC uses a  shift-and-add algorithm for multiplication. Let $A$ and $B$ be two $n$-bit binary numbers, and let $a_i$ and $b_i$ represent the $i$th bits of $A$ and $B$ respectively. Note that if $B=\sum_{i} b_i2^i$ for $b_i\in{0,1}$, then the product $AB=\sum_{i}A b_i2^i$. In this algorithm we perform the multiplication by multiplying A by each of the powers of two in $B$. Each of these partial products is added to the running sum as soon as the partial product is available.

%As an example we will show how our system multiplies two $64$ bit numbers. Next we deploy the numbers into 4 by 4 by 4 positions with each position containing an 8 by 8 matrix. %When deploy the numbers
To describe our construction we focus on the case of multiplying two 64-bit numbers.

\paragraph{Vector Label And Tile Types.}

To permit a high-level description of the multiplication TAC, we use a \emph{Vector Label} technique to designate the binding possibilities between tiles. Rather than showing the glue types and their placement on each edge of a tile, we label a tile with a vector which describes the rules by which that tile may bind to adjacent tiles. A tile with such a vector actually describes a group of tiles that contains a certain element in the vector that could attach to a certain assembly. The number of tile types described by a Vector Labeled tileset is at most $O(l^d)$, where $l$ is the number of label types and $d$ is the vector dimension. An example set of Vector Labeled tiles may be found in Figure \ref{fig:Vector_Label}.

\begin{figure}[htp]
	\centering
	\includegraphics[width=3in]{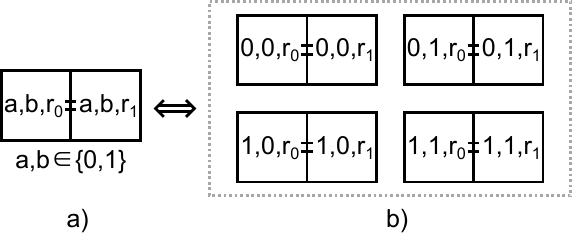}
	\caption{a) A pair of Vector Labeled tiles describes a binding rule. Dimensions $a$ and $b$ must be the same on both tiles for the tiles to bind, and tiles containing $r_1$ in the third dimension may bind only to tiles containing $r_0$ in the third dimension. b)Such a binding rule can describe 8 tile types with 4 binding possibilities.}
	\label{fig:Vector_Label}
\end{figure}

\paragraph{Input/Output Template.}

The input template encoding the multiplicand and multiplier, $A$ and $B$, is a rectangular assembly of size $2n$ tiles. The product of two $n$-bit numbers will have at most $2n$ bits. This extra $n$ bits of space comprises the west half of the assembly, as shown in Figure \ref{fig:Mult_input_template}, while the east half encodes $A$ and $B$ via vector labeled tiles. The first and second elements of the vector on each tile are a given bit of $A$ and $B$.
The final assembly for 64-bit inputs is shown in Figure \ref{fig:Mult_output_template}. The product $AB$ is encoded on the top surface of the final assembly, marked as $Output$ in Figure \ref{fig:Mult_output_template}.

Before going into details of the TAC implementation, we first give a brief overview of the different parts of the output structure (Fig. \ref{fig:Mult_output_template}). Given two $n$-bit inputs, $A$ and $B$, our multiplication algorithm involves a summation of $n$ numbers, $x_0+x_1+...+x_{n-1}$.  For any number $x_k$ within these $n$ numbers, $x_k=A b_k 2^k$. The output assembly (Fig. \ref{fig:Mult_output_template}) has four horizontal layers, one stacked on top of the next, each of which sums 16 of the 64 numbers to be summed for the 64-bit example. More generally, there are $\sqrt[3]{n}$ layers, $l_0, l_1,...,l_{\sqrt[3]{n}-1}$, and each layer $l_k$ is responsible for summing $n^{2/3}$ of the $n$ numbers, $x_{k\sqrt[3]{n}}+x_{k\sqrt[3]{n}+1}+...+x_{(k+1)\sqrt[3]{n}-1}$. On each layer, there are  $\sqrt[3]{n}$ columns running west to east, each of which sums $n^{1/3}$ numbers of the $n^{2/3}$ numbers which are summed within each layer.

\begin{figure}[htp]
	\centering
	\subfigure[Multiplication input template assembly.]{
		\label{fig:Mult_input_template}
		\includegraphics[width=4in]{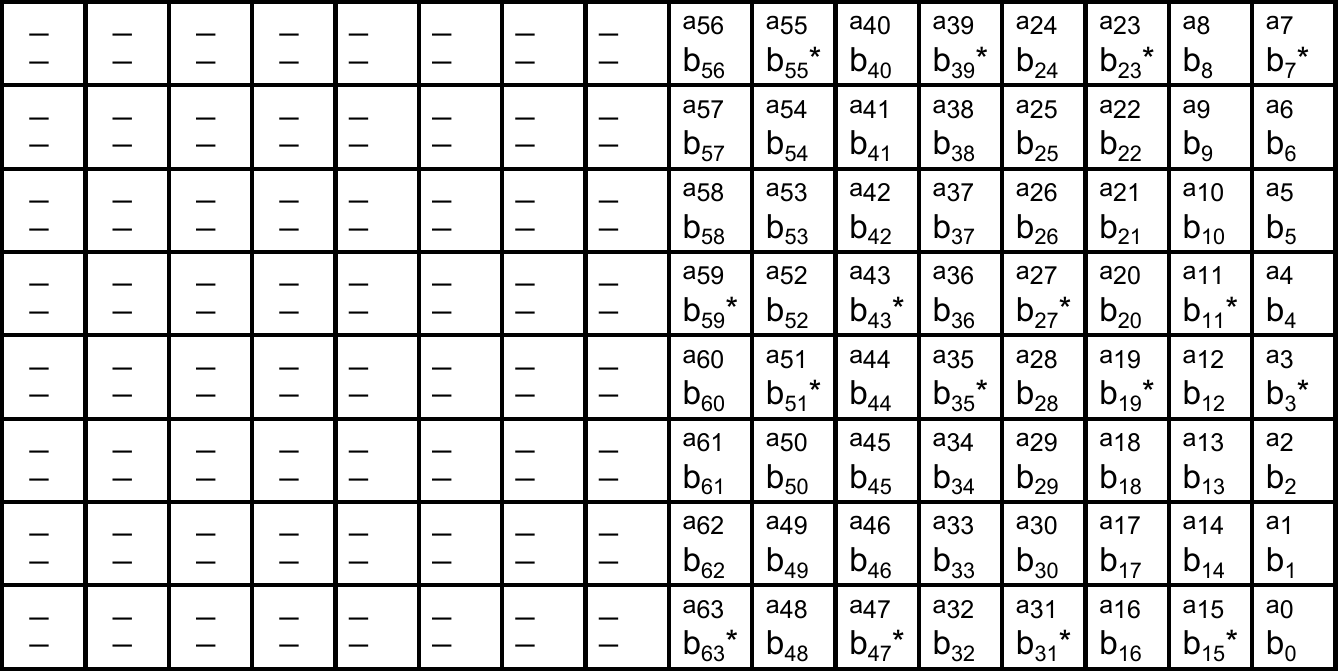}
	}
	~~~
	\subfigure[Multiplication output assembly.]{
		\label{fig:Mult_output_template}
		\includegraphics[width=1.6in]{images/Mult_Output}
	}
\caption{The input seed and the output assembly for 64-bit inputs, $A$ and $B$.}
\end{figure}

%\paragraph{Shift.}
\paragraph{Part One: Deployment.}
%, and how the output structure is formed. We later go into details about the different calculations being performed in each part of the structure.

An important step in the multiplication process is deploying the $n$ numbers to be summed, $x_0+x_1+...+x_{n-1}$. These numbers must be moved into different positions for addition to take place. First, the seed (Fig. \ref{fig:Mult_1}) is copied both up and south, as is depicted in Figure \ref{fig:Mult_2}. Figure \ref{fig:Rotate_and_copy} demonstrates that for any 1-dimensional input, it is possible to create a tile set to send the input information into two directions. This same principle can be used to copy the 2-dimensional input area both up and south into three dimensions as is shown in Figure \ref{fig:Mult_2}.

The copy of the seed to the south will then be copied to the south again and to the east, as is shown in Figure \ref{fig:Mult_3}. This copy to the east will head the formation of a column. Columns run west to east in each layer, and each layer contains $\sqrt[3]{n}$ columns. In this particular 64-bit input example, there will be four columns to each layer, each of length  $\sqrt[3]{n}$ blocks. The input information will continue to be propagated south and east until all columns for a given layer are formed. The input information that was copied up (Fig. \ref{fig:Mult_3}) will act as a seed for the next layer, which is formed in the same way that the first layer was formed.

As the input information is copied in different directions, different functions are applied. For the input information copied up, $n^{\frac{2}{3}}$ bits of zero are attached after the least significant bit of $A$, and the least significant $n^{\frac{2}{3}}$ bits of $B$ are removed. Let $p(x)=x*2$ and $q(x)=[x/2]$. Then in figure \ref{fig:Mult_3}, $A=p(A)^{16}$ means attach 16 zeros after A and $B=q(B)^{16}$ means remove the 16 least significant bits from B. We add and remove these bits by shifting numbers. Figure \ref{fig:shift_example_1} demonstrates that it is possible to create a tile set to shift any input.

The copy of the input information to the east marks the head of the first column. As this column is propagated from west to east, the least significant bits of $B$ are removed one by one, and bits of 0 are appended to the least significant end of $A$ one by one, as is seen in Figure \ref{fig:Mult_4}. At the end of the column, a total of $n^{1/3}$ 0 bits will have been appended to A and $n^{1/3}$ of the least significant bits of $B$ will have been removed.

The copy of the input to the south propagates further to the south and then east to form the next column (Figure \ref{fig:Mult_5}-\ref{fig:Mult_6}). During this southward propagation, the least significant $n^{1/3}$ bits of $B$ are removed, and $n^{1/3}$ bits of zero are appended to the end of $A$. This southward propagation continues, and each time $n^{1/3}$ additional bits of zero are appended to the least significant end of $A$ and $n^{1/3}$ least significant bits are removed from $B$, until the heads of each of the $\sqrt[3]{n}$ columns have formed on the east side. These columns are then free to grow eastward (Figure \ref{fig:Mult_8}).

The copy of the seed that is sent upward has the least significant $n^{2/3}$ bits removed from $B$ and $n^{2/3}$ bits of 0 appended to the least significant end of $A$. This is the beginning of the next layer. As the information is propagated upward, it may continue to form layers, which form in parallel (Figure \ref{fig:Mult_6}). Each time the information is propagated upwards to form the next layer, another $n^{2/3}$ bits of 0 are appended to the least significant end of $A$, and the least significant $n^{2/3}$ bits are removed from $B$.

\begin{figure}[htp]
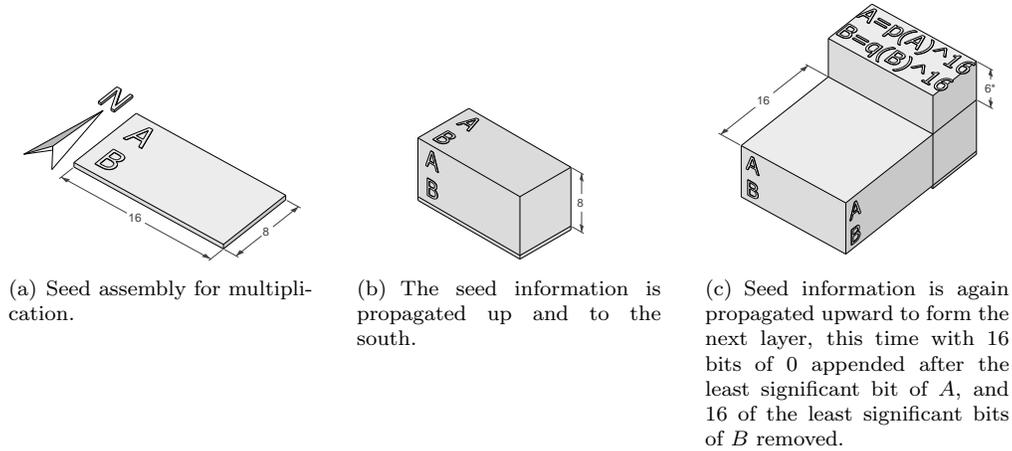

	\centering
	\subfigure[Seed assembly for multiplication.]{
	\label{fig:Mult_1}
	\includegraphics[width=1.5in]{images/Mult_1}
	}
	~~~
	\subfigure[The seed information is propagated up and to the south.]{
		\label{fig:Mult_2}
		\includegraphics[width=1.5in]{images/Mult_2}
	}
	~~~
	\subfigure[Seed information is again propagated upward to form the next layer, this time with 16 bits of 0 appended after the least significant bit of $A$, and 16 of the least significant bits of $B$ removed.]{
		\label{fig:Mult_3}
		\includegraphics[width=1.5in]{images/Mult_3}
	}

\caption{This series shows the propagation of information south for use in the formation of columns and up to form a new layer.}
\end{figure}

\begin{figure}[htp]
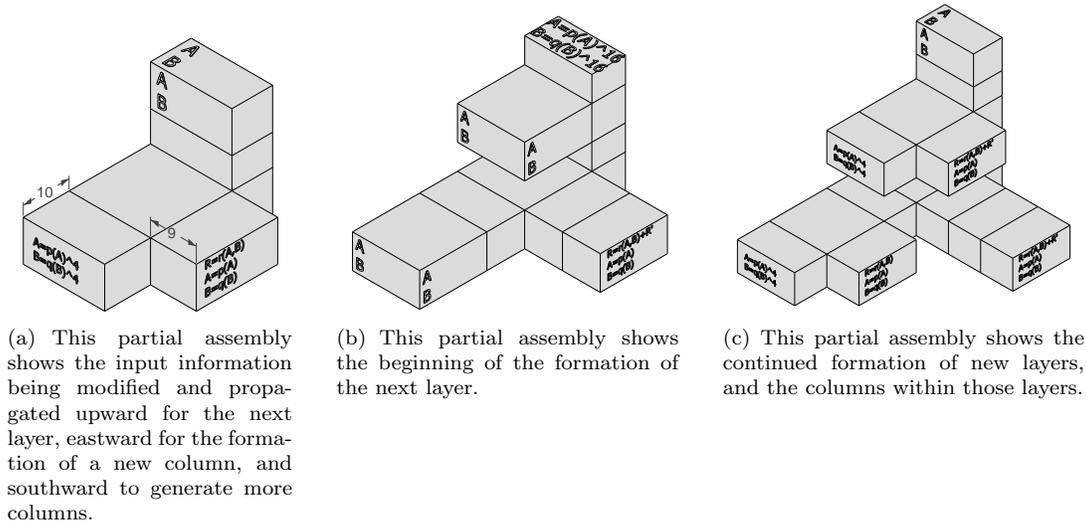

	\centering
	\subfigure[This partial assembly shows the input information being modified and propagated upward for the next layer, eastward for the formation of a new column, and southward to generate more columns.]{
		\label{fig:Mult_4}
		\includegraphics[width=1.4in]{images/Mult_4}
	}
	~~~
	\subfigure[This partial assembly shows the beginning of the formation of the next layer. ]{
		\label{fig:Mult_5}
		\includegraphics[width=1.7in]{images/Mult_5}
	}	
	~~~
	\subfigure[This partial assembly shows the continued formation of new layers, and the columns within those layers.]{
		\label{fig:Mult_6}
		\includegraphics[width=1.8in]{images/Mult_6}
	}
\caption{This series shows the process of forming new columns and layers.}
\end{figure}

\paragraph{Part Two: Addition.}\label{sec:multiplication_construction_Addition}

While shifting bits between each layer, and between and within each column, the tile set simultaneously performs the additions. The addition algorithm here is based on Section \ref{sec:worstcase}. Recall that our multiplication algorithm involves a summation of $n$ numbers, $x_0+x_1+...+x_{n-1}$.  For any number $x_k$ within these $n$ numbers, $x_k=A b_k 2^k$.  An example of a two-number multiplication, which is required to attain $x_k$, is shown in Figure \ref{fig:Mult_example}. The five elements in the vector label are the current bit of A, the current bit of B, the predicted result without a carry in, the predicted result with a carry in, and the current bit result. The numbers are arranged in zigzag fashion.

After obtaining the result of each column, shown in Figure \ref{fig:Mult_8}, we copy the result of the first column, rotate this result to the south, and send it to its neighbor (Figure \ref{fig:Mult_9}). Note that we don't need A and B any more; we only copy the result. Using the reversed design of the rotate and copy box (Figure \ref{fig:Rotate_and_copy}), we merge the two inputs from two sides into one single output. Using the same principle shown in Section \ref{sec:worstcase} and Figure \ref{fig:Mult_example} we then sum all results from each column together, as shown in Figure \ref{fig:Mult_10} and Figure \ref{fig:Mult_11}.

\begin{figure}[htp]
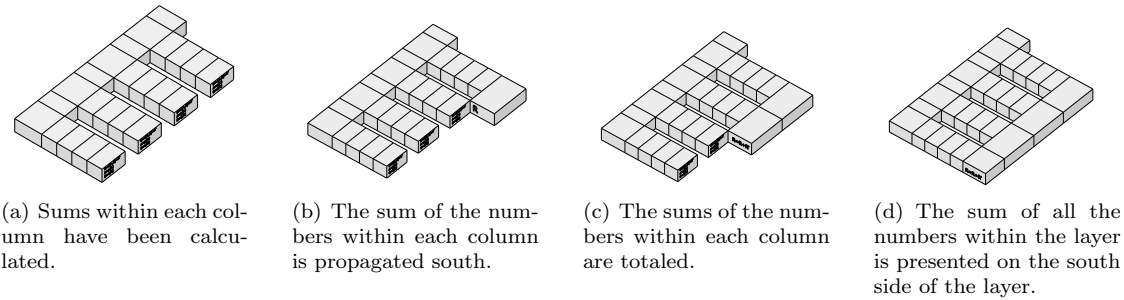

	\centering
	\subfigure[Sums within each column have been calculated.]{
		\label{fig:Mult_8}
		\includegraphics[width=1.2in]{images/Mult_8}
	}
	~~~
	\subfigure[The sum of the numbers within each column is propagated south.]{
		\label{fig:Mult_9}
		\includegraphics[width=1.2in]{images/Mult_9}
	}	
	~~~
	\subfigure[The sums of the numbers within each column are totaled.]{
		\label{fig:Mult_10}
		\includegraphics[width=1.2in]{images/Mult_10}
	}	
	~~~
	\subfigure[The sum of all the numbers within the layer is presented on the south side of the layer.]{
		\label{fig:Mult_11}
		\includegraphics[width=1.2in]{images/Mult_11}
	}
\caption{This series shows the process of summing all numbers within each layer.}
\end{figure}

The result of the summation in the first layer will be copied up to the second layer (Figure \ref{fig:Mult_14}), where it is added to the result from the second layer. This process is repeated until the last layer is reached and a final result is obtained, as shown in Figure \ref{fig:Mult_13} and Figure \ref{fig:Mult_14}.
%Same with the addition of each columns in same layer, the final result (Figure \ref{fig:Mult_14}) will be get by perform the summation of the result of each layer.

\begin{figure}[htp]
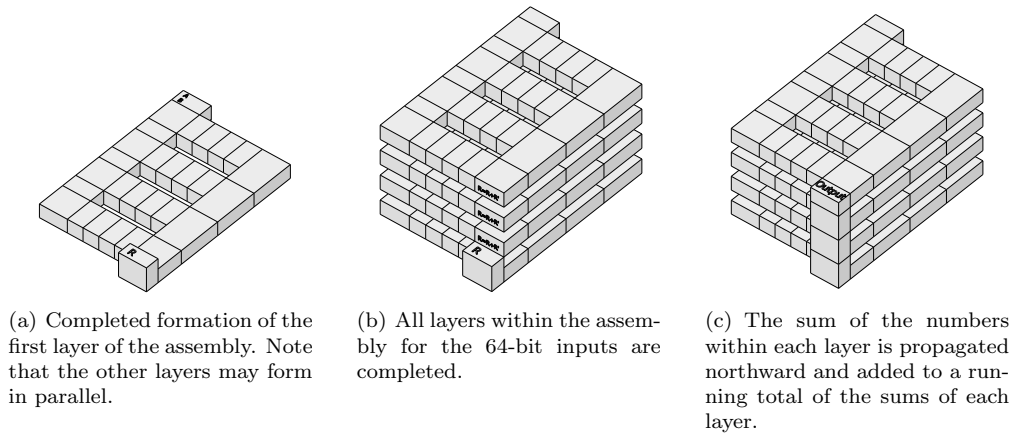

	\centering
	\subfigure[Completed formation of the first layer of the assembly. Note that the other layers may form in parallel.]{
		\label{fig:Mult_12}
		\includegraphics[width=1.5in]{images/Mult_12}
	}
	~~~
	\subfigure[All layers within the assembly for the 64-bit inputs are completed.]{
		\label{fig:Mult_13}
		\includegraphics[width=1.5in]{images/Mult_13}
	}
	~~~
	\subfigure[The sum of the numbers within each layer is propagated northward and added to a running total of the sums of each layer.]{
		\label{fig:Mult_14}
		\includegraphics[width=1.5in]{images/Mult_14}
	}
\caption{After summing the results of each layer, the multiplication result is displayed on the top surface of the assembly.}
\end{figure}

\subsection{Time Complexity.} \label{sec:multiplication_time}
\paragraph{$O(n^{\frac{5}{6}})$ Worst Case Run-Time.}
Our multiplication construction consists of copy (Figure \ref{fig:Rotate_and_copy}), shift (Figure \ref{fig:shift_example_1}), and addition (Figure \ref{fig:Mult_example}) operations.
%The construction of multiplication is composed by constructions of Rotation \ref{fig:Rotate_and_copy} Shift \ref{fig:shift_example_1} and Addition \ref{sec:multiplication_construction_Addition}.
The addition algorithm here is a 3D version of our $O(\sqrt{n})$ Worst Case Addition TAC (Section \ref{sec:worst_time}). In order to ease the analysis we will assume that each logical step of the algorithm happens in a step-by-step fashion even though parts of the algorithm are running concurrently.

\emph{Runtime For Copy Block.} The system copies the input number up and south using a copy block. The input for a copy block is an $O(\sqrt{n})$ by $O(\sqrt{n})$ rectangle. For each row running north to south a copy is generated independently, so we analyze the process in a longitudinal section and then expand it to the entire block. Figure \ref{fig:Rotate_and_copy} shows an example of a copy box copying one row up and to the left. The copy box starts the copy from the leftmost input. This tile is marked as ``r''. If a tile is marked as ``r'', the tile on its right will mark its north tile as ``r'' , otherwise, the tile will just copy itself to the north. After a tile receives an ``r'', the tile on its north will then send the label both up and to the left. After that, the tiles sending labels to the left will also copy the label from the bottom to the top, thus sending the labels up. Each tile under the ``r'' tiles is dependent on the tiles on its left and bottom, so the runtime for all tiles under the attached diagonal is $O(\sqrt{n})$. The tiles after the ``r'' tiles are dependent on the tiles on the right and bottom, so the total runtime of this rotate and copy box is $O(\sqrt{n})$.

\emph{Runtime For Shift Block.} The system shifts the input numbers in order to attach or remove bits using a Shift Block. An example of shifting one row of input to the left is shown in Figure \ref{fig:shift_example_1}.
%Upon the even tile in the input will first grow a tile to copy the even tile input up. Then upon the odd tile in the input will then grow a tile cooperative banding with the one upon the even input to obtain both input.
%The tile upon the odd input will then repeatedly grows up and then left till reach the desired position.
In this example, bits are shifted to the left. So, from the left, the assembly starts to grow on each of the even tiles of the input at the same time, copies each of the even inputs to the second layer, then grows to the left and combines the input from each of the odd tiles in parallel. Then, the tiles on top of the odd input tiles contain both the information of the input tiles below, as well as the input tile to its right. This tile can then repeatedly be copied up and left, until it reaches the desired position.
%From the side that the tiles are going to shift to
The runtime of shifting the block will be the distance it moved plus 2. In our algorithm, the length of shifting the input for the next addition is no longer than the length of the input rectangle, so the runtime is at most $O(\sqrt{n})$.

\emph{Runtime For Addition Block.} Figure \ref{fig:Mult_example} shows an example of two-number multiplication, which is performed by stacking several Addition Blocks and Shift Blocks. Recall that in the $O(\log{n})$ Average Case, $\sqrt{n}$ Worst Case Addition (Section {\ref{sec:worstcase}), we achieved a runtime of $O(\sqrt{n})$  for a two number addition in 2D. We use the same principle to build the Addition Block in 3D. Instead of growing to the center of every two rows in 2D, the Addition Block gets its result on top of the input layer. For simplicity, we abandon the $O(\log{n})$ average case addition mechanism here. Therefore, the Addition Block will take $O(\sqrt{n})$ time to finish.

In the first logical step, the assembly repeatedly stacks up the rectangular Copy blocks and Shift blocks to form a vertical column, which is used to deploy the desired numbers to all of the $O(\sqrt[3]{n})$ layers. In this vertical column, each layer contains a Copy block, and between each layer there is a Shift block, so the runtime for this vertical column is $O(\sqrt{n})\times O(\sqrt[3]{n}) + O(\sqrt{n})\times O(\sqrt[3]{n})=O(n^{\frac{5}{6}})$.

In the second step, which occurs in each layer of the assembly, a north-south row starts to grow from the copy block in the vertical column, which is also the northwest corner of each layer, in order to deploy the desired numbers to all of the $O(\sqrt[3]{n})$ columns by horizontally stacking the copy blocks and shift blocks. This process takes $O(\sqrt[3]{n})$ copy blocks and $O(\sqrt[3]{n})$ shift blocks, so the runtime for this north-south row is $O(n^{\frac{5}{6}})$.

In the third logical step, starting from each copy box in the north-south running row in each layer, columns grow eastward in order to deploy numbers while simultaneously adding them. There will be $O(\sqrt[3]{n})$ additions and $O(\sqrt[3]{n})$ shift blocks, so the runtime for each east-west columns is $O(n^{\frac{5}{6}})$

 %So the runtime unitl each column is $O(n^{\frac{5}{6}})$Similar with analyzing the worst case addition algorithm \label{sec:worst_time}. Each addition takes $O(\sqrt{n})$ and then the shifting takes a constant time. The runtime after each column finished $O(\sqrt[3]{n})$ times addition would be $O(n^{\frac{5}{6}})$.

Then, the last block of the first column in each layer will rotate its information to the south and propagate it to the end of the next column, merging the next result and adding them together, then repeating these moving and adding steps until the last result of the layer is obtained. Finally, the result in the bottom layer rotates up and the results from each layer are added together. The runtime of these moving and adding steps is $O(n^{\frac{5}{6}})$. Therefore, the total runtime of the multiplication process is $O(n^{\frac{5}{6}})$.

\begin{figure}[htp]
	\centering
	\includegraphics[width=4.5in]{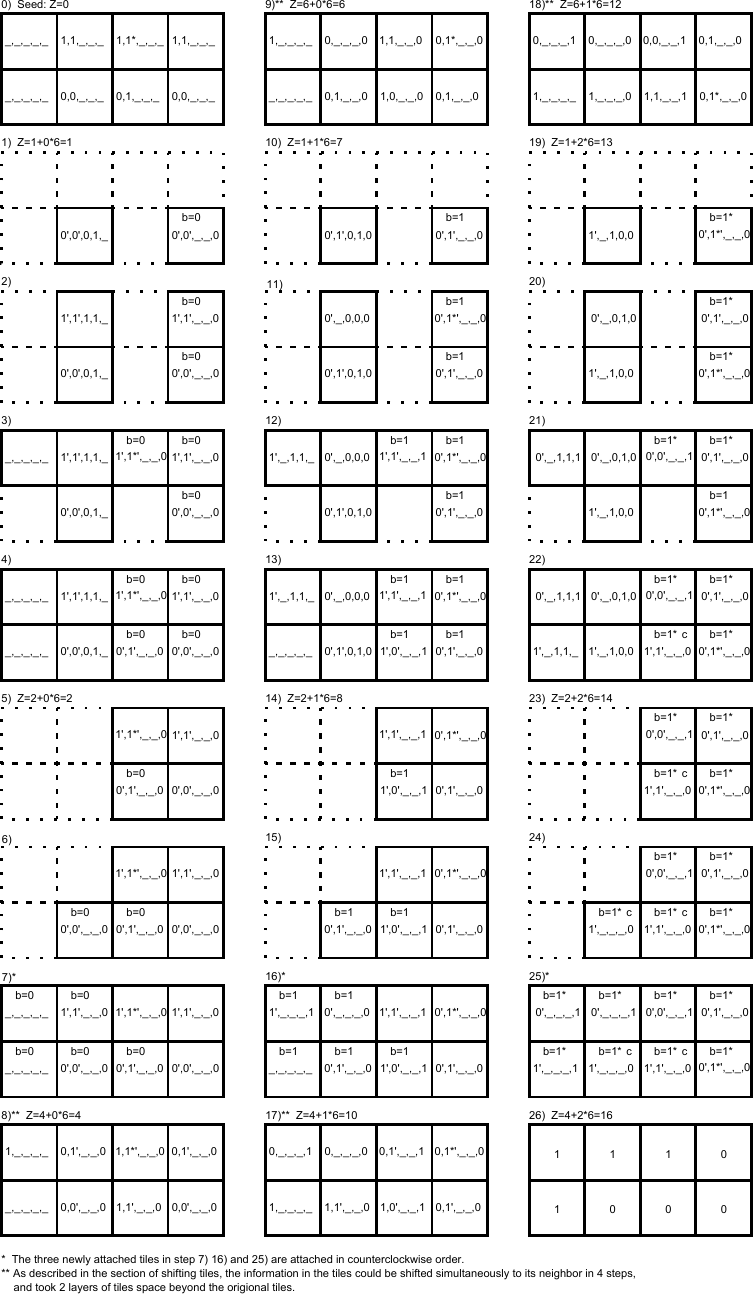}
	\caption{This sequence (Parts 1-26) demonstrates the product of $A$ and $B$, where $A=100110$ and $B=1011*10$. $Z$ denotes the number of layers shown in the figure. The effective bits of $B$ are the bits after and including the bit with the $*$ symbol. So the actual multiplication  is $100110\times110=11100100$. Parts 1-7 show the first addition block. Part 8 shows that A is been shifted one bit ahead after two layers, and Part 9 shows B has been shifted to another direction after another two layers.}
	\label{fig:Mult_example}
\end{figure}

\begin{figure}[htp]
	\centering
	\includegraphics[width=4in]{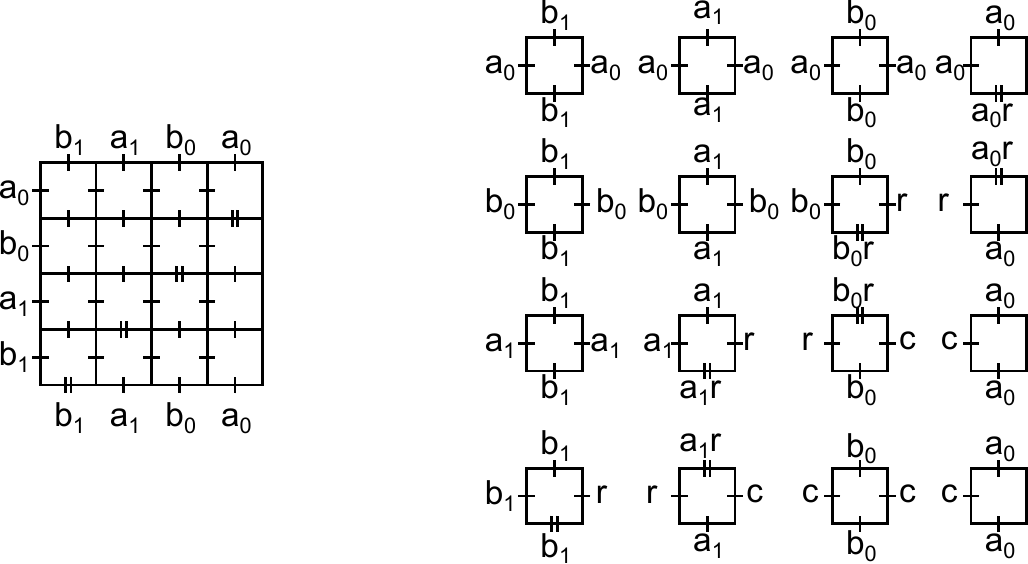}
	\caption{There exists a tile set to copy any input (bottom row) in two directions.}
	\label{fig:Rotate_and_copy}
	
\end{figure}

\begin{figure}[htp]
	\centering
	\subfigure[For any input (bottom row), there exists a tile set to shift the input one tile to the left in 4 parallel steps.]{
		\label{fig:shift_example_1}
		\includegraphics[width=2.4in]{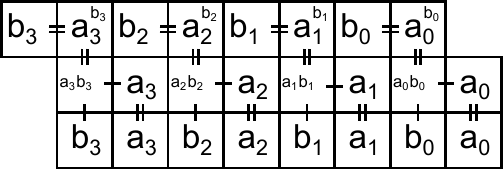}
	}
	~~~
	\subfigure[For any input (bottom row), there exists a tile set to shift the input n tiles to the left in n+2 parallel steps]{
		\label{fig:shift_example_2}
		\includegraphics[width=2.6in]{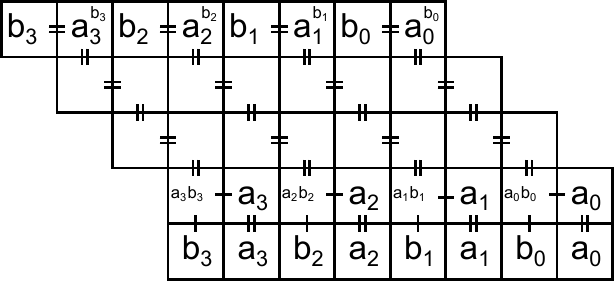}
	}
\caption{The above two figures show how an input can be bit-shifted.}
\end{figure}

\subsection{Correctness.} \label{sec:multiplication_correct}
Note that the multiplication of two n-bit numbers $A$ and $B$ is the summation of a string of numbers $C_1$ ... $C_n$, where $C_k=A*[B_k]*[2^k]$. The order in which the numbers are summed is as follows: numbers are summed starting on the first layer from from north to south, then, with the same order in each layer, from the bottom layer to the top layer. Within each layer, the west-to-east growing columns sum numbers beginning with the northmost column and ending with the southmost column. The head of each column contains $A*[2^i]$ and $b_i$, where $i$ is the number of columns times $\sqrt[3]{n}$, $1<=i<{n}$. Then, starting from the head in each column, shifting and adding are performed. There will be $\sqrt[3]{n}$ shifts performed for both $A$ and $B$. So, when the columns reach the last block, the number at the end of each column will be the $\sqrt[3]{n}$th number after the number in the head of each column, which means, at this point, all the numbers in the list $C$ have been deployed and there are no duplicates. Then, the summation of all blocks is displayed on the output block and is the product of the seed input.

\end{document}